\documentclass[oneside,12pt,letterpaper]{Classes/CUEDthesisPSnPDF}
 
\usepackage{amsthm,amsmath,amssymb,amsbsy,amsfonts,graphicx,appendix,hyperref,upgreek,html,breqn,mathscinet}

\usepackage[hyperpageref]{backref} 

\hypersetup{bookmarksopen=false} 

\usepackage{enumitem} 

\def\addsymbol #1: #2#3{$#1$ \> \parbox{5in}{#2 \dotfill \pageref{#3}}\\} 
\def\newnot#1{\label{#1}} 
\newtheorem{theo}[equation]{Theorem}

\newtheorem{lem}[equation]{Lemma}

\newtheorem{prop}[equation]{Proposition}

\def\m{\mathcal{M}}
\def\k{\mathcal{K}}
\def\g{\boldsymbol{g}}

\def\N{\mathbb{N}}
\def\R{\mathbb{R}}

\def\C{\mathbb{C}}

\ifpdf
    \pdfcatalog { /PageMode (/UseOutlines)
                  /OpenAction (fitbh)  }
\fi

\degree{Doctor of Philosophy}
\degreedate{Yet to be decided}

\usepackage{StyleFiles/watermark}

\begin{document}




\pagenumbering{roman}
\setcounter{page}{0}
\thispagestyle{empty}

\null\vskip0.5in
   \begin{center}
      \hyphenpenalty=10000\LARGE\bfseries {A Geometric Study of Superintegrable Systems}
   \end{center}
   \vfill
   \begin{center}
      \large by\\
      Amelia L. Yzaguirre \\
      ${}$ \\
      \normalsize
      \texttt{ameliay@mathstat.dal.ca}
   \end{center}
   \vfill
   \begin{center}
      Submitted in partial fulfillment of the requirements \\
      for the degree of Master of Science \\[2.5ex]
      at \\[2.5ex]
      Dalhousie University \\
      Halifax, Nova Scotia \\
      August 2012
   \end{center}
   \vskip0.75in
   \begin{center}
      \rmfamily \copyright\ Copyright by Amelia L. Yzaguirre, 2012
   \end{center}

\setcounter{secnumdepth}{3}
\setcounter{tocdepth}{3}

\frontmatter 

\begin{acknowledgements}      

My sincerest gratitude extends to my boyfriend, Antonio, who has provided me with love and light throughout the development of this thesis and beyond.  His support has kept me afloat in the most treacherous of storms, and I thank him immensely for being such a source of inspiration.  I also thank my parents, brothers, and sister for providing encouragement and expressing enthusiasm in my work while I have been so far away from home.  I miss you all beyond words.   Last but not least I want to thank my advisor Dr.~Roman Smirnov for his unbelievable patience and guidance during this study.

This material is based upon work supported by the National Science Foundation Graduate Research Fellowship under Grant No. 1048093. Any opinion, findings, and conclusions or recommendations expressed in this material are my own, and do not necessarily reflect the views of the National Science Foundation.

\vfill
I dedicate the work in this thesis to the passing of my best friend, Kraven.  I think of you every day.
\vfill

\end{acknowledgements}




\begin{abstracts}        

Superintegrable systems are classical and quantum Hamiltonian systems which enjoy much symmetry and structure that permit their solubility via analytic and even, algebraic means. They  include such well-known and important models as the Kepler potential, Calogero-Moser model, and  harmonic oscillator, as well as its integrable perturbations, for example, the Smorodinsky-Winternitz (SW)\newnot{symbol:sw} potential. Normally, the problem of classification of superintegrable systems is approached by considering associated
algebraic, or  geometric structures.   To this end, we invoke the invariant theory of Killing tensors (ITKT).  Through the ITKT, and in particular, the recursive version of the Cartan method of moving frames to derive joint invariants, we are able to intrinsically
characterise and  interpret the arbitrary parameters appearing in the general form of the SW superintegrable potential.  Specifically, we determine, using joint invariants,  that the more general the geometric structure associated with the SW potential is, the fewer arbitrary parameters it admits.   Additionally, we classify the multi-separability of a recently discovered superintegrable system which generalizes the SW potential and is dependent on an additional arbitrary parameter $k$, known as the Tremblay-Turbiner-Winternitz (TTW)\newnot{symbol:ttw} system. We provide a proof that only for the case $k = \pm 1$ does the general TTW system admit orthogonal separation of variables with respect to both Cartesian and polar coordinates.

\end{abstracts}



\tableofcontents

\listoffigures

\chapter*{List of Symbols\hfill} \addcontentsline{toc}{chapter}{List of Symbols}
\markboth{}{List of Symbols}
\begin{tabbing}
$\mathbb{E}^n$~~~~~~~~~~~~\=\parbox{5in}{Euclidean space of dimension $n$\dotfill \pageref{symbol:euclidean-plane}}\\
\addsymbol \mathbb{R}^n: {$n$-dimensional Cartesian space}{symbol:rn}
\addsymbol \m: {Pseudo-Riemannian manifold of dimension $n$ }{symbol:man}
\addsymbol \g: {Metric tensor}{symbol:metric}
\addsymbol T_p(\m) : {Tangent space of $\m$ at a point $p$}{symbol:tangentspace}
\addsymbol T(\m): {Tangent bundle of $\m$}{symbol:tangentbundle}
\addsymbol T^*_p(\m) : {Cotangent space of $\m$ at a point $p$}{symbol:cotangentspace}
\addsymbol T^*(\m): {Cotangent bundle of $\m$}{symbol:cotangentbundle}

\addsymbol T^s_{r}: {Tensor of valence $(s,r)$}{symbol:tensor}
\addsymbol \epsilon_{ij}: {Levi-civita symbol}{symbol:levi-civita}
\addsymbol \delta_{ij}: {Kronecker delta}{symbol:kronecker}
\addsymbol \Gamma^i_{jk}: {Christoffel symbol}{symbol:christoffel}

\addsymbol H: {Hamiltonian}{symbol:ham}
\addsymbol F: {First integral of motion}{symbol:first}
\addsymbol p_i: {Momenta coordinates}{symbol:mom}
\addsymbol q^i: {Position coordinates}{symbol:pos}
\addsymbol L: {Lagrangian function}{symbol:lagrangian}
\addsymbol S: {Action function}{symbol:action}
\addsymbol \boldsymbol{X}_H: {Hamiltonian vector field}{symbol:ham-flow}
\addsymbol P_0: {Canonical Poisson bivector}{symbol:pos-bivec}

\addsymbol \boldsymbol{X_i}: {Translational basis vector}{symbol:transbasis}
\addsymbol \boldsymbol{R}: {Rotational basis vector}{symbol:rotbasis}
\addsymbol \boldsymbol{K}: {Killing tensor of valence two}{symbol:killing}
\addsymbol \mathcal{K}^p(\m): {Vector space of valence $p$ Killing tensors defined on $\m$}{symbol:vspace}
\addsymbol \mathcal{K}^p(\m)/G: {Quotient space}{symbol:quotientspace}

\addsymbol SE(n): {Special Euclidean group}{symbol:speciale}
\addsymbol SO(n): {Special orthogonal group}{symbol:specialo}
\addsymbol \circlearrowright: {Acting on}{symbol:acting}

\addsymbol \dim: {Dimension}{symbol:dim}
\addsymbol \binom{n}{k}: {Number of $k-$subsets of $n$ elements}{symbol:binomial}
\addsymbol d^2(~,~): {Distance function}{symbol:distance}

\addsymbol \bigsqcup: {Disjoint union}{symbol:sqcup}
\addsymbol \otimes: {Tensor product}{symbol:crosstimes}
\addsymbol \odot: {Symmetric tensor product}{symbol:odot}
\addsymbol \wedge: {Wedge product}{symbol:wedge}
\addsymbol \times: {Cartesian product}{symbol:times}

\addsymbol \textnormal{d}: {Exterior derivative operator}{symbol:ext}
\addsymbol \mathcal{L}: {Lie derivative operator}{symbol:lie}
\addsymbol \nabla: {Covariant derivative operator}{symbol:cov}
\addsymbol [~,~]: {Schouten bracket}{symbol:schouten}
\addsymbol \{~,~\}: {Poisson bracket}{symbol:poisson-bracket}

\addsymbol \textnormal{HJ}: {Hamilton-Jacobi}{symbol:HJ}
\addsymbol \textnormal{SW}: {Smorodinsky-Winternitz}{symbol:sw}
\addsymbol \textnormal{TTW}: {Tremblay-Turbiner-Winternitz}{symbol:ttw}
\addsymbol \textnormal{PDE}: {Partial differential equation}{symbol:pde}
\end{tabbing}
 \clearpage

\mainmatter 
\chapter{Introduction} 
\label{intro}

\ifpdf
\graphicspath{{Chapter1/Chapter1Figs/PNG/}{Chapter1/Chapter1Figs/PDF/}{Chapter1/Chapter1Figs/}}
\else
\graphicspath{{Chapter1/Chapter1Figs/EPS/}{Chapter1/Chapter1Figs/}}
\fi

The origins of this thesis can be traced back to 1965 when Winternitz and Fri{\v{s}} used the theory of Lie groups to classify the \emph{separable webs} of the Euclidean plane.  Specifically, they computed the invariants of the second-order symmetries of the Laplace equation 
\[
\Delta V(x,y) = \frac{\partial^2 V}{\partial x^2} + \frac{\partial^2 V}{\partial y^2} = 0, 
\]
defined on the Euclidean plane under the action of the Euclidean group.  As we shall develop in Section \ref{section:separability}, the classification of such separable webs is intimately related to the Hamilton-Jacobi theory, which provides a formulation of \emph{Hamiltonian mechanics} that allows one to conclude the solubility of a classical \emph{Hamiltonian system} (see Section \ref{section:hamiltonian}).  Indeed, by determining the orthogonal coordinate systems that afford separation of variables in the corresponding Hamilton-Jacobi equation for a Hamiltonian system (see Section \ref{section:hamilton-jacobi}), one can ultimately obtain exact solutions, or trajectories, of the Hamiltonian system.  

The work of St{\"a}ckel in 1893 and its  extension by Eisenhart in 1934 were both essential to the development of this geometric approach to studying Hamiltonian systems.  Both of these canonical papers established the essential role of a geometric object that naturally arose when studying orthogonal separation of variables of the Hamilton-Jacobi equation: a Killing tensor of valence two.  Studying the intrinsic properties of Killing tensors corresponding to certain orthogonally separable cases thus provided a new approach to the problem of solving a Hamiltonian system.  In 2002, McLenaghan, Smirnov, and The \cite{mclenaghan2002group} utilised classical invariant theory to develop the properties of such Killing tensors, which marked the initial stages in the development of the invariant theory of Killing tensors.  Since then the theory has been utilised by a significant number of authors who have applied it in a more general setting to study Killing tensors defined in spaces of constant curvature (e.g., see \cite{mclenaghan2002group,deeley2003theory,adlam2005thesis,adlam2008joint,chanu2006geometrical,chana2009rsep,cochran2011thesis,horwood2008thesis,horwood2008invariants,horwood2005invariant,horwood2009minkowski,smirnov2006darboux,smirnov2004covariants,winternitz1965invariants,yue2005thesis} and the relevant references therein).  

In this thesis we use the invariant theory of Killing tensors in the context of solving the classification problem for \emph{superintegrable} Hamiltonian systems defined on the Euclidean plane\footnote{The term \emph{superintegrable} was introduced in 1983 by S. Wojciechowski \cite{wojciechowski1983}.}.  One system of particular interest  is the \emph{Tremblay-Turbiner-Winternitz (TTW) system} introduced in 2009 \cite{TTW2009infinite}, which is a generalisation of another superintegrable system of interest,  the \emph{Smorodinsky-Winternitz (SW) system}.
The unique feature of the TTW system is that it provides  an \emph{infinite} family of solvable and integrable (quantum) systems on the Euclidean plane with respect to a parameter $k$ \cite{TTW2009infinite, TTW2010periodic,TW2010third}.  Furthermore, this family of systems was proven to be both classically and quantum superintegrable  for any rational value of $k$ \cite{kalnins2010families,kalnins2011recurrence,kalnins2012structure,kalnins2011pogo}. 

A Hamiltonian system with $n$ degrees of freedom is \emph{completely integrable} if in addition to its Hamiltonian, $H$, it admits $n-1$ first integrals of motion, $F_i,i=1,\dots,n-1$, that are well-defined functions on the phase space.  These first integrals must be functionally independent, 
\[
\textnormal{d}H \wedge \textnormal{d}F_1 \wedge \dots \wedge \textnormal{d}F_{n-1} \neq 0,
\]
and \emph{in involution} with respect to the Poisson bracket (see Section \ref{conventions})
\[
\{H,F_i\}_{P_0} =  \{F_i, F_j\}_{P_0} = 0, \quad i,j = 1,\dots, n-1.
\]
If such a completely integrable system admits additional $n-1$ functionally independent first integrals of motion, $G_i, i=1, \dots n-1$, namely
\[
\textnormal{d}H \wedge \textnormal{d}F_1 \wedge \dots \wedge \textnormal{d}F_{n-1} \wedge \textnormal{d}G_1 \wedge \dots \wedge \textnormal{d}G_{n-1} \neq 0,
\]\newnot{symbol:ext}
which are in involution 
\[
\{H,G_i\}_{P_0} =  \{G_i, G_j\}_{P_0} = 0, \quad i,j = 1,\dots, n-1.
\]then the system is \emph{maximally superintegrable} (see Section \ref{superintegrable}).  Since such systems have more constants of motion than degrees of freedom it is possible to integrate them through analytic and algebraic methods. For this reason, superintegrable systems are often sought after as starting points for studying much more complicated systems.  Some familiar examples of superintegrable systems (in the realm of physics)  include the harmonic oscillator, the Kepler system, the Calogero-Moser system, the hydrogen atom, and others (see e.g.~\cite{bertrand1873,evans1990superclassical,misch1978,nekho1972,sparano2000,tempesta2004notes}).  

The construction of an infinite family of superintegrable systems, such as the TTW system, is fairly straightforward.  We demonstrate the technique with respect to the TTW system, which begins with the second-order system defined by a Hamiltonian (see Section \ref{section:hamiltonian}), of the form
\begin{equation}
H = p_x^2 + p_y^2 - \omega^2(x^2+y^2) + \frac{\alpha}{x^2} + \frac{\beta}{y^2},
\label{eq:sw}
\end{equation}  
where $p_x,p_y$ denote the momenta coordinates, and $x,y$ denote the position coordinates. This is the SW system, which is recognised as a seond-order perturbation of the two-dimensional harmonic oscillator and further, is an example of a superintegrable Hamiltonian system since it is \emph{multi-separable} with respect to both Cartesian and polar coordinates (see Section \ref{section:separability}).  The property of being superintegrable is what we intend to retain in the construction of a new Hamiltonian system.

To this end, we can transform the Hamiltonian (\ref{eq:sw}) to polar coordinates and replace $\theta$ with $k \theta$, where $k$ is some new arbitrary parameter.  Such a transformation of course preserves the property of the system being separable with respect to polar coordinates, though we now arrive at a new Hamiltonian of the form
\begin{equation}
H = p_r^2 + \tfrac{1}{r}p_r + \frac{1}{r^2}p_{\theta}^2 - \omega^2 r^2 + \frac{1}{r^2}\left(\frac{\alpha}{\cos^2{k\theta}} + \frac{\beta}{\sin^2{k\theta}}\right).
\label{eq:TTW}
\end{equation}
This is the TTW system.  One characterising feature of the TTW system is that much like the system (\ref{eq:sw}), it is still separable with respect to polar coordinates and it still behaves as a perturbation of the harmonic oscillator.  Though this development raises the question of whether the system remains superintegrable for all values of $k$.  In their initial publication, Tremblay et al.~\cite{TTW2009infinite} conjectured and provided strong evidence for the superintegrability of this system for all rational values of $k$, a claim which fueled the activity of researchers who were set out to prove its validity. Indeed, if the system remained superintegrable then it provided an example of a second order system that could generate an infinite number of higher order superintegrable systems, which would work towards the development of a classification theory for superintegrable systems. The claim was proven to be true in the classical case, where Kalnins et al.~\cite{kalnins2011pogo} explicitly solved for the higher order first integrals, appearing as polynomials in the momenta, by using the \emph{structure algebra} of first integrals.  In this context, one must determine the \emph{order} defining when the algebra of first integrals closes with respect to the Poisson bracket, and find the corresponding structure equations for the symmetry algebra (see \cite{kalnins2010tools} for more details).  It was also proven in the quantum case through a \emph{recurrence method} introduced by Kalnins et al.~(see e.g. \cite{kalnins2010tools,kalnins2011recurrence,kalnins2012structure}).  In this thesis we investigate whether there are values of $k$ for which the system admits two first integrals of motion quadratic in the momenta.  Prior to this thesis, such a result has not been shown.

In the Euclidean plane, \emph{maximal superintegrability} requires that a Hamiltonian system  defined by $H$ (see Section \ref{section:hamiltonian}) admits two additional first integrals of motion, $F_i, i=1,2$ functionally independent and in involution, meaning that 
\[
\textnormal{d}H \wedge \textnormal{d}F_1 \wedge \textnormal{d}F_2 \neq 0
\]
and $\{H,F_i\} = \boldsymbol{X}_H(F_i) = 0,$ for $i=1,2$.  Note that a Hamiltonian is always a first integral of motion, a property which can be  physically interpreted as a system demonstrating \emph{conservation of energy}.   For the system given by the Hamiltonian (\ref{eq:TTW}), we also readily obtain a second first integral by exploiting the fact that the system is separable with respect to polar coordinates.  Indeed, since (\ref{eq:TTW}) takes the form of a \emph{natural Hamiltonian}, we can use a result by Liouville  \cite{liouville1846} which yields that such a system will admit a first integral of motion quadratic in the momenta according to 
\begin{equation}
F(q^i,p_i) = \tfrac{1}{2}K^{i j}(q^i)p_i p_j + U(q^i), \quad i,j = 1,2,
\label{eq:first}
\end{equation}
with position coordinates $q^i$, \newnot{symbol:pos} momenta coordinates $p_i$, \newnot{symbol:mom} where $K^{ij}$ is a Killing two-tensor whose normal eigenvectors generate the polar coordinate web (see Section \ref{section:separability}), and $U(q^i)$ is a potential satisfying $dU = \boldsymbol{\hat{K}}\textnormal{d}V$, where $\boldsymbol{\hat{K}} = \boldsymbol{K} \boldsymbol{g}^{-1}$, i.e.~$\boldsymbol{\hat{K}}$ is a $(1,1)-$tensor, which in component form is given by $\hat{K}^i{}_j := K^{i\ell}g_{\ell j}$. \newnot{symbol:killing} And so, in fact, there remains only one first integral of motion to find in order for us to conclude superintegrability of the TTW system.

\section{Summary of Results}
In this thesis we provide a new perspective on the study of joint invariants of Killing tensors, which are objects originally defined by Smirnov and Yue in 2004 \cite{smirnov2004covariants}, and further extended by Adlam et al.~\cite{adlam2008joint}.  In particular, we establish a relationship between the geometric and analytic properties of superintegrable Hamiltonian systems, defined on the Euclidean plane, $\mathbb{E}^2$, \newnot{symbol:euclidean-plane} through the vanishing of joint invariants in the product space $\k^2(\mathbb{E}^2)\times \k^2(\mathbb{E}^2)$, where $\k^2(\mathbb{E}^2)$ denotes the vector space of Killing two tensors defined on the Euclidean plane (see Section \ref{section:ITKT}).  We focus our study on the characterisation of the Smorodinsky-Winternitz potential, where we can readily establish a link between the arbitrary parameters appearing in the potential and the parameters in the associated vector space of Killing tensors.  In particular, we find that the more structure we place on the product space $\k^2(\mathbb{E}^2)\times \k^2(\mathbb{E}^2)$, the more general form a corresponding superintegrable potential is allowed to take, and vice versa.  

  As a generalisation of the Smorodinsky-Winternitz potential, we then focus our attention on the Tremblay-Turbiner-Winternitz potential, where we provide a definitive answer to the following question: 
\begin{center}
\emph{For which values of $k$ is the TTW system multi-separable?}  
\end{center}
Through the geometric description of a superintegrable potential provided by certain types of Killing two-tensors (see Chapter \ref{superintegrable}), we find that only when $k= \pm 1$, which in fact reduces the system to the Smorodinsky-Winternitz potential, does the TTW system retain the property of being multi-separable, in particular, with respect to both Cartesian and polar coordinates in the canonical position.
\section{Overview and Historical Development}
\label{section:history}
We begin with a Hamiltonian system defined (on a pseudo-Riemannian manifold $\m$) \newnot{symbol:man} by a natural Hamiltonian of the form
\begin{equation}
H(q^i,p_i) = \tfrac{1}{2}g^{ij}p_ip_j + V(q^i), \quad i,j = 1,2, 
\label{eq:nat}
\end{equation}\newnot{symbol:ham}
where $(\boldsymbol{q}, \boldsymbol{p})$ denote generalised canonical coordinates.  A first integral of a Hamiltonian system is defined as a smooth function $F = F(\boldsymbol{q},\boldsymbol{p})$ which is constant along the Hamiltonian flow.  The flow can be interpreted as the geometric manifestation of Hamilton's equations in classical mechanics (see Section \ref{section:hamiltonian}), and so one is often interested in the explicit form of a Hamiltonian's first integrals of motion.

From a physical perspective, the first integrals characterise quantities that are conserved throughout the motion of a Hamiltonian system, a notion which evokes some familiar examples of \emph{constants of motion} such as energy, and both angular and linear momentum.  In a way, the more first integrals a system has, the more symmetry and structure the system must have.  If  we can determine $n$ functionally independent first integrals in involution for a system with $n$ degrees of freedom, then the system is said to be \emph{completely integrable}, which means we can explicitly determine the trajectories or orbits of the system (see e.g.~\cite{tempesta2004notes}).  

In general, an additional geometric structure is required to determine the first integrals of a Hamiltonian system.  Some methods which have been successfully employed include the Lax representation method (see e.g.~\cite{ranada2000lax,harnad2004}) and the bi-Hamiltonian approach (see e.g.~\cite{ranada2000bi,tempesta2012} and references therein).  In this thesis, we will make use of \emph{orthogonal separation of variables} in the context of the Hamilton-Jacobi theory to identify \emph{quadratic} first integrals of motion, and vice versa (see Section \ref{section:separability} for details). This approach to determining first integrals has a well-established history which we will discuss in Chapter \ref{chapter:theory}. Furthermore, in Section \ref{section:separability} we will present the natural link between orthogonal separation of variables and Killing two-tensors defined on the Euclidean plane.

\begin{figure}[h!tb]
	\centering
	\begin{tabular}{cc}
		\includegraphics[width=0.45\textwidth]{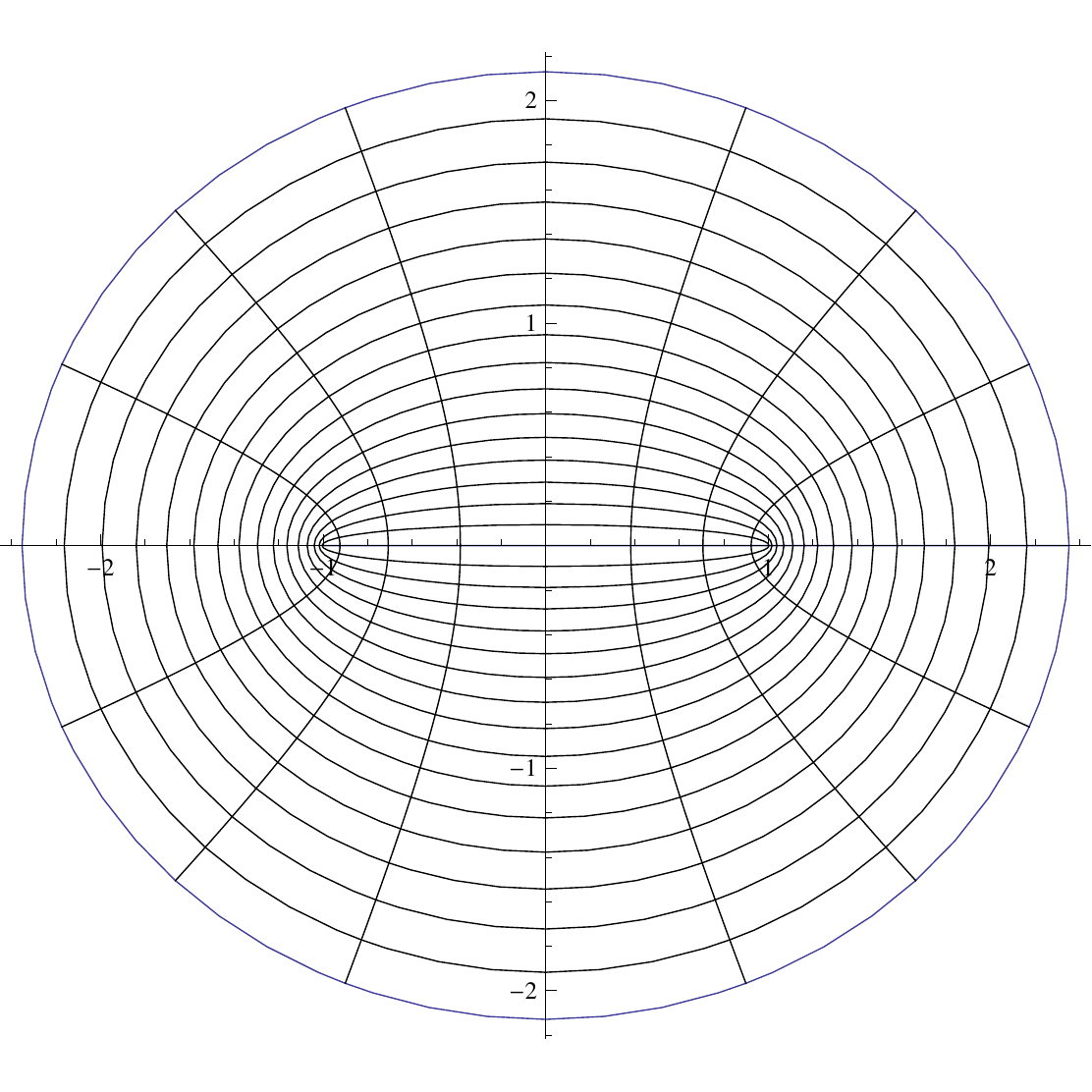}
			& \includegraphics[width=0.45\textwidth]{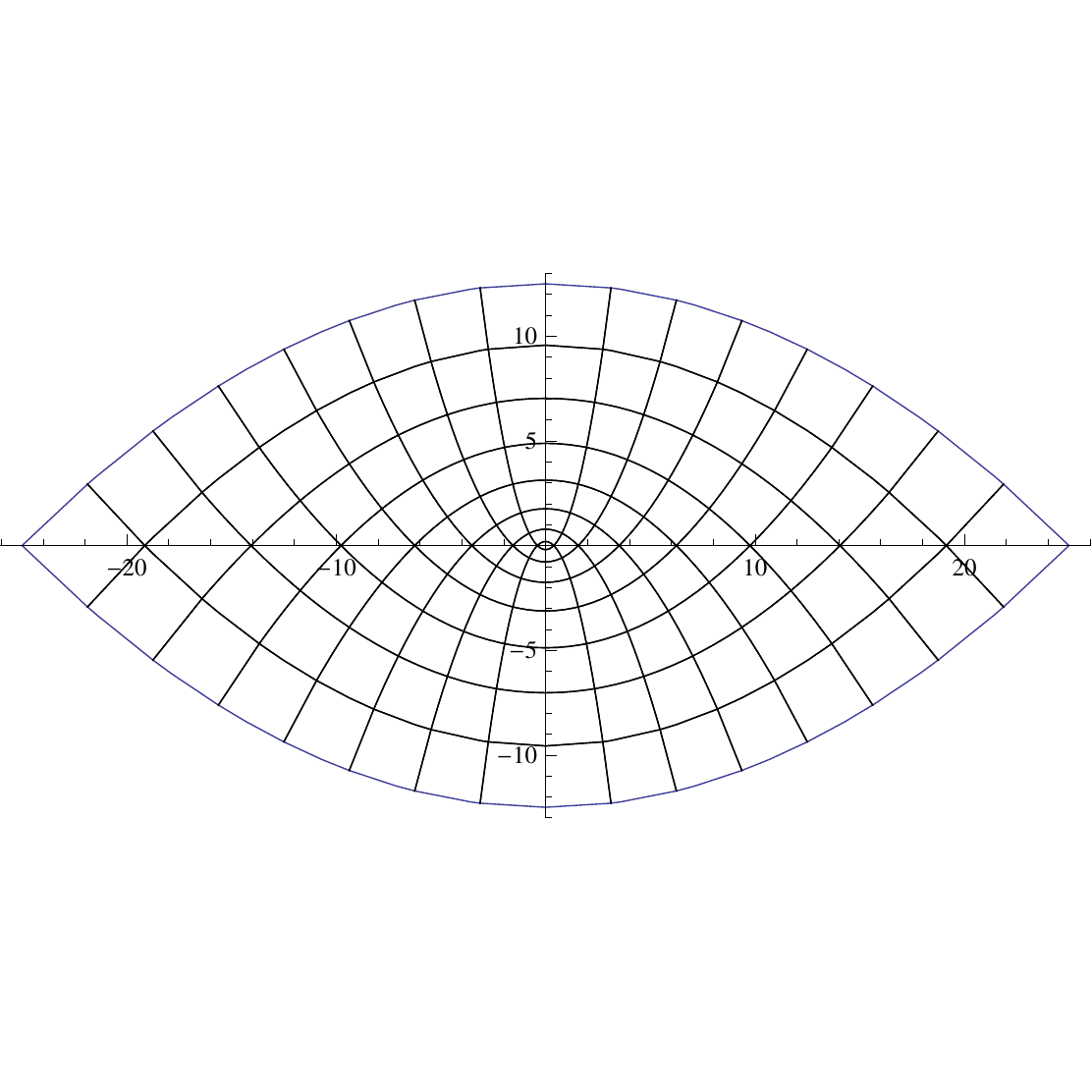} \\
		\includegraphics[width=0.45\textwidth]{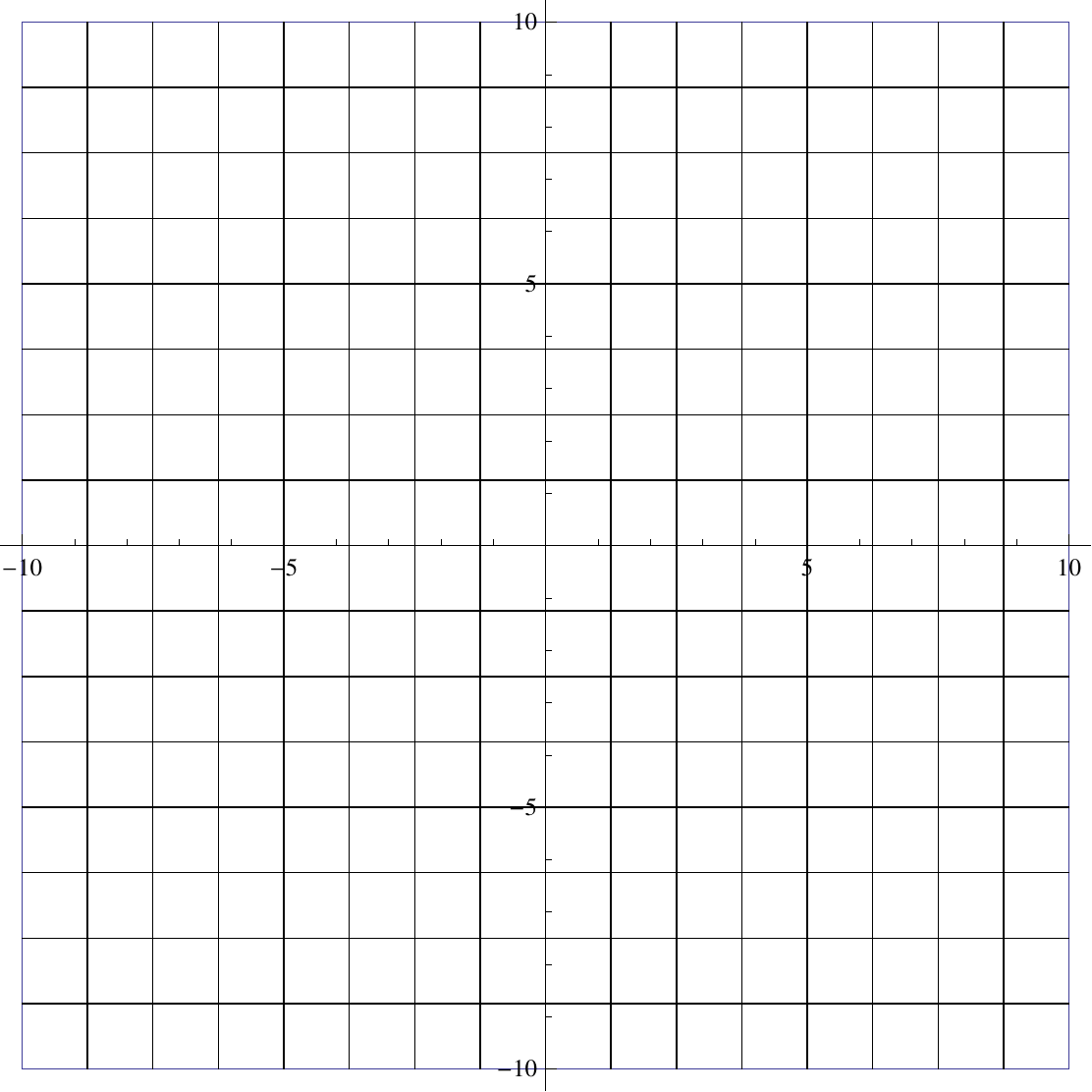}
			& \includegraphics[width=0.45\textwidth]{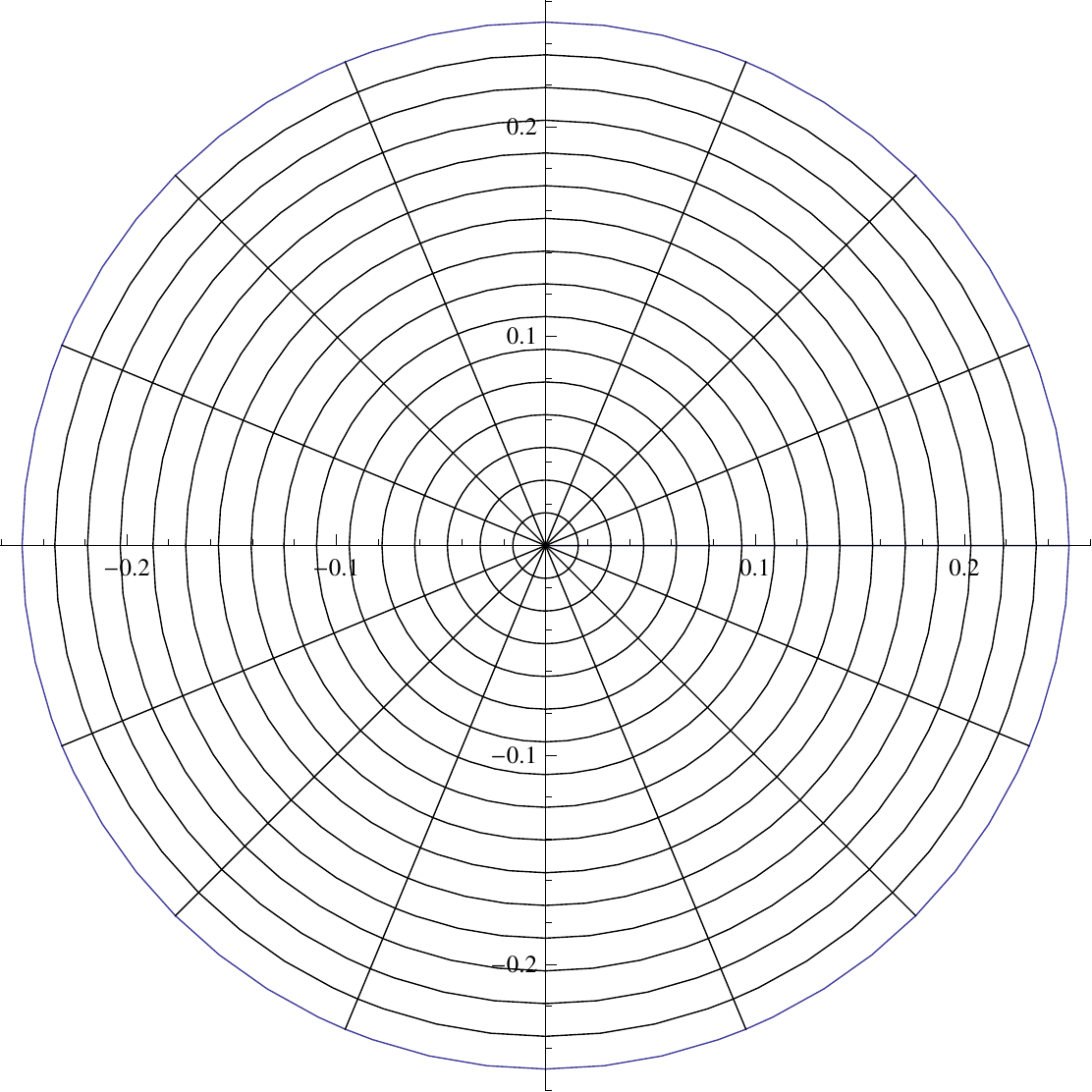}
	\end{tabular}
	\caption{The orthogonal coordinate webs in the Euclidean plane, in the canonical position, arranged according to their number of singular points. Clockwise from top left: elliptic-hyperbolic,  parabolic, polar, and Cartesian.}
\label{canon-webs}
\end{figure}

Formally, the search for superintegrable systems in classical and quantum mechanics (in the context of this thesis) also began in 1965, when Fri{\v{s}} et al.~\cite{fris1965higher}  published a list of superintegrable potentials defined on the Euclidean plane.  Two years later, the study of superintegrable systems defined on three-dimensional Euclidean space, $\mathbb{E}^3$,  was initiated by Winternitz et al.~\cite{winternitz1967symmetry} and Makarov et al.~\cite{makarov1967systematic}.  In 1990, these pioneering works received attention from Evans  \cite{evans1990superclassical}, who was able to extend their results by performing a systematic search for all maximally superintegrable potentials on $\mathbb{E}^3$ that admitted (at most) quadratic first integrals. 

Evans claimed his list of second-order superintegrable systems to be exhaustive up to the equivalence class of linear transformations \cite{evans1990superclassical}.  Though, there was one superintegrable system which did not find itself on Evans' list: the Calogero-Moser model.  If we let $(q^1,q^2,q^3) = (x,y,z)$ denote canonical Cartesian coordinates, then the Calogero-Moser system is defined in $\mathbb{E}^3$ by a natural Hamiltonian, i.e. $$H= \tfrac{1}{2}(p_x^2 + p_y^2 + p_z^2) + V(x,y,z)$$  with potential given by
\begin{equation}
V(x,y,z) = \frac{1}{(x-y)^2} + \frac{1}{(y-z)^2} + \frac{1}{(z-x)^2},
\label{calogero}
\end{equation}
which was well-known to be superintegrable  at the time \cite{wojciechowski1983}. 

The reason such a model was overlooked by Evans' method was due to the fact that he assumed the associated  Killing two-tensors would be in the canonical form \cite{adlam2007orbit}.  This assumption in  \cite{evans1990superclassical} and \cite{fris1965higher} excluded the possibility that a Hamiltonian system could be separable in an orthogonal coordinate system which was in a non-canonical position.  We direct the reader to Figure \ref{non-canon-webs} for examples of such coordinate systems on the Euclidean plane.  Indeed, the Calogero-Moser system is an example of a system that admits five first integrals of motion of the form (\ref{eq:first}), where the associated Killing tensor that comes from solving the \emph{compatibility condition} gives rise to five non-canonical characteristic Killing tensors (see \cite{horwood2005invariant,adlam2007orbit} for more details on the Calogero-Moser system; see Section \ref{section:separability} and  \ref{section:ITKT} for the theory).  

Naturally, this motivates a generalisation of the aforementioned approach which makes use of non-canonical Killing tensors. In 2005, Adlam \cite{adlam2005thesis} began this endeavor and provided a more complete list of such superintegrable potentials defined on the Euclidean plane.  Nevertheless, the case when a potential admits orthogonal separability with respect to elliptic-hyperbolic coordinates in the canonical position, and polar coordinates in a general position has not been completely classified\footnote{The choice of \emph{which} coordinate system is in the canonical position is irrelevant, only that one system is in a general position.}. 

As such, we seek to further develop the growing collection of superintegrable potentials defined on $\mathbb{E}^2$ by providing a new interpretation of their functional form characterised via the invariant theory of Killing tensors.  In particular, we derive a relationship between the joint invariants defined on the product space of Killing two tensors (consisting of Killing tensors compatible with a potential $V$), and the arbitrary parameters of a superintegrable potential which is multi-separable in two orthogonal system of coordinates.  

\begin{figure}[h!tb]
	\centering
	\begin{tabular}{cc}
		\includegraphics[width=0.45\textwidth]{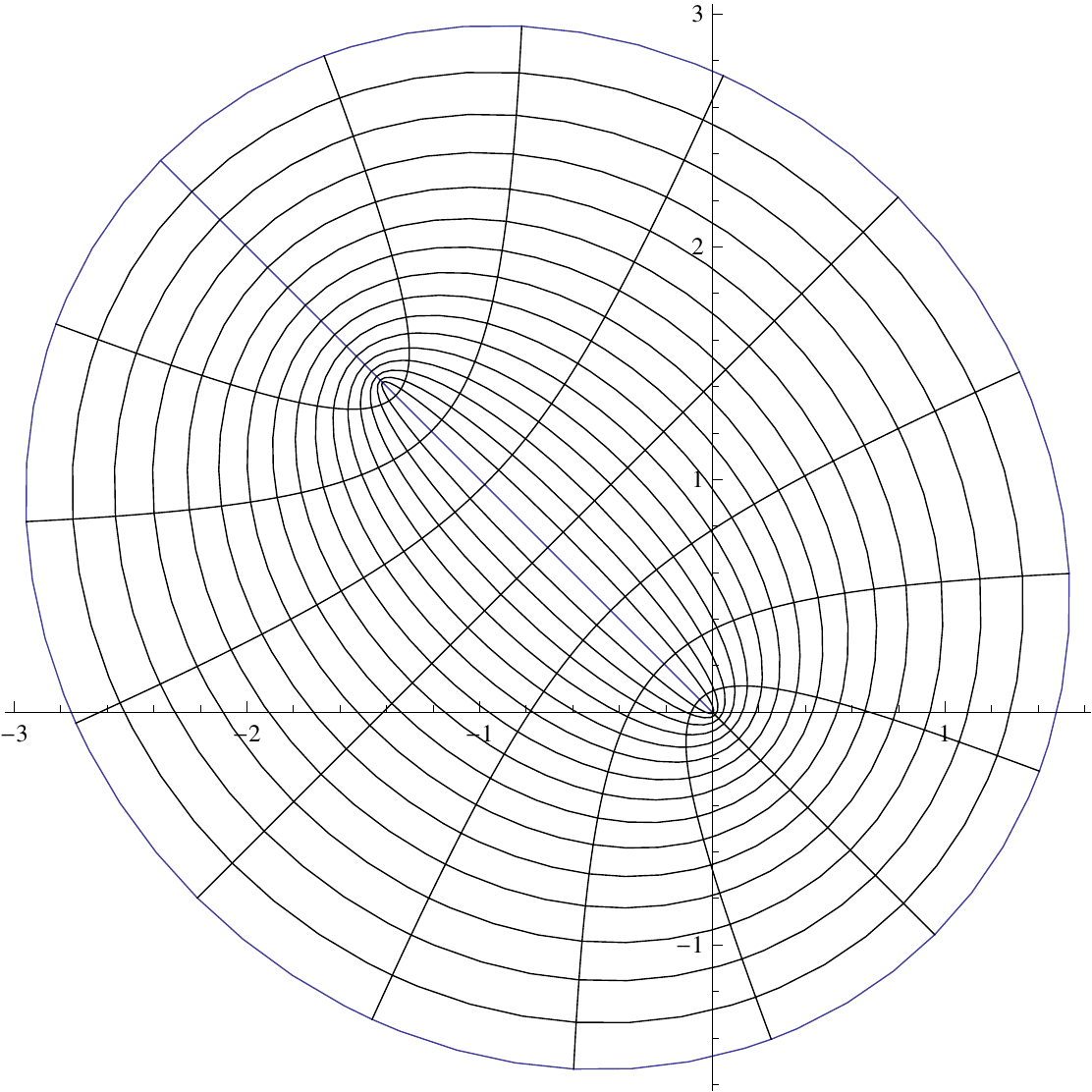}
			& \includegraphics[width=0.45\textwidth]{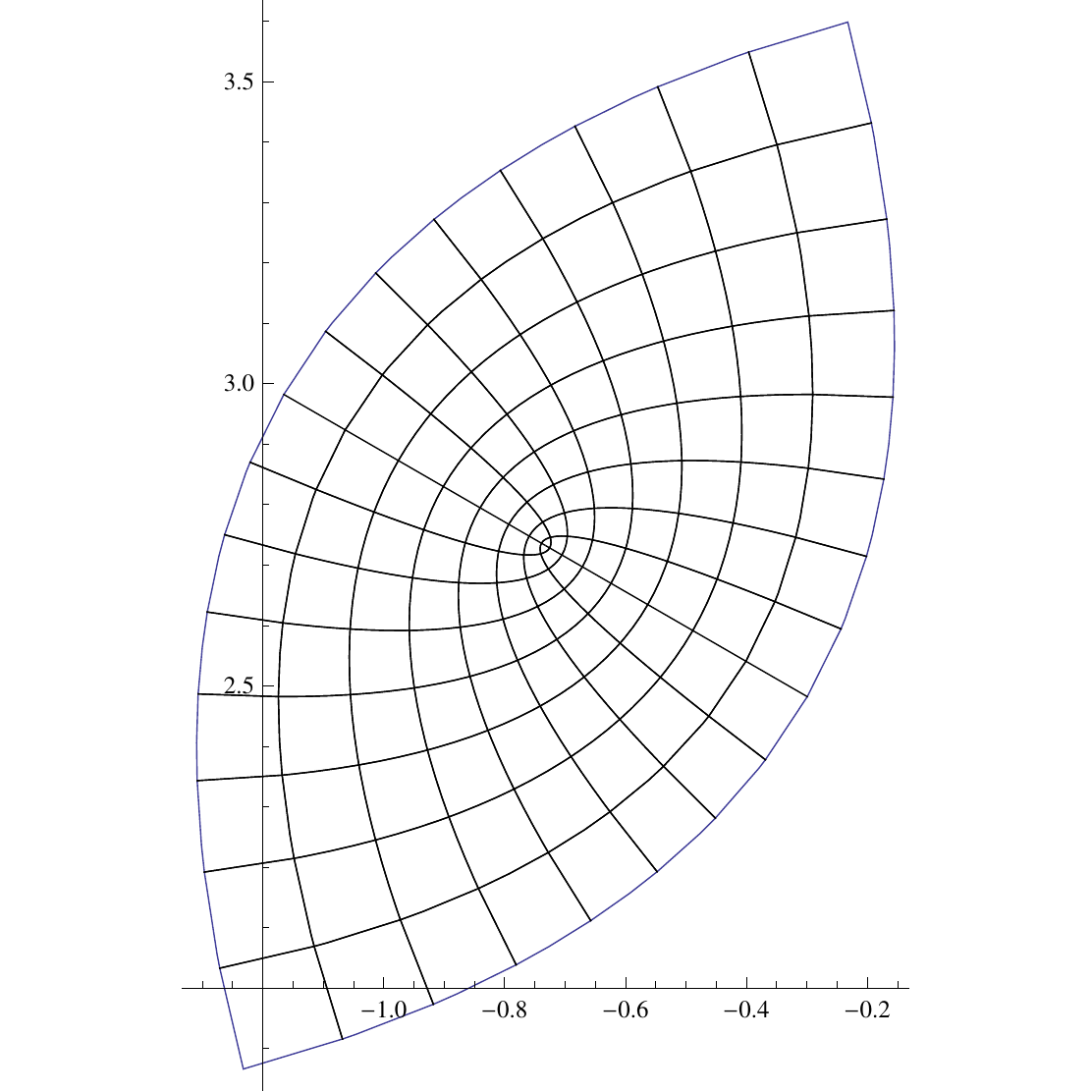} \\
		\includegraphics[width=0.45\textwidth]{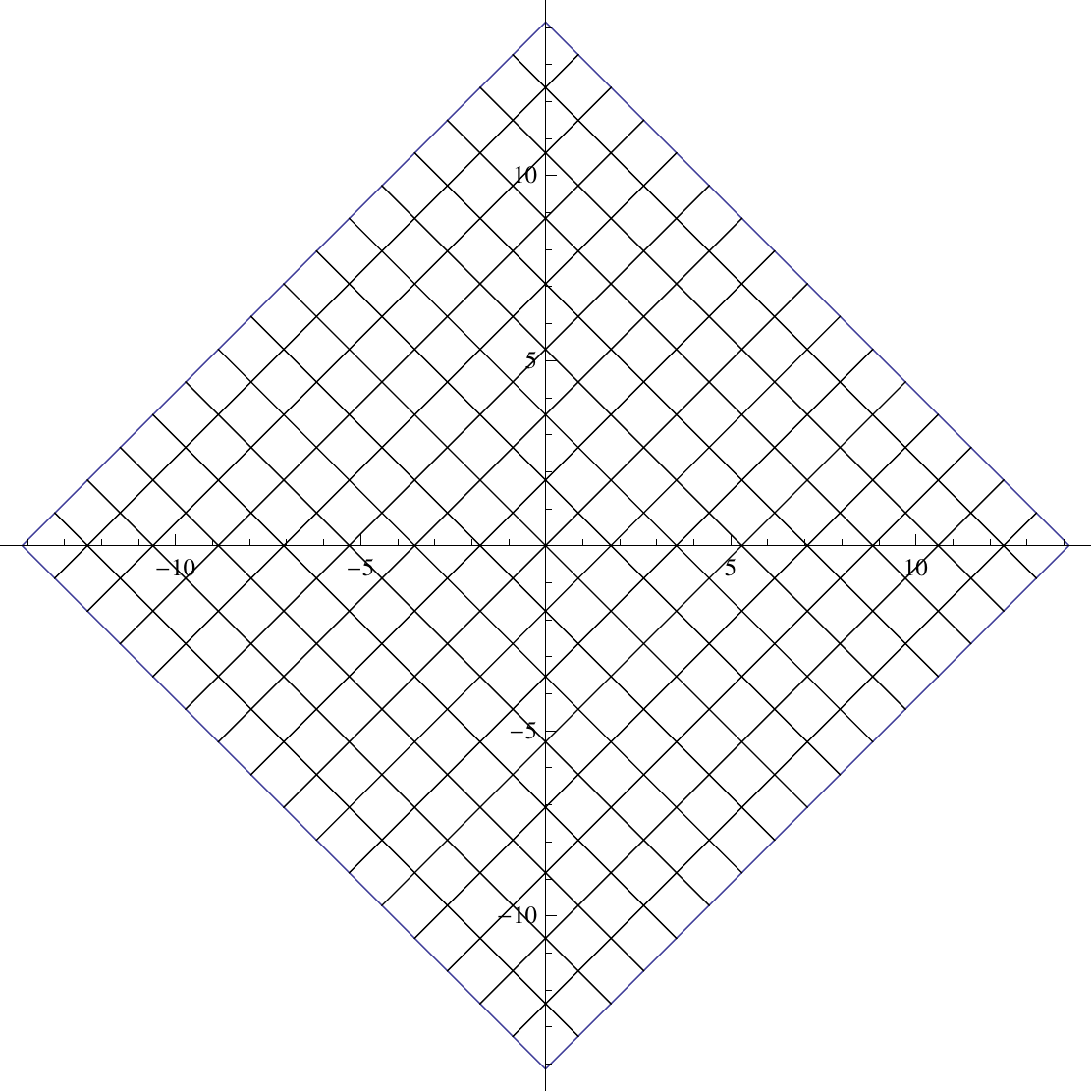}
			& \includegraphics[width=0.45\textwidth]{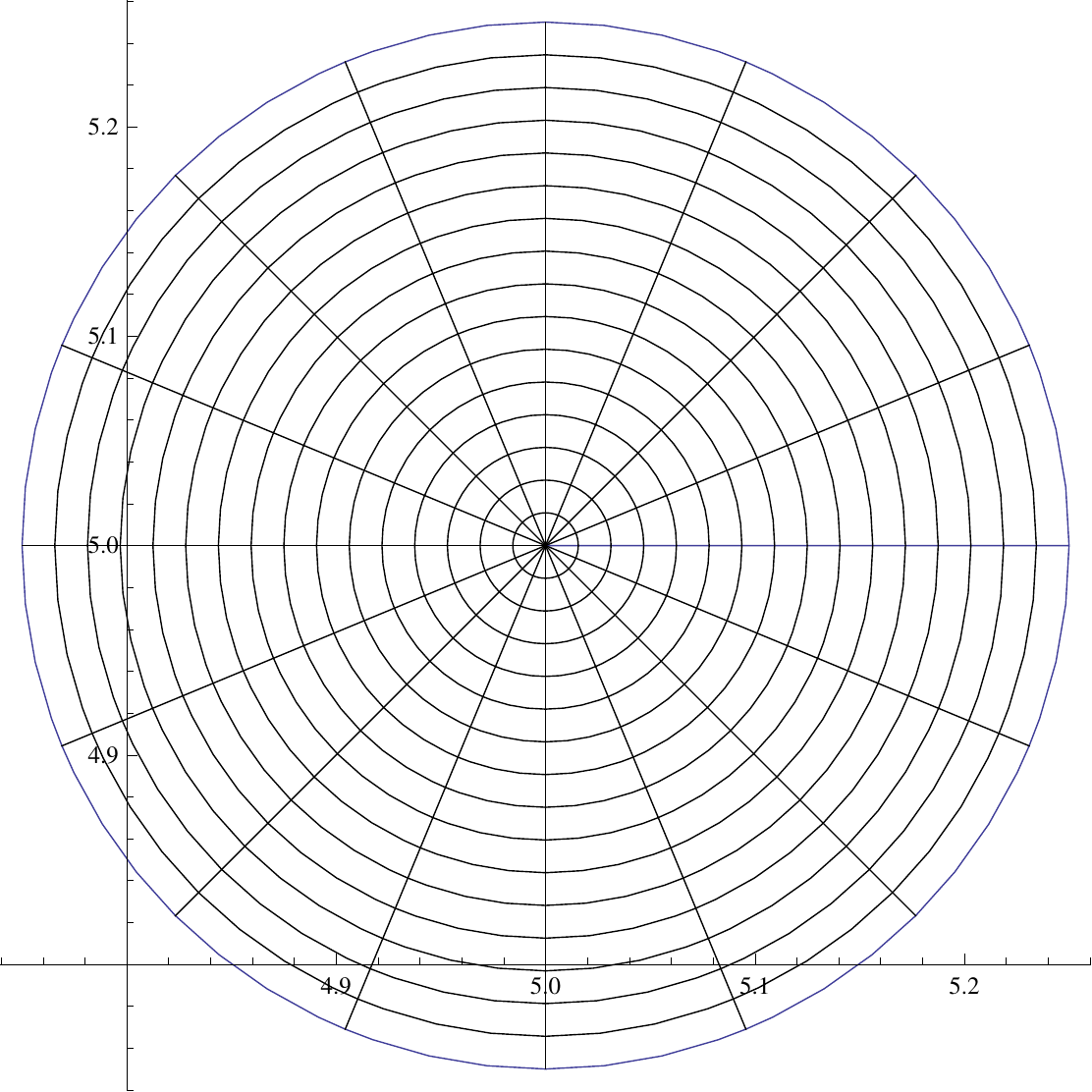}
	\end{tabular}
	\caption{Examples of orthogonal coordinate webs in the Euclidean plane, in some general position.  Clockwise from top left: elliptic-hyperbolic,  parabolic, polar, and Cartesian.}
\label{non-canon-webs}
\end{figure}
\section{Notation and Conventions}
\label{conventions}
In this thesis, we let $\m$ denote an $n$-dimensional pseudo-Riemannian manifold, meaning that $\m$ is a smooth manifold endowed with a non-degenerate covariant symmetric metric tensor $\g$. \newnot{symbol:metric}   The \emph{tangent bundle} of $\m$, $T(\m)$, is defined as the disjoint union of the tangent spaces at all points of $\m$:
\[
T(\m) = \bigsqcup_{p \in \m}{T_p(\m)},
\] \newnot{symbol:tangentspace} \newnot{symbol:tangentbundle} \newnot{symbol:sqcup}
where the tangent space $T_p(\m)$ is the set of all linear maps $X: \m \to \R$ satisfying for all $f,g \in \m$ 
\[
X(fg) = f(p) Xg + g(p)X f,
\]
i.e.~the set of all derivations of $\m$ at $p$.  An element of $T_p(\m)$ is called a tangent vector.  Whence, we also have the \emph{cotangent bundle} of $\m$, defined as the disjoint union of the cotangent spaces at all points of $\m$:
 \[
 T^{*}(\m) = \bigsqcup_{p \in \m}{T_p^{*}(\m)}, 
 \]\newnot{symbol:cotangentspace} \newnot{symbol:cotangentbundle}
where the cotangent space is the dual space of $T_p(\m)$.  Recall that a covector is a real-valued linear functional on a space, i.e.~a linear map $\omega: T_p(\m) \to \R$.  A significant fact which helps to identify $T^{*}_p(\m)$ for $p \in \m$ is then that given a basis for the tangent space, $(E_i)$, the covectors $(\epsilon^i)$ \newnot{symbol:levi-civita} defined by 
\[
\epsilon^i(E_j) = \delta^i{}_j = \left\{
     \begin{array}{c}
     1 \quad \textnormal{if~} i=j,\\
     0 \quad \textnormal{if~} i \neq j
     \end{array}
   \right.
\]   
form a basis for $ T^{*}_p(\m),~p \in \m$.  

 Our results will take place in the Euclidean plane, so in general we will take $\m = \mathbb{E}^2$ where the metric tensor's components then become trivial, e.g.~in Cartesian coordinates the components of the metric tensor are  $g_{ij} = \delta_{ij}$ where $\delta_{ij}$ denotes the familiar Kronecker delta. \newnot{symbol:kronecker}  The remaining content of this section is devoted to the presentation of some specific definitions and theorems relevant to the later sections.
\subsection{Tensors and Index Notation}
Following the presentation in \cite{lee2003intro},  a \emph{multilinear}\footnote{A map is said to be multilinear if it is linear with respect to each of its arguments, e.g.~a multilinear function in two variables is a \emph{bilinear} function.} function 
\[
T: \underbrace{T^*_p(\m) \times \dots \times T^*_p(\m)}_{\textnormal{s copies}}\times \underbrace{T_p(\m) \times \dots \times T_p(\m)}_{\textnormal{r copies}} \to \R
\]
defined at a point $p \in \m$ which maps $r$ covectors and $s$ vectors to a real number is called a \emph{tensor of valence} $(r,s)$, denoted $T^r_s$. \newnot{symbol:tensor}  In the case that $r$ or $s$ is identically zero, we will refer to such a tensor as a covariant $s$-tensor or contravariant $r$-tensor, respectively.  As a few familiar examples, every covector $\omega: T_p(\m) \to \R$ is identified as a covariant $1$-tensor, the dot product on $\R^n$ \newnot{symbol:rn} is a covariant $2$-tensor, as a bilinear form, and the determinant, considered as a function on $n$ vectors, is a covariant $n$-tensor on $\R^n$.  

In Section \ref{section:ttw} we will require all tensor components to be written in polar coordinates, and so we will make use of the usual tensor transformation law to acquire the appropriately defined objects.  For example, the metric tensor defined with respect to a new set of coordinates $(\tilde{q}^1,\dots,\tilde{q}^n)$ is determined with respect to the old coordinates $(q^1,\dots,q^n)$ via
\begin{equation*}
\tilde{g}_{ij} = \Lambda^k{}_i\Lambda^{\ell}{}_j g_{k\ell},
\end{equation*}
where $\tilde{g}_{i j}$ denotes the metric tensor in the new coordinates, and the $\Lambda^k{}_i$ denotes the usual Jacobian matrix defined according to 
\begin{equation}
\Lambda^k{}_i = \frac{\partial q^k}{\partial \tilde{q}^i},
\label{eq:jacobian}
\end{equation}
 relating ``new" and ``old" coordinates. 

As already noted above, we will frequently make use of the Einstein summation convention for tensor indices, by which any repeated upper and lower index implies summation over that index from $1$ to $ n$.  As an example which will prove useful in Section \ref{section:ITKT}, we have the components of a general Killing tensor defined on the Euclidean plane in Cartesian coordinates $(q^1,q^2) = (x,y)$ as
\[
K^{ij} = A^{i j} + 2 \epsilon^{(i}{}_{\ell}B{}^{j)}q^{\ell} + C \epsilon^i{}_m\epsilon^j{}_k q^m q^k,
\]
where 
\[
A^{i j} = \left( \begin{array}{cc}
\beta_1 & \beta_3\\
\beta_3 & \beta_2
\end{array} \right),\quad
B^i = \left(
\begin{array}{c}
	\beta_4 \\
	-\beta_5
\end{array}
	  \right),\quad C = \beta_6,
\]
and $\epsilon^i{}_j = g^{i\ell}\epsilon_{\ell j}$, where $\epsilon_{ij}$ denotes the familiar two-dimensional Levi-Civita symbol\footnote{In particular, this is a pseudotensor which is antisymmetric under the exchange of any two slots:  $\epsilon_{ij} = -\epsilon^{ji}$.  Its values are given by $\epsilon_{12} = \sqrt{|g|},$ where $|g| = \det{g_{ij}}$.\cite{mathworld}}.  Indeed, making use of the Einstein summation convention we can explicitly write out the four components of this general Killing tensor as follows
\begin{align*}
K^{11}&= A^{11} +2 \epsilon^1{}_2 B^1 q^2 + C \epsilon^1{}_2\epsilon^1{}_2 q^2 q^2 \\
			&= \beta_1 + 2\beta_4 y + \beta_6 y^2,\\
K^{12}&= A^{12} + \epsilon^1{}_2 B^2 q^2 + \epsilon^2{}_1 B^1 q^1 + C\epsilon^1{}_2\epsilon^2{}_1 q^1 q^2\\
			&= \beta_3 - \beta_5 y - \beta_4 x - \beta_6 x y,\\
K^{22}&= A^{22} + 2 \epsilon^2{}_1 B^2 q^1 + C\epsilon^2{}_1\epsilon^2{}_1 q^1 q^1\\
			&= \beta_2 + 2 \beta_5 x + \beta_6 x^2.
\end{align*}
We have also made use of round brackets enclosing tensor indices to imply symmetrisation according to 
\[
T_{(ij)} = \tfrac{1}{2}(T_{ij} + T_{ji}),
\]
In an analogous manner, anti-symmetrisation is denoted by the use of square brackets
\[
T_{[i j]} = \tfrac{1}{2}(T_{i j} - T_{ji}).
\]
Lastly we will discuss a few basic tensor operations.  Tensors of the same valence can be added or subtracted by summing corresponding components in the anticipated way.  Multiplication of tensors introduces the tensor product operation $\otimes$ \newnot{symbol:crosstimes} as follows.  If $T \in \mathcal{T}^{r_1}_{s_1}(\m)$ and $S \in \mathcal{T}^{r_2}_{s_2}(\m)$, then their tensor product $T \otimes S$ is a new tensor in $\mathcal{T}^{r_1 + r_2}_{s_1 + s_2}(\m)$ defined by
\begin{align*}
\left(T \otimes S \right) & \left(\omega_1,\dots,\omega_{r_1},\eta_1,\dots,\eta_{r_2};X_1,\dots,X_{s_1},Y_1,\dots,Y_{s_2}\right) \\
	& = T(\omega_1,\dots,\omega_{r_1},X_1,\dots,X_{s_1})S(\eta_1,\dots,\eta_{r_2},Y_1,\dots,Y_{s_2}).
\end{align*}
Of particular importance in Section \ref{section:ITKT} will be the \emph{symmetric tensor product}.  If $T_1 \in \mathcal{T}_{s_1}(\m)$ and $T_2 \in \mathcal{T}_{s_2}(\m),$ then their symmetric tensor product $\odot$ \newnot{symbol:odot} is a symmetric tensor $T_1 \odot T_2$ given by
\begin{align*}
(T_1 \odot T_2)&(X_1,\dots,X_{s_1 + s_2}) \\
&= \frac{1}{(s_1 + s_2)!}\sum_{\sigma \in S_{s_1 + s_2}}{T_1(X_{\sigma{(1)}},\dots,X_{\sigma{(s_1)}})T_2(X_{\sigma{(s_1 + 1)}},\dots,X_{\sigma{(s_1 + s_2)}})},
\end{align*}
where $S_n$ denotes the symmetric group on a set of $n$ elements.  It can be shown that the symmetric tensor product is associative, 
commutative, and bilinear.  An analogous construction is defined by the wedge product $\wedge$ \newnot{symbol:wedge} of tensors, which yields an \emph{antisymmetric} tensor (see e.g.~\cite{lee2003intro}).

We remark that the lowering and raising of tensor indices will be done with the covariant metric $g_{ij}$ and its inverse, the contravariant metric, respectively.  For example, given a vector $\boldsymbol{V}$ with components $V^k$ relative to polar coordinates in the Euclidean plane, its corresponding 1-form $\tilde{V}$ has components $V_i$ given by
\[
V_i = g_{i k}V^k,
\]
where $g_{ik}$ is the metric tensor in polar coordinates, i.e~$\g = \left(\begin{array}{cc}1&0 \\0& r^2\end{array}\right)$. In this case we see that $V_1 = V^1$ and $V_2 = r^2 V^2$.  
\subsection{The Lie Bracket and its Generalisation}
The covariant derivative defines how a vector field changes, i.e.~it is a generalisation of the directional derivative from vector calculus.  It can be defined as a \emph{connection} $\nabla$ \newnot{symbol:cov} on the tangent bundle of a manifold, where a connection provides a prescription for moving vectors from one point to another on the manifold.  The definition of a connection on a manifold does not require that a metric be defined on the manifold, but in the case that a non-degenerate metric \emph{is} defined on a manifold, then there exists a unique torsion-free connection called the \emph{Levi-Civita} connection. This choice of connection is defined by $\nabla \g =0$, i.e.~the choice of connection \emph{with respect to which the covariant derivative of the metric vanishes}.  The components of the Levi-Civita connection are given by the Christoffel symbols, 
\[
\Gamma^i{}_{jk} = \tfrac{1}{2} g^{i \ell}(g_{\ell j,k} + g_{\ell k, j} - g_{j k,\ell}), 
\] \newnot{symbol:christoffel}
so that with this choice the covariant derivative of a covector field $\omega_i$ is given by
\[
\omega_{i;j} = \nabla_j \omega_i = \omega_{i,j} - \Gamma^k{}_{j i }\omega_k,
\]
where a comma followed by an index, e.g.~``,j" denotes partial differentiation with respect to the coordinate $q^j$.  The covariant derivative of a contravariant vector field $V^i$ is given by
\[
V^i{}_{;j} = \nabla_j V^i = V^i{}_{,j} + \Gamma^i{}_{jk} V^k.
\]
This can be generalised for tensor fields of higher valence (see e.g.~\cite{lee2003intro}).  If we work in Cartesian coordinates for $\mathbb{E}^2$ then all components of the Levi-Civita connection vanish, and so all covariant derivatives reduce to simple partial derivatives. 
Accordingly, the covariant derivative gives us a way to ``compare" vectors in neighboring tangent spaces, i.e.~it gives a generalisation of parallel transport (see e.g.~\cite{lee2003intro}).  The \emph{Lie derivative} is then another type of derivative between vector fields that one can define on a manifold.  In particular, it is a canonical way of carrying out such a comparison that does not depend on the existence of a connection. Explicitly, the Lie derivative of a vector field $Y$ along the flow of $X$ is given by
\[
\mathcal{L}_{\boldsymbol{X}}{\boldsymbol{Y}} = X^{\mu}Y^{\nu}{}_{,\mu} - Y^{\mu}X^{\nu}{}_{, \mu} \equiv [X,Y]^{\nu},
\] \newnot{symbol:lie}
where $[~,~]$ denotes the \emph{Lie bracket} which defines a new vector field as the commutator of two vector fields.  It can be shown that the Lie bracket is bilinear, anti-symmetric, and satisfies the Jacobi identity.
 
We now introduce the \emph{Schouten bracket} and discuss a few of its properties.  The Schouten bracket for contravariant tensor fields is a generalisation of the Lie bracket for vector fields. Formally, the Schouten bracket is defined as a real bilinear operator  $[~,~]: T^p(\m) \times T^q(\m) \to T^{p + q-1}$ \newnot{symbol:schouten} whose operation on a pair of contravariant tensor fields $P \in T^p(\m), Q \in T^q(\m)$ is defined by  \cite{schouten1940}
\begin{align}
[P, Q]^{i_{1\dots} i_{p+q-1}} & = \sum_{k=1}^{p}{ P^{(i_{1}{\dots}i_{k-1}\mid\mu \mid i_{k}{\dots}i_{p-1}}} \partial_{\mu}Q^{i_{p}{\dots}i_{p+q-1})} \nonumber \\
&\quad + \sum_{k=1}^{p}{(-1)^{k} P^{[i_{1}{\dots} i_{k-1}\mid \mu \mid i_{k}{\dots}i_{p-1}}}\partial_{\mu}Q^{i_{p}{\dots} i_{p+q-1}]} \nonumber \\
&\quad  - \sum_{\ell=1}^{q}{Q^{(i_{1}{\dots}i_{\ell - 1} \mid \mu \mid i_{\ell}{\dots} i_{q-1}}}\partial_{\mu} P^{i_{q}{\dots}i_{p+q-1})}\nonumber \\
&\quad - \sum_{\ell = 1}^{q}{(-1)^{pq + p + q + \ell}Q^{[i_{1}{\dots}i_{\ell - 1}\mid \mu \mid i_{\ell}{\dots} i_{q-1}}}\partial_{\mu} P^{i_{q}{\dots} i_{p+q-1}]},
\label{eq:schouten-bracket}
\end{align}
where any index located between vertical bars means that the index is \emph{excluded} from any anti-symmetrisation or symmetrisation. If we take $P$ and $Q$ to be tensors of type $(p,0)$ and $(q,0)$, respectively, then the Schouten bracket \cite{schouten1940} of $P$ and $Q$ yields a tensor, $[P,Q]$, of type $(p + q - 1,0)$.  As an example, if we let $P$ and $Q$ be tensors with $p=1$, so that $P$ is now a vector field, and allow $q$ to remain arbitrary, then by definition of the Schouten bracket we have that 
\[
[P, Q]^{k_1\dots k_q} = \left(\mathcal{L}_P Q\right)^{k_1\dots k_q}.
\]
And so indeed, when $q = 1$, we see that we obtain the Lie bracket of the vector fields $P$ and $Q$. 

In general, a \emph{Killing tensor field of valence p} defined on $(\m,\g)$ is a symmetric $(p,0)$ (contravariant) tensor field $\boldsymbol{K}$ satisfying the generalised Killing tensor equation given by
\begin{equation}
[\g,\boldsymbol{K}] = 0,
\label{killme}
\end{equation}
where $[~,~]$ denotes the Schouten bracket.  Since the Schouten bracket is $\R-$bilinear, the set of solutions to the system of overdetermined PDEs \newnot{symbol:pde} given by (\ref{killme}) forms a vector space over $\R$.  We denote the vector space of valence $p$ Killing tensor fields defined on $\m$ by $\k^p(\m)$. \newnot{symbol:vspace}
\subsection{Poisson Manifolds}
Consider the commutator of respective Hamiltonian flows in the classical Poisson bracket (see Section \ref{section:hamiltonian}).  First we define a \emph{Poisson bivector} as a (2,0)-tensor, 
\[
P_0 = P_0^{i j} \frac{\partial}{\partial q^i} \wedge \frac{\partial}{\partial p_j} 
\]\newnot{symbol:pos-bivec}
 defined on $\m$ which satisfies
\[
[P_0, P_0] = 0,
\]
where $[~,~]$ denotes the Schouten bracket.  Any smooth manifold which admits a Poisson bivector, or the general Poisson bracket defined for smooth functions $f,g: \m \to \R$ via
\[
\{f, g\}_{P_0} = P^{ij}_0 \frac{\partial f}{\partial q^i} \frac{\partial g}{\partial p_j}, 
\]\newnot{symbol:poisson-bracket}
is called a \emph{Poisson manifold}.  If we identify $P_0$ as the \emph{canonical} Poisson bivector, then the previous equation becomes
\[
\{f, g\}_{P_0} = \sum_{i}{\frac{\partial f}{\partial p_i} \frac{\partial g}{\partial q^i} - \frac{\partial g}{\partial p_i}\frac{\partial f}{\partial q^i}}.
\]
Following \cite{lee2003intro}, we can construct a smooth vector field on $\m$ for some smooth function $H: T^*(\m) \to \R$ via
\[
\boldsymbol{X}_H = [P_0, H], 
\]\newnot{symbol:ham-flow}
where the flow of this Hamiltonian vector field is called its \emph{Hamiltonian flow}.  With respect to (local) coordinates $x^\alpha = (q^i,p_i)$, the coordinates of the flow are then given as
\[
X_H^i = P^{i \alpha}\frac{\partial H}{\partial x^\alpha}.
\]
A Poisson manifold endowed with a smooth function $H$ defines a \emph{Hamiltonian system}, where the function $H$ is aptly denoted as the \emph{Hamiltonian} of the system.  The integral curves of the Hamiltonian flow are then called the \emph{orbits} or \emph{trajectories} of the system.

The Lie bracket and Poisson bracket can be shown to admit the following relationship:
\begin{prop}[cf. \cite{lee2003intro}]
Let $F, G: \m \to \R$ be smooth functions, where $\m$ is a Poisson manifold.  Denote their respective Hamiltonian flows by $\boldsymbol{X}_F$ and $\boldsymbol{X}_G$.  We then have that
\[
[\boldsymbol{X}_F, \boldsymbol{X}_G] = - \boldsymbol{X}_{\{F,G\}},
\]
where $[~,~]$ is the Lie bracket from the previous section, and $\{~,~\}$ is the Poisson bracket.
\end{prop}
A proof is provided in \cite{lee2003intro}.




\chapter{Theoretical Overview}
\label{chapter:theory}

\ifpdf
\graphicspath{{Chapter2/Chapter2Figs/PNG/}{Chapter2/Chapter2Figs/PDF/}{Chapter2/Chapter2Figs/}}
\else
\graphicspath{{Chapter2/Chapter2Figs/EPS/}{Chapter2/Chapter2Figs/}}
\fi

 In this chapter we provide the reader with the mathematical framework on which this study is built.  This begins with a brief review of the Hamiltonian formulation of classical mechanics, which is the context of the problems we are interested in studying.  We define and discuss what it means to \emph{solve} a Hamiltonian system, which motivates our review of the Hamilton-Jacobi (HJ) \newnot{symbol:HJ} theory of separability as a powerful method for solving particular Hamiltonian systems.  The HJ theory will allow us to explicitly define what is meant in saying that a system is \emph{separable} with respect to a system of coordinates on a smooth manifold, where we focus on results relevant to the Euclidean plane.  In this context, we present theorems by Jacobi, Eisenhart, Liouville, and Benenti which establish the link between \emph{orthogonal separation of variables} and \emph{Killing two-tensors} defined on $\mathbb{E}^2$.  This will motivate the final section, where we present the reader with a review of the invariant theory of Killing tensors.  One goal in this final section is to introduce the method of moving frames by using a \emph{recursive} version of it to derive the fundamental invariants of the vector space of Killing tensors of valence two under the action of the isometry group on $\mathbb{E}^2$, $SE(2)$.
\section{Hamiltonian Mechanics}
\label{section:hamiltonian}
In the \emph{Lagrangian formulation} of classical mechanics the state of a mechanical system defined on an $n$-dimensional manifold $\m$ is completely determined once we specify its generalised position and velocity coordinates.  In the most general formulation, this is accomplished by defining a function, $L(q^i,\dot{q}^i)$, \newnot{symbol:lagrangian} of the position coordinates $q^i$ and their velocities $\dot{q}^i$, where the dot denotes a derivative with respect to time, denoted by $t$.  This function $L$ is called the \emph{Lagrangian} of the system and essentially defines a function from the tangent bundle on a manifold $\m$, denoted by $T(\m)$, to $\R$.  

The evolution of a system can then be determined by employing \emph{Hamilton's principle} or the \emph{principle of least action}.  We take the functional (see e.g.~\cite{landau1960})
\begin{equation}
S[q^i] = \int_{t_i}^{t_f}{L(q^i,\dot{q}^i) dt},
\label{eq:action}
\end{equation} \newnot{symbol:action}
to define the \emph{action} of the system, and from it compute the functional derivative of $S$ with respect to each of the $q^i$ coordinates.  By Hamilton's principle we have that the coordinates evolve in such a way (as functions of $t$) for which $S$ remains \emph{stationary} with respect to variations in $q^i$ that leave the initial and final time values unaffected.  In other words, the action should satisfy $\updelta S =0$, where $\updelta$ represents the variation with respect to $q^i$.  Effecting this computation leads to the well-known \emph{Euler-Lagrange equations}, or the \emph{equations of motion}, for a given system, namely
\begin{equation}
\frac{\partial L}{\partial q^i} - \frac{d}{dt}\left(\frac{\partial L}{\partial \dot{q}^i}\right) = 0.
\label{eq:euler-lagrange}
\end{equation} 

In the \emph{Hamiltonian formulation} of classical mechanics, we consider a slightly different set of generalised coordinates.  In particular, we derive the canonical momenta coordinates $p_i$ from the Lagrangian $L$ via 
\begin{equation}
p_i = \frac{\partial L}{\partial \dot{q}^i}.
\label{eq:canonical-momenta}
\end{equation}
In this way, the canonical momenta coordinates can be taken as defining a covector field in the cotangent bundle $T^{*}(\m)$.  The relationship between the Lagrangian and Hamiltonian formulations is then given by the \emph{Legendre transformation} between their respective coordinates, namely a smooth function $H: T^{*}(\m) \to \R$ defined by
\begin{equation}
H(p_i,q^i) = p_i \dot{q}^i - L(q^i,\dot{q}^i),
\label{eq:gen-ham}
\end{equation}
where, for the purpose of this thesis, we have assumed that $H$ has no explicit time dependence.  This function $H$ is called the \emph{Hamiltonian} of the system.  Through equation (\ref{eq:canonical-momenta}) we see that the time derivatives appearing in (\ref{eq:gen-ham}) can be taken as functions of the generalised position and momenta coordinates.  Indeed, substituting (\ref{eq:canonical-momenta}) into the Euler-Lagrange equations (\ref{eq:euler-lagrange}) and taking the partial derivative of (\ref{eq:gen-ham}) with respect to $q^i$ yields the evolution of our momenta coordinates,
\begin{equation}
\dot{p}_i = -\frac{\partial H}{\partial q^i}.
\label{eq:canonical-ham-mom}
\end{equation}
Similarly, we can take a partial derivative of (\ref{eq:gen-ham}) with respect to $p_i$ to obtain the evolution equation of our position coordinates, namely
\begin{equation}
\dot{q}^i = \frac{\partial H}{\partial p_i}.
\label{eq:canonical-ham-pos}
\end{equation}
Together, the equations (\ref{eq:canonical-ham-mom}) and (\ref{eq:canonical-ham-pos}) are called \emph{Hamilton's canonical equations}.  An important observation to make is that Hamilton's equations form a set of $2n$ \emph{first-order} differential equations for the $2n$ unknown functions $q^i(t)$ and $p_i(t)$ of the system.  This system of first-order differential equations thus replaces the $n$ \emph{second-order} differential equations (\ref{eq:euler-lagrange}) admitted by the Lagrangian treatment.  

Hamilton's equations define a vector field on the cotangent bundle $T^{*}(\m)$. The $2n$-dimensional \emph{phase space} of a Hamiltonian, admitted by generalised position coordinates $q^i$ and generalised momenta coordinates $p_i$, is an example of a \emph{symplectic manifold}, which is defined as a smooth manifold with a smooth, non-degenerate, closed $2$-form.   In particular, we can construct such a 2-form via
\begin{equation}
\omega = \sum_{i = 1}^{n}{\textnormal{d}p_i \wedge \textnormal{d}q^i},
\label{eq:canon-form}
\end{equation}
where $\textnormal{d}$ denotes the exterior derivative and $\wedge$ denotes the exterior, or wedge, product. This is the \emph{canonical symplectic 2-form}.  A fundamental result in the theory of symplectic structures, due to Darboux (see e.g.~\cite{lee2003intro}), states that for any $2n$-dimensional symplectic manifold with symplectic form $\omega$
there exist local canonical coordinates such that $\omega$ assumes the canonical form (\ref{eq:canon-form}).  Such coordinates are called \emph{Darboux coordinates}, or \emph{canonical coordinates}.  Further, the canonical symplectic 2-form $\omega$ induces a pairing between the tangent and cotangent bundle, namely a bivector field on the manifold called the canonical \emph{Poisson bivector}, which is  obtained by taking the inverse of the canonical symplectic form
\begin{equation}
\omega^{-1} =\sum_{i =1}^{n}{ \frac{\partial}{\partial q^i} \wedge \frac{\partial }{\partial p_i} \equiv P_0.}
\label{eq:poisson-vec}
\end{equation}
Using the Poisson bivector, we can define the Poisson bracket for two smooth functions $f$ and $g$ defined on the cotangent bundle, by
\begin{align*}
\{f, g\}_{P_0} &= P^{ij}_0 \frac{\partial f}{\partial q^i} \frac{\partial g}{\partial p_j}, \\
				 &= \sum_{i}{\frac{\partial f}{\partial p_i} \frac{\partial g}{\partial q^i} - \frac{\partial g}{\partial p_i}\frac{\partial f}{\partial q^i}}. 
\end{align*}
Given a symplectic manifold $\m$ with a canonical Poisson bivector $P_0$ (\ref{eq:poisson-vec}), if we have a smooth function $H:T^*(\m) \to \R$, then
\begin{align}
\boldsymbol{X}_H &= [P_0, H] \nonumber \\
				&=\sum_i{\frac{\partial H}{\partial p_i} \frac{\partial}{\partial q^i} - \frac{\partial H}{\partial q^i}\frac{\partial}{\partial p_i}}, \quad \textnormal{(in local coordinates)}
\label{eq:hamiltonian-flow}
\end{align}
defines the \emph{Hamiltonian vector field}, where $[~,~]$ denotes the Schouten bracket defined in Section \ref{conventions}.  A symplectic manifold $\m$ with a Hamiltonian vector field defined with respect to a Poisson bivector, and smooth function $H:T^{*}(\m) \to \R$, namely $\{\m, \boldsymbol{X}_H, H\}$, comprise a \emph{classical Hamiltonian system}.  This development also yields the condition for a smooth function $F:T^*(\m) \to \R$ to be an \emph{integral of motion}.  Indeed, to remain constant during the evolution  of the system the Poisson bracket of a time-independent $F$ with respect to the Hamiltonian $H$ must vanish,
\begin{equation}
\left\{H,F\right\}_{P_0} = \boldsymbol{X}_H(F) = 0.
\label{eq:yo}
\end{equation}
The proof of this statement is straightforward upon considering the total time derivative of $F = F(q^i,p_i;t)$,
\[
\frac{d F}{dt} = \frac{\partial F}{\partial t} + \sum_{i}{\frac{\partial F}{\partial q^i} \dot{q}^i + \frac{\partial F}{\partial p_i}\dot{p}_i}.
\]
By Hamilton's equations (\ref{eq:canonical-ham-pos}) and (\ref{eq:canonical-ham-mom}) this takes the form
\begin{equation}
\frac{d F}{dt} = \frac{\partial F}{\partial t} + \{H,F\}_{P_0},
\label{eq:yoyo}
\end{equation}
where 
\[
\{H, F\}_{P_0} = \sum_{i}{\frac{\partial H}{\partial p_i} \frac{\partial F}{\partial q^i} - \frac{\partial F}{\partial p_i}\frac{\partial H}{\partial q^i}},
\]
is the Poisson bracket of $H$ and $F$ with respect to a canonical Poisson bivector $P_0$. In the case that $F$ is time-independent, (\ref{eq:yo}) follows immediately from (\ref{eq:yoyo}). 

Hamiltonian systems play an essential role in classical and quantum mechanics and the theory of differential equations.  They can be used to describe simple Newtonian systems such as the motion of a pendulum, or a heavy symmetrical top with a fixed lower point, and other dynamical systems such as planetary orbits arising in celestial mechanics.  In this thesis, we will focus on a class of Hamiltonian systems whose Hamiltonian functions take the form of a \emph{natural Hamiltonian}.  To establish this distinction we will first note what it means in the case of the Lagrangian formalism, as this readily motivates the interpretation in the Hamiltonian formalism. 

In the context of the Lagrangian formalism, a \emph{natural Lagrangian} is one which takes the form
\begin{equation}
L(q^i,\dot{q}^i) = \tfrac{1}{2} g_{i j} \dot{q}^i\dot{q}^j - V(q^i),
\label{eq:natural-lagrangian}
\end{equation}
where $g_{i j}$ denotes the covariant components of the associated metric tensor $\g$, which is a function of the coordinates $q^i$, and $V(q^i)$ represents the \emph{potential energy}, which characterises the interaction between the particles in a system.   We remark that in physical terms the Lagrangian is defined as the difference between kinetic and potential energy of a mechanical system, which explains the aforementioned terminology.  With this identification, the first term of (\ref{eq:natural-lagrangian}) can be interpreted as the kinetic energy of the system.  For the natural Lagrangian (\ref{eq:natural-lagrangian}) the Euler-Lagrange equations (\ref{eq:euler-lagrange}) written in covariant form appear as
\begin{equation}
\ddot{q}^i + \Gamma^i_{j k} \dot{q}^j \dot{q}^k = g^{i j} V_{, j},
\label{eq:geo}
\end{equation}
where $\Gamma^i_{j k}$ are the components of the Levi-Civita connection discussed in Section \ref{conventions}.  In the case that we have a vanishing potential, i.e.~a system free of particle interactions, (\ref{eq:geo}) takes the form of the \emph{geodesic equation}, which is very familiar: the geodesic equation models the evolution of a system with $n$ degrees of freedom and no external forces.  Through the Legendre transformation (\ref{eq:gen-ham}) we then obtain the equivalent \emph{natural Hamiltonian}, which assumes the form
\begin{equation}
H(q^i,p_i) = \tfrac{1}{2} g^{i j} p_i p_j + V(q^i),
\label{eq:natural-hamiltonian}
\end{equation} 
which our discussion above makes it clear that the Hamiltonian can be interpreted as the \emph{total energy} of our system, i.e.~the sum of kinetic and potential energy, respectively.

Neither the Lagrangian nor the Hamiltonian formulation present a simple system of equations to solve.  Even though the Hamiltonian formulation reduces the Hamiltonian system to a set of first-order equations, the fact remains that we still must solve a coupled set of non-linear ordinary differential equations.  This prompts the need for methods by which we can complete the complicated task of integrating the equations of motion. In the case of systems with two degrees of freedom, we can focus on the method of separation of variables in the context of the Hamilton-Jacobi theory, which will allow us to link integrability of the equations of motion with the existence of an associated valence two Killing tensor.  This link is in part a result of the Arnold-Liouville theorem \cite{arnoldbook} which states that a Hamiltonian system with $n$ degrees of freedom will be integrable by quadratures if it admits $n$ functionally independent first integrals of motion in involution with the Hamiltonian.  
\section{Hamilton-Jacobi Theory}
\label{section:hamilton-jacobi}
The strategy behind the Hamilton-Jacobi theory is to rewrite Hamilton's equations, (\ref{eq:canonical-ham-mom}) and (\ref{eq:canonical-ham-pos}) in a revealing form through a choice of a new system of coordinates (see e.g.~\cite{landau1960,arnoldbook}).  In particular, we will require a \emph{canonical transformation} of the coordinates, i.e.~a transformation that leaves Hamilton's equations invariant.  Indeed, a canonical transformation would allow us to  obtain an equivalent formulation of classical mechanics just as in the previous section.  Explicitly, if we change to a new system of coordinates $(Q^i,P_i)$ through a canonical transformation then the new Hamiltonian $K$ must satisfy
\begin{equation*}
\dot{Q}^i = \frac{\partial  K}{\partial P_i}, \quad \dot{P}_i = -\frac{\partial K}{\partial Q^i}.
\end{equation*}
To this end we can consider the action (\ref{eq:action}) as a function of the coordinates, i.e. independent of time.  We will then study the change in the action $S$, defined in the previous section, from one path to a neighbouring path, which is readily calculated as
\begin{equation*}
\updelta S = \left[\frac{\partial L}{\partial \dot{q}} \updelta q \right]^{t_f}_{t_i} + \int_{t_i}^{t_f}{\left(\frac{\partial L}{\partial q} - \frac{d}{dt}\frac{\partial L}{\partial \dot{q}}\right) \updelta q~dt}.
\end{equation*}

Since the path of the system satisfies the Euler-Lagrange equations (\ref{eq:euler-lagrange}), then the integral in $\updelta S$ vanishes.  If we consider the canonical momenta coordinates (\ref{eq:canonical-momenta}) then we see that the above equation gives us, for an arbitrary number of degrees of freedom,
\begin{equation}
\frac{\partial S}{\partial q^i} = p_i.
\label{eq:action-wrtq}
\end{equation}
From the definition of the action (\ref{eq:action}) we also have the necessary link between this alternative formulation and the Lagrangian formulation.  In particular we see that the time-derivative of the action $S(q^i,t)$ is equal to the Lagrangian of Section \ref{section:hamiltonian}.  Namely,
\begin{equation}
 dS/dt = L.
 \label{superimportant}
\end{equation}
Combining this with (\ref{eq:action-wrtq}) then yields the following relationship,
\begin{align}
\frac{d S(q^i,t)}{dt} &= \frac{\partial S}{\partial t} + \frac{\partial S}{\partial q^i}\frac{d q^i}{dt} \nonumber \\
											&=\frac{\partial S}{\partial t} + p_i \dot{q}^i.
\label{eq:total-time-deriv-action}
\end{align}

Now consider that we seek a canonical transformation.  In this case by imposing that these new coordinates $(Q^i,P_i)$ also satisfy the principle of least action we will derive an analogous relationship which must hold with respect to the new Hamiltonian.  Thus we can take the difference between the new action and the original action, which we write as
\begin{equation*}
p_i \dot{q}^i - P_i \dot{Q}^i=H - K + \frac{d F}{dt},
\end{equation*}
where $F = F(t,q^i,p_i,Q^i,P_i)$ defines a \emph{generating function} that characterises the canonical transformation used to define the new coordinates.  We can determine the appropriate Legendre transformation in the  above equation by first taking our new generating function to be $G = F + P_i q^i$, so that it depends on the old position coordinates and the new momenta coordinates.  With respect to the \emph{canonically conjugate} quantities $q^i$ and $P_i$, we then arrive at the following system of equations
\begin{equation}
p_i = \frac{\partial G}{\partial q_i}, \quad Q_i = \frac{\partial G}{\partial P_i}, \quad K = H + \frac{\partial G}{\partial t}.
\label{eq:gen-function-ham-eqns}
\end{equation}
The generating function will be chosen in such a way that the new Hamiltonian $K$ is identically zero\footnote{This choice refers to a system which has thus been \emph{linearised}.} (see e.g.~\cite{landau1960}).  In this case the above (final) equation (\ref{eq:gen-function-ham-eqns}) yields (replacing with $G = S$)
\begin{equation}
\frac{\partial S}{\partial t} + H\left(q^i, \frac{\partial S}{\partial q^i}\right) = 0,
\label{eq:hamilton-jacobi-time}
\end{equation}
which one can take as the \emph{Hamilton-Jacobi equation}, though in this thesis, we will reserve this term for another version of the equation.

We proceed a bit more in the derivation upon noting that the Hamilton equations (\ref{eq:canonical-ham-mom}) and (\ref{eq:canonical-ham-pos}) in the new coordinates become trivial, i.e. $\dot{Q}^i = \dot{P}_i = 0,$ so that $Q^i = d^i$ and $P_i = c_i$, where $d^i$ and $c_i$ represent an $n$-tuple of constants.  Thus, we can write
\begin{equation*}
S = S(t,q^i;c_i), \quad p_i = \frac{\partial S}{\partial q^i}, \quad d^i = \frac{\partial S}{\partial c_i}.
\end{equation*}
Since the natural Hamiltonian (\ref{eq:natural-hamiltonian}) does not have any explicit time dependence, then we can let
\begin{equation*}
S(t,q^i;c_i) = S_0(t;c_i) + W(q^i,c_i),
\end{equation*}
which makes it clear that the time dependence is trivial. In particular we can choose $S_0(t;E) = -E t$, which corresponds to the \emph{separation of the time variable} since $E$ only represents a constant. With this choice equation (\ref{eq:hamilton-jacobi-time}) then takes the form of what we will refer to as the \emph{Hamilton-Jacobi} equation
\begin{equation}
H\left(q^i,\frac{\partial W}{\partial q^i} \right) = E.
\label{eq:hamilton-jacobi}
\end{equation}
A \emph{complete integral} of the HJ equation will then be of the form $W = W(q^i;c_i)$ which satisfies the non-degeneracy condition
\begin{equation}
\det\left(\frac{\partial^2 W}{\partial q^i \partial c_j} \right)_{n\times n} \neq 0.
\label{eq:non-degeneracy-hamilton-jacobi}
\end{equation}

By Jacobi's theorem once a complete integral is known then one can compute the motion of the system explicitly, i.e.~determine the trajectories of the Hamiltonian system.  And so  finding a complete integral to the Hamilton-Jacobi equation affords us an alternative way to derive a solution to the associated Hamiltonian system. 
\begin{theo}[Jacobi]
Let $W(q^i;c_i)$ be a complete integral of the Hamilton-Jacobi equation (\ref{eq:hamilton-jacobi}) and let $t_0$ and $d^1,\dots,d^{n-1}$ be arbitrary constants.  Then the functions 
\begin{equation*}
q^i = q^i(t;c_i,d^i)
\end{equation*}
defined by the relations
\begin{equation*}
t - t_0 = \frac{\partial W}{\partial E}, \quad d^i = \frac{\partial W}{\partial c_i}, \quad i = 1,\dots, n-1
\end{equation*}
together with the functions
\[
p_i = \frac{\partial W}{\partial q^i}, \quad i = 1,\dots, n,
\]
form a general solution of the canonical Hamilton equations, (\ref{eq:canonical-ham-mom}) and (\ref{eq:canonical-ham-pos}).
\label{theo:Jacobi}
\end{theo}

This result must be considered in connection with the Arnold-Liouville theorem which states that a system in $n$ degrees of freedom will be integrable by quadratures if it has $n$ functionally independent first integrals in involution (see, e.g~\cite{arnoldbook}).  
In our case of systems with two degrees of freedom, i.e.~defined on the Euclidean plane, we have as a corollary to the Arnold-Liouville theorem that finding one first integral $F$ independent of the Hamiltonian $H$ renders the system integrable by quadratures.  Other criterion and results which determine the integrability of a general $n$-dimensional pseudo-Riemannian manifold have also been established by St\"{a}ckel  in the 1890's \cite{stackel1891}, and Levi-Civita  in the early 1900's \cite{levi1904}.

A powerful way of integrating a Hamiltonian system is crafted through an appropriate canonical transformation that places the associated HJ equation into a separable form.  In this case the HJ equation is said to be integrable via \emph{separation of variables} with respect to the newfound \emph{separable coordinates}, which we will denote as $u^i$.  In this situation the separation ansatz for integrating the HJ equation takes an additive form according to
\begin{equation}
W(u^i;c_i) = W_1(u^1;c_i) + W_2(u^2;c_i) + \dots + W_n(u^n;c_i),
\label{eq:separation-ansatz-hj}
\end{equation}
which is subject to the non-degeneracy condition (\ref{eq:non-degeneracy-hamilton-jacobi}).  Any Hamiltonian system for which there exists such a system of separable coordinates $u^i$ on $\m$ that yields a complete integral of the form (\ref{eq:separation-ansatz-hj}) is termed to be \emph{separable}. \emph{Orthogonal separation of variables} occurs in the case that the canonical transformation is a point transformation that leaves the metric tensor $\g$ diagonalised with respect to the new separable coordinates.  Such a  Hamiltonian system is said to be \emph{orthogonally separable}.  In the next section we elaborate on the historical development of separation of variables on $\mathbb{E}^2$.
\section{Separability on the Euclidean Plane}
\label{section:separability}
One of the primary results of relevance to the present study comes from Liouville in 1846  \cite{liouville1846}. Liouville studied Hamilton's equations for the motion of a particle on a curved surface which was under the influence of a time-independent potential.  He found that if the metric and potential of a corresponding Hamiltonian function took a special (separable) form then the Hamiltonian system was solvable. 
\begin{theo}[Liouville]
\label{theo:liouville}
Let $\m$ be a two-dimensional manifold with local position coordinates $(u,v)$ and corresponding canonical momenta coordinates $(p_u,p_v)$.  If a Hamiltonian function is of the form
\begin{equation}
H = \left(A(u) + B(v)\right)^{-1}\left[\tfrac{1}{2}\left(p_u^2 + p_v^2\right) + C(u) + D(v)\right],
\label{eq:hamiltonian-liouville}
\end{equation}
where $A(u), B(v), C(u)$ and $D(v)$ are arbitrary smooth functions, then the Hamiltonian system given by $\{\m, \boldsymbol{X}_H, H\}$ can be solved by quadratures.\footnote{By \emph{quadratures}, one means through algebraic means or by taking an integral.}  
\end{theo}
Through the equivalence of the Hamilton-Jacobi theory, Liouville's result also demonstrated that the associated HJ equation of (\ref{eq:hamiltonian-liouville}) would be solvable under additive separation of variables.  In covariant form, we can write the metric of the kinetic part of (\ref{eq:hamiltonian-liouville}) as 
\begin{equation}
ds^2 = (A(u) + B(v))(du^2 + dv^2), 
\label{eq:liouville-form}
\end{equation}
so that any potential which appears as in the Hamiltonian (\ref{eq:hamiltonian-liouville}), or any metric of the form (\ref{eq:liouville-form}) is said to be in the \emph{Liouville form}. 

In 1881, Morera proved the converse of Liouville's result \cite{morera1881}.  This established that if a system defined by a natural Hamiltonian (\ref{eq:natural-hamiltonian}) was separable in the context of HJ theory then its metric $\g$ and potential $V(q^i)$ would take the Liouville forms with respect to the separable coordinates.   Nevertheless, Morea's equivalence did not provide a way of determining the separable coordinates, nor did it address the possibility that a Hamiltonian system could be separable in some other system of coordinates, with respect to which the metric and potential were not in the Liouville form.  Bertrand and Darboux,  in 1857 \cite{bertrand1857} and 1901 \cite{darboux1901},  respectively, worked to address these concerns. 

Bertrand searched for natural Hamiltonian systems (\ref{eq:natural-hamiltonian}) that could admit a first integral of motion which took the form
\begin{equation}
F(q^i,p_i) = \tfrac{1}{2}K^{i j}(q^i)p_i p_j + U(q^i), \quad i,j = 1,2, 
\label{eq:first-integral-killing}
\end{equation}\newnot{symbol:first}
for some  $K$ and $U$ in the coordinates $q^i$.  Such a first integral is at most \emph{quadratic} in the momenta.  Bertrand showed that subject to the forms (\ref{eq:natural-hamiltonian}) and (\ref{eq:first-integral-killing}), the vanishing of the Poisson bracket of $H$ and $F$, namely $\{H, F\}_{P_0}=0$, imposed two conditions on the Hamiltonian system.  The first of these is the \emph{Killing tensor equation}
\begin{equation}
[\g,\boldsymbol{K}] = 0,
\label{eq:killing-tensor-equation}
\end{equation}
where $[~,~]$ denotes the Schouten bracket presented in Section \ref{conventions}.
The second is the \emph{compatibility condition}
\begin{equation}
\textnormal{d}\left(\boldsymbol{\hat{K}}\textnormal{d}V\right) = 0,
\label{eq:compatibility-condition}
\end{equation}
which in (local) coordinates yields what is called the \emph{Bertrand-Darboux PDE} (see e.g.~\cite{smirnov2006darboux}).  The compatibility condition appears as the integrability condition imposed on $U(q^i)$ in (\ref{eq:first-integral-killing}) by requiring that $H$ and $F$ be in involution, namely
\begin{equation}
\textnormal{d}U = \boldsymbol{\hat{K}} \textnormal{d}V,
\label{eq:exactness-of-potential}
\end{equation}
where $\boldsymbol{\hat{K}}$ is the (1,1)-tensor defined by $\boldsymbol{\hat{K}}:= \boldsymbol{K} \g^{-1}$, i.e. $\hat{K}^i{}_j := K^{i\ell}g_{\ell j}$.  Together, the Killing tensor equation (\ref{eq:killing-tensor-equation}) and  (\ref{eq:exactness-of-potential}) are equivalent to the vanishing of the Poisson bracket of $H$ and $F$ subject to (\ref{eq:natural-hamiltonian}) and (\ref{eq:first-integral-killing}).

  In 1901, Darboux solved the  linear, second-order PDE admitted by (\ref{eq:compatibility-condition}) using the method of characteristics \cite{darboux1901}, by first simplifying the PDE through a coordinate rotation and translation.  This essentially had the effect of transforming the associated Killing tensor $K^{ij}$ in (\ref{eq:first-integral-killing}) to a canonical form, which he could then diagonalise to derive the appropriate separable coordinates $(u,v)$ \cite{smirnov2006darboux}.  In a systematic way, Darboux arrived at a potential satisfying (\ref{eq:compatibility-condition}) with respect to a canonical Killing tensor whose normal eigenvectors generated elliptic-hyperbolic coordinates. 
  
  Indeed, the problem that Darboux solved was to determine the most general potential which would admit orthogonal separation of variables with respect to an elliptic-hyperbolic coordinate system.  This method is very familiar in classifying systems of superintegrable potentials, where one requires that a system be separable with respect to some particular orthogonal coordinate systems and then solves for the most general potential compatible via (\ref{eq:compatibility-condition}) (see Section \ref{superintegrable}).  The complicated part of the above procedure rests in solving (\ref{eq:compatibility-condition}), which yields a non-linear second-order PDE when expressed in local coordinates.  
  
 In 1934, a new approach to this problem was presented in a famous paper by Eisenhart  \cite{eisenhart-famous} who provided an intrinsic characterisation of orthogonally separable Hamiltonian systems.  Eisenhart worked to extend the results of St{\"a}ckel \cite{stackel1893}, who studied the orthogonal separability of the HJ equation.  St{\"a}ckel established that a necessary condition for a system defined by a natural Hamiltonian to be orthogonally separable was that it admits $n-1$ quadratic first integrals of the form (\ref{eq:first-integral-killing}), all functionally independent and in involution with $H$.  Eisenhart  observed that a smooth function $F \in T^*(\m)$ in the restricted form
\begin{equation}
F(q^i,p_i) = \tfrac{1}{2}K^{i j}(q^i)p_i p_j \quad i,j = 1,2,
\label{eq:eisen}
\end{equation}
would be a first integral of the geodesic Hamiltonian,
\begin{equation}
H(q^i,p_i) = \tfrac{1}{2} g^{i j } p_i p_j,
\label{eq:geodesic-hamiltonian}
\end{equation}
if and only if the functions $K^{i j}$ in (\ref{eq:eisen}) were the components of a characteristic Killing tensor field $\boldsymbol{K} \in \mathcal{K}^2(\m)$.  This result is readily established upon computation of the Poisson bracket of $F$ and $H$ in the assumed forms (\ref{eq:eisen}) and (\ref{eq:geodesic-hamiltonian}), respectively.  The pivotal result Eisenhart established in the theory of orthogonal separation of variables is given by
\begin{theo}[Eisenhart]
\label{theo:eisen}
The Hamiltonian system defined by a geodesic Hamiltonian (\ref{eq:geodesic-hamiltonian}) is orthogonally separable if and only if it admits $n-1$ functionally independent first integrals of motion taking the form (\ref{eq:eisen}), such that
\begin{enumerate}[label=(\arabic*)]
	\item all corresponding Killing two-tensors have real and pointwise distinct (almost everywhere) eigenvalues,
	\item all corresponding eigenvector fields of the Killing two-tensors are normal,
	\item the Killing two-tensors defined by the $n-1$ first integrals all have the same eigenvectors.
\end{enumerate}
\end{theo}
A Killing tensor satisfying (1) and (2) is called a \emph{characteristic Killing tensor}.  
Eisenhart's results demonstrated that Killing tensor fields would play a crucial role in determining the orthogonal separability of a Hamiltonian system.  In 1997, Benenti \cite{benenti1997intrinsic} extended this result to Hamiltonian systems defined by a natural Hamiltonian, i.e.~with non-vanishing potential.    
\begin{theo}[Benenti]
The Hamiltonian system defined by a natural Hamiltonian (\ref{eq:natural-hamiltonian}) is orthogonally separable if and only if there exists a characteristic Killing two-tensor $\boldsymbol{K}$ such that
\[
\textnormal{d}\left(\boldsymbol{\hat{K}}\textnormal{d}V\right) = 0.
\]
\end{theo}The requirement of $n-1$ first integrals in Eisenhart's theorem is reduced to the existence of a single characteristic Killing tensor in Benenti's theorem.  To determine such a Killing tensor, one can start with the $n-1$ Killing tensors afforded by the $n-1$ first integrals of motion,  and take a linear combination of these Killing tensors.  If we let $K_1, \dots, K_{n-1}$ be as in Theorem \ref{theo:eisen}, then adding the metric tensor $\g$ of $\m$ will yield a basis generating an $n$-dimensional vector subspace of $\k^2(\m)$.  A general Killing tensor in this space
\[
\boldsymbol{K} = \g + \sum_{i = 1}^{n -1} \boldsymbol{K}_i,
\]
is then guaranteed to still have real and distinct eigenvalues and the same eigenvectors admitted by the individual $\boldsymbol{K}_i, i =1, \dots, n-1.$  The condition that the eigenvectors be \emph{normal}, that is,
\begin{equation}
E_i \wedge \textnormal{d}E_i = 0, \quad i = 1, \dots, n \quad \textnormal{(no summation)},
\label{eigen}
\end{equation}
where $E_i$ denote the eigenforms of $\boldsymbol{K}$, then implies that the eigenforms generate $n$ foliations\footnote{In the case of $\mathbb{E}^2,$ we have that (\ref{eigen}) is always satisfied, by dimensional considerations, and so that the eigenvectors of Killing tensors defined on $\mathbb{E}^2$ are always normal.}. These will  consist of $(n-1)$-dimensional hypersurfaces \emph{orthogonal} to the eigenvectors of each Killing tensor (see \cite{bruce2001benenti,bruce2001geometrical} for details).  This is the geometric construction of the \emph{orthogonal coordinate web}.  Such a web defines the separable coordinates with respect to which an associated HJ equation separates.  

These two theorems establish the explicit role characteristic Killing tensors play with regards to the HJ theory of orthogonal separation of variables, and essentially provides a link between their algebraic and geometric properties.  In the next section we will explore this link through the theoretical framework of the invariant theory of Killing tensors with an emphasis on defining \emph{joint invariants}, as first introduced in \cite{smirnov2004covariants}, and later utilised in the theory of superintegrable Hamiltonian systems in \cite{adlam2008joint}.
\section{Invariant Theory of Killing Tensors}
\label{section:ITKT}
We begin this section with a brief historical account on the emergence of Killing tensor fields in mathematics.  For a more in-depth account of this history, we direct the reader to \cite{hawkins2000book}, on which this introduction is based. 

In the 1880s, Wilhelm Killing began an extensive study of non-Euclidean geometry in which he sought to develop a theory of \emph{space forms}, which in modern language is a complete Riemannian manifold $\m$ with constant curvature.  Killing began his study analytically by studying the behaviour of infinitesimal motions of an $n$-dimensional continuous manifold of points $(x_1,\dots,x_n)$ with $n$ degrees of freedom.  In his attempt to deal with all possible space forms, Killing imposed an (unconventional) condition on his infinitesimal motions which lead to the implication that they formed a finite-dimensional Lie algebra, a theory mostly unknown to him at the time.  This naturally brought Killing into contact with Sophus Lie, and through his and Lie's results on transformation groups he succeeded in proving that a \emph{proper space form} would have degree $n(n+1)/2$ and admit  a \emph{Riemannian metric}.  From the fact that an infinitesimal motion ought to leave the metric invariant, Killing then arrived at the equations defining what we now call a Killing vector field
\begin{equation}
\mathcal{L}_{\boldsymbol{K}}\g = 0,
\label{eq:killing-vector}
\end{equation}
where $\mathcal{L}_{\boldsymbol{K}}$ denotes the well-studied \emph{Lie derivative} along the vector field $\boldsymbol{K}$, and $\g$ is the metric tensor of the Riemannian manifold $\m$.  Indeed, we can use this Killing vector field to define a function on the cotangent bundle $T^*(\m)$, namely $K^i p_i$, and (\ref{eq:killing-vector}) then implies that this function will be a first integral of the geodesic flow on $T^*(\m)$ with respect to the Riemannian metric $\g$.     
\subsection{Preliminaries}
\label{prelim}
 The Killing fields with which Killing himself worked were \emph{Killing vector fields}, which act as the \emph{infinitesimal generators of isometries} on a Riemannian manifold.  This is made obvious by considering equation (\ref{eq:killing-vector}), where we observed that the vector fields $\boldsymbol{K}$ admitted by this equation were those which preserved the metric $\g$.  Killing tensor fields of higher valence, $p >1$, provide information regarding quadratic, cubic, and higher-order first integrals of the Hamiltonian geodesic flow.  With this in mind we begin our study of the vector space in $(\m,\g)$ formed by Killing tensors of the same valence $p$, which we will denote as $\mathcal{K}^p(\m)$.  The fact that this collection of tensors forms a vector space follows from the bilinear nature of the Schouten bracket (see Section \ref{conventions}).

If $\m$ is a space of constant curvature then the dimension of $\mathcal{K}^p(\m^n)$ will be maximal, and given by the \emph{Delong-Takeuchi-Thompson (DTT) formula} \cite{delongthesis, takeuchi1983killing,thompson1986}
\begin{equation}
d = \dim\k^{p}(\m^n) = \frac{1}{n}\binom{n+p}{p+1}\binom{ n+p-1}{p}, \quad p \geq 1. 
\label{eq:dtt}
\end{equation}\newnot{symbol:dim} \newnot{symbol:binomial}
In this situation, we see that the vector space $\mathcal{K}^p(\m)$ is determined with respect to $d$ arbitrary parameters.  In other words, an element of $\mathcal{K}^p(\m)$, namely a Killing tensor with fixed valence $p$, can be viewed as being an algebraic object in a vector space.  This identification provides the pivotal link between vector spaces of Killing tensors defined on pseudo-Riemannian spaces of constant curvature and the classical invariant theory of vector spaces of homogeneous polynomials (for more details on the latter see e.g.~\cite{olverblue}).  

The focus of this thesis is the vector space of Killing \emph{two-tensors} defined on the Euclidean plane, namely $\mathcal{K}^2(\mathbb{E}^2)$ and products of this space, e.g.~$\mathcal{K}^2(\mathbb{E}^2) \times \mathcal{K}^2(\mathbb{E}^2)$.  The DTT formula yields that the dimension of $\mathcal{K}^2(\mathbb{E}^2)$ is six, since $n =p =2$.  An alternate way of arriving at this number is to formally derive the solution to the Killing tensor equation (\ref{eq:killing-tensor-equation}) in the Euclidean plane, which we do in Appendix \ref{kt-solved}, and count the total number of integration constants in the end.  Additionally, this approach yields that the general solution to (\ref{eq:killing-tensor-equation}) in $\mathbb{E}^2$, with respect to Cartesian coordinates, is a Killing two-tensor $\boldsymbol{K} \in \k^2(\mathbb{E}^2)$, with components given by 
\begin{align}
K^{11} & =  \beta_1 +2 \beta_4 y + \beta_6 y^2,  \nonumber \\
K^{12} & =  \beta_3 -  \beta_4 x  -\beta_5 y  - \beta_6 x y,  \label{eq:general-killing-tensor}\\
K^{22} & =  \beta_2 +2 \beta_5 x  + \beta_6 x^2,  \nonumber
\end{align}
where $\beta_i, i=1,\dots, 6$ are arbitrary parameters which appear as constants of integration in solving (\ref{eq:killing-tensor-equation}).  And so indeed, we can identify the vector space of Killing two-tensors defined on $\mathbb{E}^2$ with $\R^6$.  Another result which will be frequently made use of throughout this thesis is the following lemma, independently arrived at by \cite{delongthesis,takeuchi1983killing,thompson1986} in the early 1980's: 
\begin{lem}
Any Killing tensor of valence $p$ defined on a constant curvature pseudo-Riemannian manifold $(\m,\g)$ can be expressed as a sum of symmetrised tensor products of a basis of Killing vectors on $\m$.
\label{lemma:killing-tensor-as-sum}
\end{lem}
We make use of the above lemma to write our general Killing tensor $\boldsymbol{K} \in \mathcal{K}^2(\mathbb{E}^2)$ as
\begin{equation}
\boldsymbol{K} = A^{i j} \boldsymbol{X}_i \odot \boldsymbol{X}_j + B^i \boldsymbol{X}_i \odot \boldsymbol{R} + C \boldsymbol{R} \odot \boldsymbol{R},
\label{eq:killing-tensor-as-sum}
\end{equation}
where $\boldsymbol{X}_i$ \newnot{symbol:transbasis} denote the translational basis Killing vectors and $ \boldsymbol{R}$ \newnot{symbol:rotbasis} denotes the rotational basis vector. 
In this case the $A^{i j}, B^{i}$ and $C$ define the \emph{Killing tensor parameters}, which we can identify with (\ref{eq:general-killing-tensor}) by letting
\begin{equation}
A^{i j} = \left(
\begin{array}{cc}
	\beta_1 & \beta_3\\
	\beta_3 & \beta_2  
\end{array}
		  \right), 
\quad 
B^i = \left(
\begin{array}{c}
	\beta_4 \\
	-\beta_5
\end{array}
	  \right),
\quad C = \beta_6,
\label{eq:killing-tensor-constants}
\end{equation}
which will be subject to the symmetry property
\[
A^{i j} = A^{(i j)}.
\]
Furthermore, we will take as a basis for the space of Killing vectors on $\mathbb{E}^2$  the usual  vectors given in Cartesian coordinates, $(q^1,q^2) = (x,y)$, namely 
\begin{equation}
\boldsymbol{X}_1 = \frac{\partial}{\partial x}, \quad \boldsymbol{X}_2 = \frac{\partial}{\partial y}, \quad \boldsymbol{R} = y \frac{\partial}{\partial x} - x \frac{\partial}{\partial y}.
\label{eq:basis-killing-vectors}
\end{equation}
The components of the general Killing tensor (\ref{eq:killing-tensor-as-sum}) with respect to the basis (\ref{eq:basis-killing-vectors}) are then expressible in a compact way as
\begin{equation}
K^{i j} = A^{i j} + 2 \epsilon^{(i}{}_{\ell} B_{}^{j)} x^{\ell} + C \epsilon^i{}_m \epsilon^j{}_k x^m x^k.
\label{eq:killing-tensor-as-sum-components}
\end{equation}
One can check (see Section \ref{conventions}) that indeed, subject to (\ref{eq:killing-tensor-constants}) this reproduces the general Killing tensor components defined in (\ref{eq:general-killing-tensor}).  In this compact presentation, it is straightforward to compute the transformation law for the parameters $\beta_i, i=1,\dots, 6$ under the action induced by the isometry group $SE(2) \circlearrowright \mathbb{E}^2$ \newnot{symbol:speciale} \newnot{symbol:acting} which acts  transitively on the Euclidean plane  via
\begin{equation}
q^i \to \Lambda^i{}_j q^j + \delta^i, 
\label{eq:se-group-action}
\end{equation}
where $\Lambda^i{}_j \in SO(2)$, \newnot{symbol:specialo} $\boldsymbol{\delta} \in \R^2$, so $\delta^i=p_i, i=1,2$.  This computation is carried out using the compact notation in Appendix \ref{group-action-stuff}.  For clarity we will now provide the explicit computation of the induced action. 

We seek to derive the transformation law for the parameters in $\k^2(\mathbb{E}^2)$ under the action of $SE(2)$, which we can obtain by making use of the usual transformation rules for tensor components (mentioned in Section \ref{conventions}). In particular, a Killing two-tensor $K^{ij}$ transforms according to the tensor transformation rules via
\begin{equation}
\tilde{K}^{i j} = \Lambda^{i}{}_{\ell} \Lambda^{j}{}_k K^{\ell k},
\label{eq:trans-kt}
\end{equation}
where the $\Lambda^{i}{}_j$ represent the Jacobian matrices (\ref{eq:jacobian}) defined in Section \ref{conventions}. The action of $SE(2) \circlearrowright \mathbb{E}^2$ is identified by
\begin{equation}
	\left( \begin{array}{c}
	\tilde{q}^1 \\ \tilde{q}^2 \end{array}\right) = \left(\begin{array}{cc} \cos\theta & -\sin\theta \\ \sin\theta &\cos\theta \end{array}\right) \left(\begin{array}{c} q^1 \\ q^2 \end{array}\right) + \left(\begin{array}{c} p_1\\ p_2 \end{array} \right),
\label{eq:actionyes}
\end{equation}
so that identifying the components of the Jacobian via (\ref{eq:jacobian}), namely
\begin{equation}
\Lambda^{j}{}_{k} = \frac{\partial \tilde{q}^{j}}{\partial q^k},
\label{eq:jacob}
\end{equation}
we see that indeed this is given by a rotation matrix:
\[
\boldsymbol{\Lambda} = \left(\begin{array}{cc} {\partial \tilde{q}^1}/{\partial q^1} & {\partial \tilde{q}^1}/{\partial q^2} \\ {\partial \tilde{q}^2}/{\partial q^1} & {\partial \tilde{q}^2}/{\partial q^2} \end{array}\right) = \left(
\begin{array}{cc}
	\cos\theta & -\sin\theta\\
	\sin\theta & \cos\theta  
\end{array}
		  \right).
\]
Thus, computationally, equation (\ref{eq:trans-kt}) amounts to (appropriately) multiplying the $2 \times 2$-matrix defining the Killing two tensor $K^{ij}$ with two rotation matrices:
\[
\tilde{\boldsymbol{K}} = \boldsymbol{\Lambda}^T \boldsymbol{K} \boldsymbol{\Lambda},
\]
where we see that we must multiply on the left by the \emph{transpose} of the Jacobian matrix in order to carry out the proper tensor index transformation.  Carrying out this calculation, we obtain in a simple manner that the components of $K^{ij}$ transform according to
\begin{align}
\tilde{K}^{11} &= - K^{11}\cos^2\theta + 2K^{12}\cos\theta\sin\theta + K^{22}\sin^2\theta, \nonumber \\
\tilde{K}^{12} &= -K^{11}\cos\theta\sin\theta + K^{12}(\cos^2\theta -\sin^2\theta) + K^{22}\cos\theta\sin\theta, \label{eq:kt-tranny}\\
\tilde{K}^{22} &= K^{11}\sin^2\theta - 2K^{12}\cos\theta\sin\theta + K^{22}\cos^2\theta.   \nonumber
\end{align}
Using the compact formulae (\ref{eq:killing-tensor-as-sum-components}), the non-transitive action of $SE(2) \circlearrowright \mathcal{K}^2(\mathbb{E}^2)$ is elegantly derived by substituting into (\ref{eq:killing-tensor-as-sum}) with the basis vector transformation rules (\ref{eq:basis-tran}), which yields equation (\ref{A3})\footnote{See Appendix \ref{group-action-stuff} for this computation.}.  Since we work in $\mathbb{E}^2$ we can get by without the compact notation, and so we continue with the above procedure for deriving the invariants as follows.  First, we must substitute with the polynomial components of $K^{ij}$ given by (\ref{eq:general-killing-tensor}) into the transformed components (\ref{eq:kt-tranny}) and then set this equal to $\tilde{K}^{ij}$ written in \emph{transformed} parameters, i.e. we arrive at the following system of equations (we have let $(q^1,q^2) = (x,y)$ for clarity)
\begin{align*}
\tilde{\beta_1} +2 \tilde{\beta_4} \tilde{y} + \tilde{\beta_6} \tilde{y}^2 &= - (\beta_1 +2 \beta_4 y + \beta_6 y^2)\cos^2\theta \\
& \quad + 2(\beta_3 -  \beta_4 x  -\beta_5 y  - \beta_6 x y)\cos\theta\sin\theta \\
& \quad + ( \beta_2 +2 \beta_5 x  + \beta_6 x^2)\sin^2\theta \\
\tilde{\beta_3} -  \tilde{\beta_4} \tilde{x}  -\tilde{\beta_5} \tilde{y}  - \tilde{\beta_6} \tilde{x}\tilde{ y} &=  -(\beta_1 +2 \beta_4 y + \beta_6 y^2)\cos\theta\sin\theta \\
& \quad + (\beta_3 -  \beta_4 x  -\beta_5 y  - \beta_6 x y)(\cos^2\theta -\sin^2\theta) \\
& \quad + ( \beta_2 +2 \beta_5 x  + \beta_6 x^2)\cos\theta\sin\theta\\
\tilde{\beta_2} +2 \tilde{\beta_5} \tilde{x}  + \tilde{\beta_6} \tilde{x}^2 &= (\beta_1 +2 \beta_4 y + \beta_6 y^2)\sin^2\theta \\
& \quad - 2(\beta_3 -  \beta_4 x  -\beta_5 y  - \beta_6 x y)\cos\theta\sin\theta \\
& \quad + ( \beta_2 +2 \beta_5 x  + \beta_6 x^2)\cos^2\theta
\end{align*}
Next, we must substitute for $\tilde{x}$ and  $\tilde{y}$ according to (\ref{eq:actionyes}), after which we can equate coefficients of like terms.  This then yields the following explicit transformation laws for the parameters of $\k^2(\mathbb{E}^2)$:
\begin{equation}
\label{eq:transformed-KT-parameters}
\begin{array}{rcl}
\tilde{\beta_1} &=& \beta_1\cos^2\theta - 2\beta_3\cos \theta\sin \theta + \beta_2\sin^2\theta - 2p_2\beta_4\cos \theta, 
\\
                  & & - 2p_2\beta_5\sin \theta + \beta_6p_2^2, \\ 
\tilde{\beta_2} & =& \beta_1\sin^2 \theta - 2\beta_3\cos \theta\sin \theta, +
\beta_2\cos^2\theta - 2p_1\beta_5\cos \theta,  \\ 
 & &+ 2p_1\beta_4\sin \theta  + \beta_6p_1^2,\\
\tilde{\beta_3} & = & (\beta_1-\beta_2)\sin \theta\cos \theta + \beta_3(\cos^2 \theta - \sin^2 \theta),
\\ 
 & &  + (p_1\beta_4 + p_2\beta_5)\cos \theta + (p_1\beta_5 - p_2\beta_4)\sin \theta - \beta_6p_1p_2, \\
 \tilde{\beta_4} & = & \beta_4\cos \theta + \beta_5\sin \theta - \beta_6p_2, \\
 \tilde{\beta_5} & = & \beta_5\cos \theta - \beta_4\sin \theta - \beta_6p_1,\\
 \tilde{\beta_6} &=& \beta_6. 		  
\end{array}
\end{equation}
This brings to light our first (algebraic) invariant, namely $\Delta_1 = \beta_6.$\footnote{See Appendix A \ref{A2} for the compact form of these six parameters under the group action.}   In order to derive the remaining fundamental invariants we use the method of moving frames. 

\subsection{The Moving Frames Method}
\label{section:group-action}
In what follows we provide a synopsis of the \emph{moving frames method} and its role in the invariant theory of Killing tensors.  For a more complete account we direct the reader to the wealth of literature that has nourished the development of this theory, namely \cite{mclenaghan2002group,deeley2003theory,adlam2005thesis,adlam2008joint,chanu2006geometrical,chana2009rsep,cochran2011thesis,horwood2008thesis,horwood2008invariants,horwood2005invariant,horwood2009minkowski,smirnov2006darboux,smirnov2004covariants,winternitz1965invariants,yue2005thesis} and references therein. In particular, \cite{mclenaghan2002group} pioneered the study of invariants under the isometry group acting on $\mathcal{K}^p(\m)$ into the theory of Killing tensors as a way of classifying the orthogonal coordinate webs in the Euclidean plane.

As noted in  Section \ref{section:separability}, it has been shown (see e.g,  \cite{eisenhart-famous,benenti1997intrinsic,horwood2005invariant}) that an element $\boldsymbol{K} \in \mathcal{K}^2(\m)$ with real and distinct eigenvalues and normal eigenvectors generates an orthogonal coordinate web with $n= \dim(\m)$ foliations whose leaves are $n-1$ dimensional hypersurfaces orthogonal to the eigenvectors of $\boldsymbol{K}$.  In this situation Cartan's geometry \cite{cartan1935} can be combined with the study of principal fiber bundles to provide a framework for the invariant theory.  For a detailed account of the underlying ideas of the generalised moving frames method we direct the reader to \cite{olverblue}, from which this brief overview is based on.  In this section, we focus on demonstrating the recursive approach to constructing the moving frame map, developed by Kogan in 2003 \cite{kogan2003recursive}.  We begin with some preliminary definitions mostly compatible with \cite{olverblue,smirnov2004covariants,smirnov2006darboux,mclenaghan2002group}.

Let $G$ be an $r$-dimensional (Lie transformation) group acting smoothly on an $n$-dimensional manifold $\m$ with $s$-dimensional orbits.  Representative points in each group orbit can be chosen to depend continuously on the orbits provided that the group acts \emph{regularly}\footnote{Regularity means that all the orbits of the group action on $\m$ have the same dimension and each point $p \in \m$ has a system of arbitrarily small neighborhoods whose intersection with each orbit is a pathwise connected subset of the orbit \cite{olverblue}}.  A (local) cross-section is an $(n-s)-$dimensional submanifold $K \subset \m$ that intersects each group orbit transversally\footnote{Two submanifolds $K,N$ are said to intersect transversally at a point $p \in K \cap N$ if $T_p(K) \cap T_p(N) = \{0\}$,i.e.~if they share no common non-zero tangent vectors.} and at most once.  In general, a local cross-section passing through any point in $p \in \m$ can be constructed provided that the the group acts regularly on $\m$.  

Choosing a cross-section amounts to fixing a \emph{moving frame}, i.e.~identifying a map $f: \mathcal{K}^2(\m)/G \to G$, \newnot{symbol:quotientspace} where we focus on the quotient space, considering $G$ as the subspace $G = \{c \g|c \in \R\}$ generated by the metric, which in our case is trivial since it provides no information regarding separable coordinate systems \cite{mclenaghan2002group}.  A moving frame is formally defined as  a smooth, G-equivariant map\footnote{ A map $\gamma: \m \to G$ is G-equivariant if for $g \in G$, $p \in \m$, $\gamma(g \cdot p) = \gamma(p) \cdot g^{-1}$ with respect to the actions of $G$ on $\m$ and on itself by right multiplication.} $\rho: \m \to G$ whose existence  is guaranteed in a neighborhood of a point $p \in \m$ provided that $G$ acts \emph{freely}\footnote{A transformation group $G$ acts freely provided that the isotropy group, $G_p = \{g \in G | g \cdot p = p \}$ where $p\in \m$ is trivial: $G_p = {e},$ for all $p \in \m$, where $e$ denotes the identity element of $G$.  In other words, only the identity element fixes any $p \in \m$.} and \emph{regularly} near $p$ \cite{olverblue}.  To construct a moving frame, one utilises Cartan's normalisation method, whereby the elements of the cross-section can be interpreted as canonical forms for general elements in the underlying manifold $\m$.  The normalisation procedure amounts to choosing local coordinates $g = (g_1, \dots, g_r)$ on $G$ near the identity element, identifying explicit formulae for the group transformations in the coordinates, and then equating the first $r$ components of the formulae in the previous step to given constants.  The Implicit Function theorem then guarantees  that such a system of equations is locally soluble.  

We shall illustrate the method explicitly by taking $G$ to be the isometry group of the Euclidean plane, which acts as an automorphism, i.e.~Killing tensors are mapped to Killing tensors in $\mathcal{K}^2(\mathbb{E}^2)$, and so preserves the geometry of the vector space.  Explicitly, we take $G = SE(2)$, the special Euclidean group, which acts freely and regularly in the vector space of Killing tensors \cite{smirnov2004covariants}.  

Recall that $SE(2)$ consists of the \emph{rigid body} motions, translations and rotations, and so it can be identified as the semi-direct product of the \emph{special orthogonal group} $SO(2)$ with the group of translations on $\mathbb{E}^2$.  The recursive version of the moving frames method to derive invariants then proceeds in two steps.  First, we compute a set of fundamental invariants under the subgroup of translations and choose an appropriate cross-section to obtain a set of translational invariants.  Second, we use these translational invariants as new coordinates in the vector space and compute the action under the subgroup of rotations, $SO(2)$.  After choosing an additional cross-section, or \emph{normalisation} equation, we then arrive at a final set of invariants under the action of the full isometry group \cite{kogan2003recursive}.  

The action of the subgroup of translations amounts to substituting with $\Lambda^j{}_i = \delta^j{}_i$ in (\ref{eq:se-group-action}) and (\ref{eq:transformed-KT-parameters}), where $\delta^j{}_i$ denotes the Kronecker delta.   In other words, the group of translations acts on $\mathbb{E}^2$ via $q^i \to \tilde{q}^i + \delta^i$ and so we find that this induces the transformation\footnote{See Appendix \ref{noway}, eq. \ref{A3} for the compact form of this transformation.}, 
\begin{equation}
\label{eq:trans-group-action}
\begin{array}{rcl}
\tilde{\beta_1} &=& \beta_1 + 2 \beta_4 p_2 + \beta_6 p_2^2,
\\
\tilde{\beta_2} & =& \beta_2 + 2 \beta_5 p_1 + \beta_6 p_1^2,\\
\tilde{\beta_3} & = & \beta_3 - \beta_4 p_1 - \beta_5 p_2 - \beta_6 p_1 p_2,
\\ 
 \tilde{\beta_4} & = & \beta_4 + \beta_6 p_2,\\
 \tilde{\beta_5} & = & \beta_5 + \beta_6 p_1,\\
 \tilde{\beta_6} &=& \beta_6. 		  
\end{array}
\end{equation}
We then choose our cross-section for the restricted group action via
\begin{equation}
\tilde{\beta_4} = \tilde{\beta_5} = 0.
\label{eq:cross-section-restricted}
\end{equation}
The moving frame map for the restricted group action is then defined by the following normalisation equations
\begin{equation}
p_1 = -\beta_4/\beta_6, \quad p_2 = -\beta_5/\beta_6
\label{eq:normal-equations}
\end{equation}
We then substitute these relationships into the remaining four equations (\ref{eq:trans-group-action}) and thus arrive at the following fundamental (translational) invariants of $\k^2(\mathbb{E}^2)$
\begin{align}
I_1 &=\beta_1\beta_6 - \beta_4^2,  \nonumber \\
I_2 &=\beta_2\beta_6-\beta_5^2, \label{eq:trans-group-invariants}\\
I_3 &=\beta_3\beta_6 + \beta_4\beta_5, \nonumber\\
I_4 &=\beta_6. \nonumber
\end{align}
Next, we determine the action of the subgroup of rotations $SO(2)$ on the space of the four invariants (\ref{eq:trans-group-invariants}).  Since $SO(2)$ acts on $\mathbb{E}^2$ via $x \to \Lambda^i{}_j \tilde{x}^j,$ then we can use the same technique as above to determine the transformation formulas induced by this action.  Thus, we find that $SO(2)$ acts on $\k^2(\mathbb{E}^2)$ according to\footnote{See Appendix \ref{noway} eq.~\ref{A4} for compact formulae.}
\begin{equation}
\label{eq:rot-group-action}
\begin{array}{rcl}
\tilde{\beta_1} &=& \beta_1 \cos^2\theta + \beta_2 \sin^2\theta - 2\beta_3 \sin\theta\cos\theta,
\\
\tilde{\beta_2} & =& \beta_1 \sin^2\theta + \beta_2\cos^2\theta - 2\beta_3 \sin\theta\cos\theta,\\
\tilde{\beta_3} & = & (\beta_1-\beta_2)\sin\theta\cos\theta + \beta_3(\cos^2\theta - \sin^2\theta),
\\ 
 \tilde{\beta_4} & = & \beta_4 \cos\theta + \beta_5\sin\theta,\\
 \tilde{\beta_5} & = & \beta_4\sin\theta - \beta_5 \cos\theta,\\
 \tilde{\beta_6} &=& \beta_6. 		  
\end{array}
\end{equation}
Therefore, taking into consideration (\ref{eq:trans-group-invariants}) and (\ref{eq:rot-group-action}), we find that
\begin{align}
\tilde{I}_1 &=-\left(\cos \theta  \beta _4+\sin \theta  \beta _5\right){}^2+\left(\cos ^2 \theta \beta _1+\sin \theta  \left(\sin \theta  \beta _2-2 \cos \theta  \beta _3\right)\right) \beta _6, \nonumber\\
\tilde{I}_2 &=-\left(\sin \theta  \beta _4-\cos \theta  \beta _5\right){}^2+\left(\sin ^2 \theta \beta _1+\cos \theta  \left(\cos \theta  \beta _2-2 \sin \theta  \beta _3\right)\right) \beta _6, \nonumber \\
\tilde{I}_3 &=\tfrac{1}{2} \sin2\theta\left(-\beta _4^2+\beta _5^2+\left(\beta _1-\beta _2\right) \beta _6\right)-\cos 2 \theta  \left(\beta _4 \beta _5+\beta _3 \beta _6\right),\nonumber \\
\tilde{I}_4 &=I_4. \label{invariantos}
\end{align}
We can then obtain the final normalisation equation by taking $\tilde{I}_3= 0$, which yields
\begin{equation*}
\theta = \frac{1}{2}\arctan{\left(\frac{2(\beta_4\beta_5 + \beta_3 \beta_6)}{\beta_6(\beta_1 - \beta_2)-\beta_4^2 + \beta_5^2}\right)}.
\end{equation*}
Substituting this into (\ref{invariantos}), we recover the $SE(2)$ invariants 
\begin{align}
\Delta_1 &= \beta_6, \nonumber \\
\Delta_2 &= \beta_6(\beta_1 + \beta_2)-\beta_4^2 -\beta_5^2, \label{eq:invariants}\\
\Delta_3 &= \left(\beta_6(\beta_1-\beta_2)-\beta_4^2+\beta_5^2\right)^2 + 4\left(\beta_6\beta_3+\beta_4\beta_5 \right)^2, \nonumber
\end{align}
established first in \cite{winternitz1965invariants} and later in \cite{mclenaghan2002group,smirnov2004covariants,smirnov2006darboux}. As was shown first by Winternitz and Fri{\v{s}} in 1965 \cite{winternitz1965invariants} and again in 2002  by McLenaghan et al.~\cite{mclenaghan2002group} via the invariant theory of Killing tensors, the invariants $\Delta_1$ and $\Delta_3$ can be used to classify the orbits of $\k^2(\mathbb{E}^2)/SE(2)$ as follows
\begin{align*}
\textnormal{Elliptic-hyperbolic} &:~ \Delta_1 \neq 0, \quad \Delta_3 \neq 0,\\
\textnormal{Parabolic} &:~ \Delta_1 = 0, \quad \Delta_3 \neq 0,\\
\textnormal{Polar} &:~ \Delta_1 \neq 0, \quad\Delta_3 = 0, \\
\textnormal{Cartesian} &:~ \Delta_1 = 0,\quad  \Delta_3 = 0.
\end{align*}
 
The final tool we require from the invariant theory comes from considering the group action on the product space obtained by taking two copies of the vector space of Killing tensors: $\k^2(\mathbb{E}^2)\times \k^2(\mathbb{E}^2)$. \newnot{symbol:times}  This action leads to the notion of \emph{joint invariants} of Killing tensors, the topic of discussion featured in our next section. 
\subsection{Joint Invariants of Killing Tensors}
\label{section:joints}
We now present the geometric properties of joint invariants of Killing two-tensors  (to be defined below) as first introduced by Smirnov and Yue in 2004 \cite{smirnov2004covariants}.  Our notation will be compatible with Adlam et al.~\cite{adlam2008joint}, who extended the study established in \cite{smirnov2004covariants} and demonstrated its applicability to superintegrable systems.  We specialise this development for the isometry group which acts on the product space $\k^{2}(\mathbb{E}^2) \times \k^{2}(\mathbb{E}^2)$, where we begin with the joint invariants of \emph{non-degenerate} orbits of the orbit space $({\k}^2(\mathbb{E}^2) \times {\k}^2(\mathbb{E}^2))/SE(2)$ studied in \cite{adlam2008joint}.  

A three-dimensional orbit is defined to be \emph{non-degenerate} provided that, along the orbit, the invariant given by $k^2 = \sqrt{\Delta_3}/(\Delta'_1)^2,$ is non-vanishing.  The action $SE(2) \circlearrowright {\k}^2(\mathbb{E}^2) \times {\k}^2(\mathbb{E}^2)$ for which  $k^2 \neq 0$ is likewise termed non-degenerate.  The geometric interpretation of the invariant $k^2$ is the (half) distance between the foci of an elliptic-hyperbolic coordinate web generated by the normal eigenvectors of a Killing two-tensor, and so it follows that such Killing tensors will have a correspondence with the non-degenerate orbits \cite{mclenaghan2002group,smirnov2006darboux,winternitz1965invariants}.  

In this situation, the isometry group $SE(2)$ acts on each copy of $\k^{2}(\mathbb{E}^2)$ with three-dimensional non-degenerate orbits. If we denote the parameters of each vector space by $\alpha_i,\beta_i, i=1, \dots, 6,$ respectively, then the group action is induced by the corresponding transformation laws given by (\ref{eq:transformed-KT-parameters}).  Explicitly, we have that the parameters $\alpha_i, i=1,\dots,6$ in the first copy of $\k^{2}(\mathbb{E}^2)$ transform as
\begin{align}
\label{JAa}
\begin{array}{rcl}
\tilde{\alpha}_1 &=& \alpha_1\cos^2\theta - 2\alpha_3\cos \theta\sin \theta + \alpha_2\sin^2\theta - 2p_2\alpha_4\cos \theta 
\\
                  & & - 2p_2\alpha_5\sin \theta + \alpha_6p_2^2, \\ 
\tilde{\alpha}_2 & =& \alpha_1\sin^2\theta - 2\alpha_3\cos \theta\sin \theta +
\alpha_2\cos^2\theta - 2p_1\alpha_5\cos \theta  \\ 
 & &+ 2p_1\alpha_4\sin \theta  + \alpha_6p_1^2,\\
\tilde{\alpha}_3 & = & (\alpha_1-\alpha_2)\sin \theta\cos \theta + \alpha_3(\cos^2\theta - \sin^2\theta) 
\\ 
 & &  + (p_1\alpha_4 + p_2\alpha_5)\cos \theta + (p_1\alpha_5 - p_2\alpha_4)\sin \theta - \alpha_6p_1p_2, \\
 \tilde{\alpha}_4 & = & \alpha_4\cos \theta + \alpha_5\sin \theta - \alpha_6p_2, \\
 \tilde{\alpha}_5 & = & \alpha_5\cos \theta - \alpha_4\sin \theta - \alpha_6p_1,\\
 \tilde{\alpha}_6 &=& \alpha_6,
\end{array}
\end{align}
and the equivalent relationship in terms of the parameters $\beta_i, i=1,\dots,6,$ for the second copy,
\begin{align}
\label{JAb}
\begin{array}{rcl}
\tilde{\beta}_1 &=& \beta_1\cos^2\theta - 2\beta_3\cos \theta\sin \theta + \beta_2\sin^2\theta - 2p_2\beta_4\cos \theta \\
                  & & - 2p_2\beta_5\sin \theta + \beta_6p_2^2, \\ 
\tilde{\beta}_2 & =& \beta_1\sin^2\theta - 2\beta_3\cos \theta\sin \theta +
\beta_2\cos^2\theta - 2p_1\beta_5\cos \theta  \\ 
 & &+ 2p_1\beta_4\sin \theta  + \beta_6p_1^2,\\
\tilde{\beta}_3 & = & (\beta_1-\beta_2)\sin \theta\cos \theta + \beta_3(\cos^2\theta - \sin^2\theta) 
\\ 
 & &  + (p_1\beta_4 + p_2\beta_5)\cos \theta + (p_1\beta_5 - p_2\beta_4)\sin \theta - \beta_6p_1p_2, \\
 \tilde{\beta}_4 & = & \beta_4\cos \theta + \beta_5\sin \theta - \beta_6p_2, \\
 \tilde{\beta}_5 & = & \beta_5\cos \theta - \beta_4\sin \theta - \beta_6p_1,\\
 \tilde{\beta}_6 &=& \beta_6.
\end{array}
\end{align} 
Additionally, the conditions 
\[
k^2_1 = \frac{\sqrt{(\alpha_4^2-\alpha_5^2 + \alpha_6(\alpha_2-\alpha_1))^2 + 4(\alpha_6\alpha_3 + \alpha_4\alpha_5)^2 }}{\alpha_6} \not=0,
\]
\[
k^2_2 = \frac{\sqrt{(\beta_4^2-\beta5^2 + \beta_6(\beta_2-\beta_1))^2 + 4(\beta_6\beta_3 + \beta_4\beta_5)^2 }}{\beta_6} \not=0
\]
hold true.

If we denote the six dimensional space defined by the parameters $\alpha_i, i=1,\dots,6$ by $A \simeq \R^6$, then we can define a \emph{joint invariant} \cite{smirnov2004covariants} of the product space ${\k}^2(\mathbb{E}^2) \times {\k}^2(\mathbb{E}^2)$ as a function $\mathcal{J}:A \times A \to \R$ that satisfies
\begin{align}
\mathcal{J} & = F(\alpha_1, \dots, \alpha_6, \beta_1,\dots,\beta_6)
\label{eq:JIdef}\\
& = F(\tilde{\alpha_1}, \dots, \tilde{\alpha_6}, \tilde{\beta_1},\dots,\tilde{\beta_6}),
\end{align}
under the transformation laws induced by the isometry group $SE(2)$.  Six joint invariants of the action $SE(2) \circlearrowright {\k}^2(\mathbb{E}^2) \times {\k}^2(\mathbb{E}^2)$ then immediately follow from (\ref{eq:invariants}), namely
\begin{equation}
\label{eq:joint-invariants}
\begin{array}{rcl}
{\Delta}'_1 & = & \alpha_6,\\
{\Delta}'_2 & = & \alpha_6(\alpha_1 + \alpha_2) - \alpha_4^2 - \alpha_5^2,\\
{\Delta}'_3 & = & (\alpha_6(\alpha_1 - \alpha_2) - \alpha_4^2 + \alpha_5^2)^2 + 4(\alpha_6 \alpha_3 + \alpha_4 \alpha_5)^2,\\
{\Delta}'_4 & = & \beta_6,\\
{\Delta}'_5 & = & \beta_6(\beta_1 + \beta_2) - \beta_4^2 - \beta_5^2,\\
{\Delta}'_6 & = & (\beta_6(\beta_1 - \beta_2) - \beta_4^2 + \beta_5^2)^2 + 4(\beta_6 \beta_3 + \beta_4 \beta_5)^2.
\end{array}
\end{equation}
We remark that since the joint invariants are obtained via the moving frames method then they are functionally independent.  Further, any analytic function of the invariants (\ref{eq:joint-invariants}) will also be a joint invariant of the non-degenerate group action $SE(2) \circlearrowright {\k}^2(\mathbb{E}^2) \times {\k}^2(\mathbb{E}^2)$ \cite{olverblue}.  The \emph{Fundamental Theorem on invariants of regular Lie group actions} tells us that we should obtain a total of 9 fundamental invariants, which we arrive at by subtracting the dimension of the orbits, 3, from the dimension of the product space   ${\k}^2(\mathbb{E}^2) \times {\k}^2(\mathbb{E}^2)$, 12.  Thus, we are in search of three more fundamental invariants.  

These remaining joint invariants can be derived by employing geometric consideration.  Recall that each element of a \emph{non-degenerate} orbit of the action ${\k}^2(\mathbb{E}^2) \times {\k}^2(\mathbb{E}^2)$ corresponds to a Killing tensor whose normal eigenvectors generate an elliptic-hyperbolic web.  The foci represent the \emph{singular points} of this coordinate web, which are characterised, following \cite{benenti1997intrinsic}, as the points where the eigenvalues of the associated Killing tensor are equal.  Establishing this degeneracy, one can solve for the explicit form of the singular points in terms of the arbitrary parameters, from which we have that the coordinates for the two foci, denoted $S_1, S_2$,  in the elliptic-hyperbolic web are (see \cite{mclenaghan2002group})
\begin{equation}
\label{S1S2}
\begin{array}{lr}
(x_1,y_1)_{S_1}  = & \\
& \left(\frac{-\beta_5}{\beta_6} + \frac{1}{\beta_6}\left(\frac{\sqrt{\Delta'_6} - \sigma_1}{2}\right)^{1/2}, 
\frac{-\beta_4}{\beta_6} + \frac{1}{\beta_6}\left(\frac{\sqrt{\Delta'_6} + \sigma_1}{2}\right)^{1/2}\right), \\[1cm]
(x_2,y_2)_{S_2}  = & \\
& \left(\frac{-\beta_5}{\beta_6} - \frac{1}{\beta_6}\left(\frac{\sqrt{\Delta'_6} - \sigma_1}{2}\right)^{1/2}, 
\frac{-\beta_4}{\beta_6} - \frac{1}{\beta_6}\left(\frac{\sqrt{\Delta'_6} + \sigma_1}{2}\right)^{1/2}\right), 
\end{array}
\end{equation} 
where $\sigma_1 = \beta_4^2 - \beta_5^2 + \beta_6(\beta_2 - \beta_1)$ and $\Delta'_6$ is as defined in (\ref{eq:joint-invariants}).  Considering the product space $ \k^2(\mathbb{E}^2) \times \k^2(\mathbb{E}^2)$,  we have another set of foci, $(x_3,y_3)_{S_3}$ and $(x_4,y_4)_{S_4},$ appropriately defined with respect to the parameters $\alpha_i, i=1,\dots, 6$, $\sigma_1$ in terms of $\alpha_i$, and $\Delta'_3$.  In \cite{adlam2008joint}, this identification motivates the authors to take the action $SE(2) \circlearrowright  \k^2(\mathbb{E}^2) \times \k^2(\mathbb{E}^2)$ as the free and regular action $SE(2) \circlearrowright \mathbb{E}^2 \times \mathbb{E}^2 \times \mathbb{E}^2 \times \mathbb{E}^2$, and  hence, use the \emph{Weyl theorem on joint invariants} to conclude that the square of the distances between the foci will be a joint invariant of the preceding action.  Thus, in the case of a non-degenerate action, we can take the remaining three invariants as 
\begin{equation}
\label{eq:joint-invariants-three}
\begin{array}{rcl}
{\Delta}'_7 & = & d^2(S_2,S_3) = (x_2 - x_3)^2 + (y_2 - y_3)^2,\\
{\Delta}'_8 & = & d^2(S_1,S_3) = (x_1 - x_3)^2 + (y_1 - y_3)^2,\\ 
{\Delta}'_9 & = & d^2(S_2,S_4) = (x_2 - x_4)^2 + (y_2 - y_4)^2,\\
\end{array}
\end{equation}\newnot{symbol:distance}
where $(x_i,y_i), i=1,\dots,4$ are given by (\ref{S1S2}) in terms of $\beta_i$ and analogously in terms of $\alpha_i$.  Note that we require the quadrilateral $S_1S_3S_4S_2$ to be \emph{rigid} so that we need to only specify the distance of one of the diagonals and the other can be determined as a function of these nine invariants.  Other geometrically motivated joint invariants can be established in this manner as well, which would be obtained as functions of the invariants ${\Delta}'_1 - {\Delta}'_9$, such as the areas within the quadrilateral whose vertices are given by the four singular points--- see Figure \ref{quad}.  In the next chapter, we characterise the case when one of the Killing tensors $\boldsymbol{K} \in \k^2(\mathbb{E}^2)$ belongs to a degenerate orbit of the orbit space $\k^2(\mathbb{E}^2)/SE(2)$.
\begin{figure}[h!tb]
	\centering
		\includegraphics[width=\textwidth]{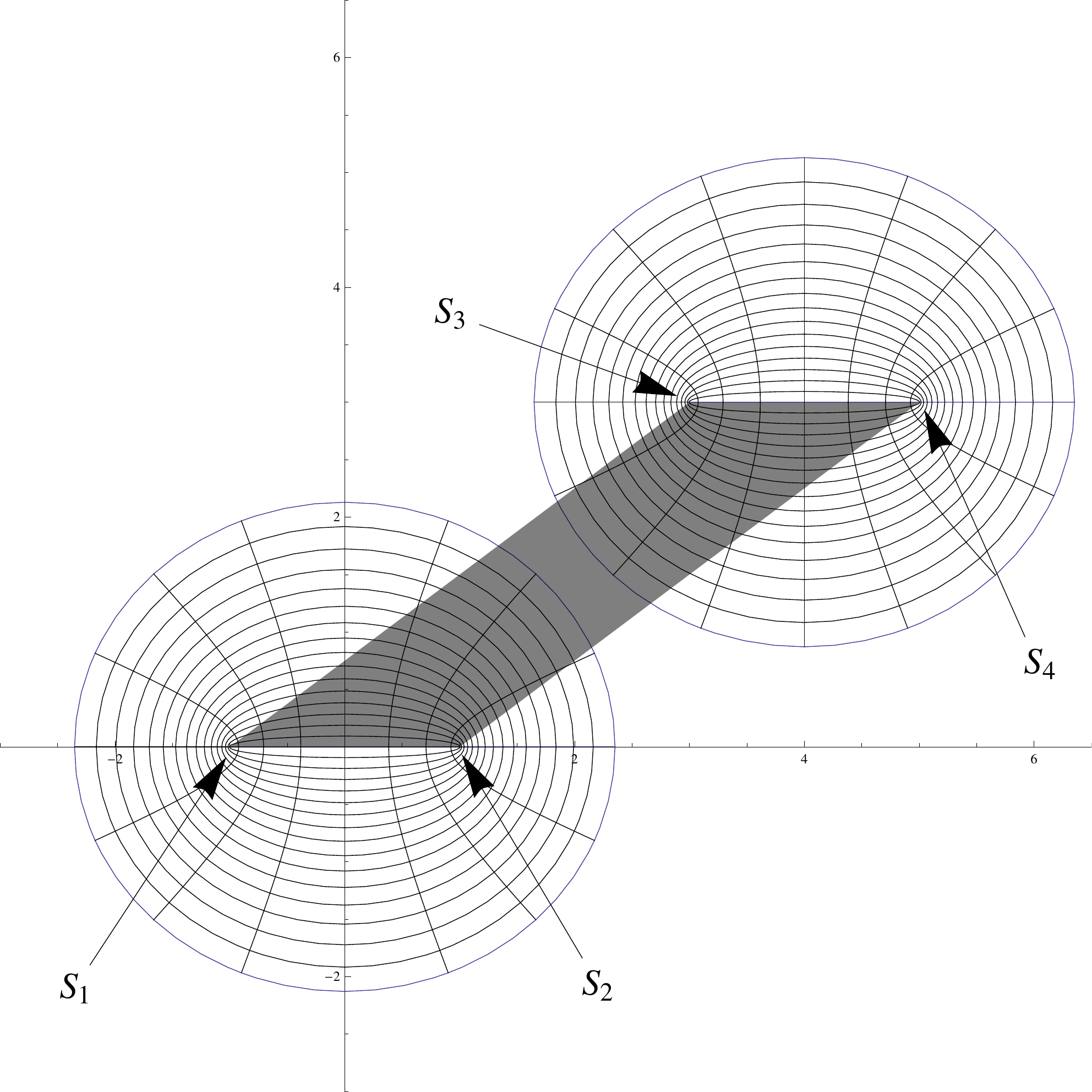}
	\caption{Quadrilateral whose vertices are given by the singular points of two non-degenerate orbits $\{\boldsymbol{K}_{EH},\boldsymbol{K}_{EH}\} \in \k^2(\mathbb{E}^2) \times \k^2(\mathbb{E}^2)$}
\label{quad}
\end{figure}



\chapter{Superintegrable Systems on the Euclidean Plane}
\label{chapter:sw}

\ifpdf
\graphicspath{{Chapter3/Chapter3Figs/PNG/}{Chapter3/Chapter3Figs/PDF/}{Chapter3/Chapter3Figs/}}
\else
\graphicspath{{Chapter3/Chapter3Figs/EPS/}{Chapter3/Chapter3Figs/}}
\fi

We recall the significance of the first integrals of motion for a Hamiltonian system.  These are smooth functions in the generalised canonical coordinates that remain constant along the Hamiltonian flow, or orbits, of a system.  As such, they can be identified as \emph{constants of motion} for a system which elicits many familiar examples of first integrals.  Indeed, the total energy, angular momentum, and linear momentum are all examples of constants of motion that provide physical insight about a Hamiltonian system: in particular, each of these is an example of a conserved quantity that corresponds to a symmetry of the Lagrangian.  A systematic way of deriving these quantities from symmetry considerations is the heart of \emph{Noether's theorem} (see e.g.~\cite{goldstein}) where, for instance, conservation of energy is an artifact of the Lagrangian being invariant under changes in time.  In brief, we see that the first integrals of motion help to establish physical properties of a Hamiltonian system by signifying the existence of structure and symmetry.

\section{Development of Superintegrable Systems}
\label{superintegrable}
A Hamiltonian system with $n$ degrees of freedom is said to be \emph{completely integrable} if it admits $n$ globally defined, functionally independent first integrals of motion in involution with respect to the Poisson bracket (see Section \ref{intro}).  In the case that such a system admits $n+1$ such first integrals of motion the system is said to be \emph{minimally superintegrable}, whereas it is termed to be \emph{maximally superintegrable} if it admits $2n-1$ such first integrals.  In some cases we find that a system is solvable via orthogonal separation of variables in more than one system of coordinates, which qualifies such systems as being \emph{multi-separable}.  For the purpose of this study, we focus on superintegrable systems that exhibit this property of being multi-separable.

As discussed in Section \ref{section:history}, we can study the first integrals of natural Hamiltonian systems 
\begin{equation}
H(q^i,p_i) = \tfrac{1}{2} g^{i j} p_i p_j + V(q^i),
\label{eq:natty}
\end{equation}
 defined on flat pseudo-Riemannian manifolds which take a form quadratic in the momenta according to
\begin{equation}
F(q^i,p_i) = \tfrac{1}{2}K^{i j}(q^i)p_i p_j + U(q^i), \quad i,j = 1,2,
\label{eq:superfirst}
\end{equation}
which will afford orthogonal separation of variables in the associated Hamilton-Jacobi equation in the case that $K^{ij}$ is a characteristic Killing two-tensor satisfying the compatibility condition
\begin{equation}
\textnormal{d}\left(\boldsymbol{\hat{K}}\textnormal{d}V\right) = 0.
\label{eq:supercompat}
\end{equation}
This is the content of Benenti's theorem (see Section \ref{section:separability}). 

We can generalise this development to the study of first integrals \emph{polynomial} in the momenta of degree $p$ upon considering Killing tensors of valence $p$ (for a recent account, see e.g.~\cite{horwood2007higher})\footnote{Though, in studying Killing tensors of valence $p > 2$, we lose the interpretation of a system being orthogonally separable.}.  Indeed, the search for first integrals of motion in higher degrees of the momenta has been extensively studied for over a century, with Drach initiating the study of first integrals \emph{cubic} in the momenta in 1908 \cite{drach1908, ranada1997drach}.  In 2002, Gravel and Winternitz considered superintegrable systems on $\mathbb{E}^2$ characterised by one linear and one (non-trivial) cubic first integral \cite{gravel2002third}.  This work was continued by Gravel who then characterised systems with one quadratic first integral (derived from the system admitting separability in Cartesian coordinates), and one cubic first integral \cite{gravel2004}.  The origins of the invariant theory approach can be traced back to Winternitz and Fri{\v{s}} in 1965 \cite{winternitz1965invariants}, in which the authors computed the invariants of second-order symmetries of the Laplace equation defined on $\mathbb{E}^2$ under the action of the isometry group, and then used these invariants to classify the orthogonal separable webs on $\mathbb{E}^2$.  In 2002, this geometric approach resurfaced with McLenaghan et al.~in \cite{mclenaghan2004cubic}, who presented it in the language of the invariant theory of  Killing tensors.  In 2008, Adlam et al.~\cite{adlam2008joint} extended the theory constructed in \cite{mclenaghan2002group} and derived an invariant characterisation of the superintegrable potential admitted by the Kepler problem through the use of invariants and joint invariants of Killing two tensors that determined first integrals according to (\ref{eq:superfirst}).  

In this thesis, we continue the work of classifying superintegrable potentials on the Euclidean plane by providing a new perspective on the ideas developed in \cite{adlam2008joint}.  In particular, we investigate the joint invariants of a degenerate and non-degenerate orbit of the orbit space $(\k^2(\mathbb{E}^2)\times \k^2(\mathbb{E}^2))/SE(2)$ admitted by a Killing tensor whose eigenvalues generate an elliptic-hyperbolic web in the canonical position, and a second Killing tensor whose eigenvalues generate a polar web in the non-canonical position.  To this end we begin by providing a joint invariant characterisation of the Smorodinsky-Winternitz system, 
\begin{equation}
H = \tfrac{1}{2}\left(p_x^2 +  p_y^2\right) -\omega^2 \left(x^2 + y^2\right) + \frac{\alpha}{x^2} + \frac{\beta}{y^2},
\label{eq:sw-ham}
\end{equation}
where $\alpha, \beta,$ and $\omega$ are arbitrary parameters.  Our goal is to use joint invariants of Killing two tensors $\boldsymbol{K}$ defined on $\mathbb{E}^2$ \cite{smirnov2004covariants} to provide a link between the arbitrary parameters of the SW potential and the parameters which characterise the associated vector space of Killing tensors $\k^2(\mathbb{E}^2)$. 

Further, this will provide motivation for studying an interesting integrable perturbation of the SW potential that has recently appeared in the literature known as the \emph{Tremblay-Turbiner-Winternitz (TTW) potential}.  Our interest will be in determining which potentials admitted by the TTW system remain multi-separable in the Euclidean plane.  We remark that the results of the next two sections have been submitted for publication (see \cite{mypaper}).
\section{Characterisation of the Smorodinsky-Winternitz Potential} 
\label{section:sw}
The superintegrability of the Smorodinsky-Winternitz (SW) potential is a simple consequence of its multi-separability with respect to both canonical polar and canonical Cartesian coordinates.  The property of the coordinate web being \emph{canonical} can be interpreted either geometrically, (see Figure \ref{canon-webs}), or with respect to the form of the general characteristic Killing tensor defining the orthogonal coordinate web (see \cite{mclenaghan2002group}).  Geometrically, canonical Cartesian coordinates are aligned with the coordinate axes and canonical polar coordinates have their singular point coincide with the origin of the coordinate axes; the associated Killing tensors for each canonical web respectively take the forms (\ref{eq:cart-canon-KT}) and (\ref{eq:polar-canon-KT}) given below.  

Since the SW potential is of the form $V = f(x) + g(y)$, it immediately follows that the system enjoys separability with respect to (canonical) Cartesian coordinates, and so the Hamiltonian flow defined with respect to (\ref{eq:sw-ham}) admits a first integral quadratic in the momenta according to
\begin{equation}
F_{1} = \tfrac{1}{2}K^{ij}_{1} p_i p_j + U_{1}, \quad i, j = 1, 2,
\label{eq:f1}
\end{equation}
where $K^{ij}_1$ is the Killing two-tensor whose normal eigenvectors generate the (canonical) Cartesian coordinate web:
\begin{equation}
K^{ij}_{1} = \left(
\begin{array}{cc}
 1 & 0 \\
 0 & 0
\end{array}
\right).
\label{eq:cart-canon-KT}
\end{equation}
Imposing the compatibility condition (\ref{eq:supercompat}) requires that $\textnormal{d}U_1 = \boldsymbol{\hat{K}}_1 \textnormal{d}V,$ which we can readily integrate to obtain that  $U_1(x,y) = \frac{\alpha}{x^2} - \omega^2 x^2$.  This provides us with the first non-trivial quadratic first integral. 

If we transform the potential of (\ref{eq:sw-ham}) to polar coordinates we immediately see that it is of the from $V(r, \theta) = \frac{f(\theta)}{r^2} + g(r)$, and so that indeed the Hamiltonian system defined by (\ref{eq:sw-ham}) is also separable with respect to (canonical) polar coordinates.  Thus, we have that it admits a second quadratic first integral according to 
\begin{equation}
F^{ij}_2 = \tfrac{1}{2}K^{ij}_2 p_i p_j + U_2, \quad i,j = 1,2,
\label{eq:f2}
\end{equation}
where $K^{ij}_2$ is given by the Killing two-tensor whose normal eigenvectors generate the (canonical) polar web:
\begin{equation}
K^{ij}_2 = \left(
\begin{array}{cc}
 y^2 & -x y \\
 -x y & x^2
\end{array}
\right)
\label{eq:polar-canon-KT}
\end{equation}
In this case, we find upon integrating $\textnormal{d}U_2 = \boldsymbol{\hat{K}}_2 \textnormal{d}V$  that $U_2(x,y) = \frac{a y^2}{x^2} + \frac{b x^2}{y^2}$, and so we have obtained two non-trivial quadratic first integrals.

The functional independence of $F_1$ and $F_2$ follows from considering that the wedge product of their differentials with $\textnormal{d}H$ at a point $p \in \mathbb{E}^2$ is non-vanishing
\[
\textnormal{d}H \wedge \textnormal{d}F_1 \wedge \textnormal{d}F_2(p) \neq 0,
\] 
a criterion which yields that the set of $F_1$ and $F_2$ defines a coordinate system, since this is equivalent to saying that the Jacobian determinant of the functions $F_1$ and $F_2$ at a point $p \in \m$ is non-vanishing, and so by the Implicit Function theorem the set indeed forms a coordinate system (e.g.~see \cite{lee2003intro}).  Furthermore, since both are mutually in involution with the Hamiltonian, i.e.~$\{F_1,H\} = \{F_2,H\} = 0$, then we have at once that the Hamiltonian system (\ref{eq:sw-ham}) is superintegrable (since it is multi-separable) with $2n-1 = 3$ functionally independent first integrals of motion, namely the set $\{H, F_1, F_2\}$.  Since the function $F_3 = F_1 + F_2$ is also a first integral of motion, with an associated Killing two-tensor $\boldsymbol{K}_1 + \boldsymbol{K}_2 = \boldsymbol{K}_{EH},$ where $\boldsymbol{K}_{EH}$ denotes the (canonical) Killing tensor whose normal eigenvectors generate the elliptic-hyperbolic coordinate web, then it follows that the SW potential will also be separable in (canonical) elliptic-hyperbolic coordinates.  Indeed, the canonical characteristic Killing tensor that generates the elliptic-hyperbolic coordinate web is given by
\begin{equation}
K^{ij}_{EH} = \left(
\begin{array}{cc}
 y^2 +k^2  & -x y \\
 -x y & x^2
\end{array}
\right),
\end{equation}
(where $k^2$ is an invariant to be defined below) which makes it clear that any potential of a Hamiltonian system compatible with both the valence two Killing tensor of polar type and the valence two Killing tensor of Cartesian type will also be compatible with the valence two Killing tensor of elliptic-hyperbolic type, by linearity of the exterior derivative.  

We can now make use of joint invariants, introduced in Section \ref{section:joints}, to conclude that separation in (canonical) polar and (canonical) elliptic-hyperbolic coordinates implies that the quadrilateral (see Figure \ref{quad}) $S_1S_2S_4S_3$ degenerates in the following manner.  Denote by $S_1$ and $S_2$ the singular points of the Killing two-tensor $\boldsymbol{K}_3$ whose first integral is given by $F_3$ as defined above (i.e.~$\boldsymbol{K}_3 = \boldsymbol{K}_{EH}$), and similarly, take $S_3$ and $S_4$ as the singular points of $\boldsymbol{K}_2$.  Then since both $\boldsymbol{K}_1$ and $\boldsymbol{K}_2$, and hence $\boldsymbol{K}_3$ are in canonical form, we find that the quadrilateral $S_1S_2S_4S_3$ degenerates according to the invariant conditions $S_3 = S_4$, and $\textnormal{dist}(S_1,S_3) = \textnormal{dist}(S_2,S_4)$.  In terms of the fundamental joint invariants (\ref{eq:joint-invariants}), we can equivalently write these conditions as 
\[
\Delta'_3 = 0, \qquad \Delta'_8 = \Delta'_9,
\]  
respectively.

Conversely, if we begin with a general potential $V$ of a natural Hamiltonian (\ref{eq:natty}) and impose that the Hamiltonian system be multi-separable with respect to canonical polar and canonical Cartesian coordinates, i.e.~impose that it admits two quadratic first integrals given by $F_1$ and $F_2$ as above, then integrating the  conditions $\textnormal{d}U_i = \boldsymbol{\hat{K_i}}\textnormal{d}V$ for $i=1,2$ will yield the Smorodinsky-Winternitz potential, which follows from \cite{fris1965higher,adlam2005thesis}.  Therefore, we can conclude that the superintegrable system given by the Smorodinsky-Winternitz potential is characterised by a pair of Killing tensors $(\boldsymbol{K}_2,\boldsymbol{K}_3) \in \k^2(\mathbb{E}^2) \times \k^2(\mathbb{E}^2)$ whose position in the orbit space   $(\k^2(\mathbb{E}^2) \times \k^2(\mathbb{E}^2))/SE(2)$ is determined by the invariant conditions 
\[
\Delta'_1 \neq 0, \quad \Delta'_3 = 0,\quad \Delta'_4 \neq 0,\quad \Delta'_6 \neq 0,\quad \Delta'_7 = \Delta'_8 = \Delta'_9.
\]
Thus, we have the following:
\begin{theo}
Let $V$ be the potential of a natural Hamiltonian 
\[ 
H(q^i,p_i) = \tfrac{1}{2} g^{ij}(q^i)p_ip_j + V(q^i), i, j = 1,2
\]
 satisfying $\textnormal{d}\left(\boldsymbol{\hat{K}}_{\ell} \textnormal{d}V\right) = 0, \ell=2,3$ for a pair of Killing tensors $(\boldsymbol{K}_2,\boldsymbol{K}_3) \in \k^2(\mathbb{E}^2) \times \k^2(\mathbb{E}^2)$.  Then the following statements are equivalent:
\begin{enumerate}[label=(\arabic*)]
	\item The pair $(\boldsymbol{K}_2,\boldsymbol{K}_3) \in \k^2(\mathbb{E}^2) \times \k^2(\mathbb{E}^2)$ is invariantly characterised by the conditions 
	\begin{align}
	\Delta'_1 &\neq 0, \quad \Delta'_3 = 0,\quad \Delta'_4 \neq 0,\nonumber \\
	\Delta'_6 &\neq 0,\quad \Delta'_7 = \Delta'_8 = \Delta'_9.
  \label{inv}
  \end{align}
  \item The potential $V$ given by the natural Hamiltonian is the Smorodinsky-Winternitz potential given in Cartesian coordinates $(q^1, q^2) = (x,y)$ by
  \begin{equation}
  V(x,y) = -\omega^2 (x^2 + y^2) + \frac{\alpha}{x^2} + \frac{\beta}{y^2}.
  \label{eq:swpotty}
  \end{equation}
\end{enumerate}
\end{theo}
Thus, we have shown that by invariantly characterising the pair $(\boldsymbol{K}_2,\boldsymbol{K}_3) \in \k^2(\mathbb{E}^2) \times \k^2(\mathbb{E}^2)$, we have in fact characterised the SW potential which they define.
Next, we study what happens to the form of the SW potential in the case that we weaken the invariant conditions (\ref{inv}).  
\subsection{Weakening the Conditions}
We shall first suppose that the Hamiltonian system (\ref{eq:natty}) admits two quadratic first integrals of motion according to (\ref{eq:f1}) and (\ref{eq:f2}).  Further, suppose that the Killing tensor $\boldsymbol{K}_{P}$ generates (non-canonical) polar coordinates, and that the Killing tensor $\boldsymbol{K}_{EH}$ generates (non-canonical) elliptic-hyperbolic coordinates.  
Without loss of generality, and to facilitate calculations, we will assume that $\boldsymbol{K}_{EH}$ is in the canonical form, so that the elliptic-hyperbolic web is generated in the canonical position.   Then the singular points of the two webs form a general triangle $\triangle S_1S_2S_3$--- see Figure \ref{triangle}.  
\begin{figure}[h!tb]
	\centering
		\includegraphics[width=0.75\textwidth]{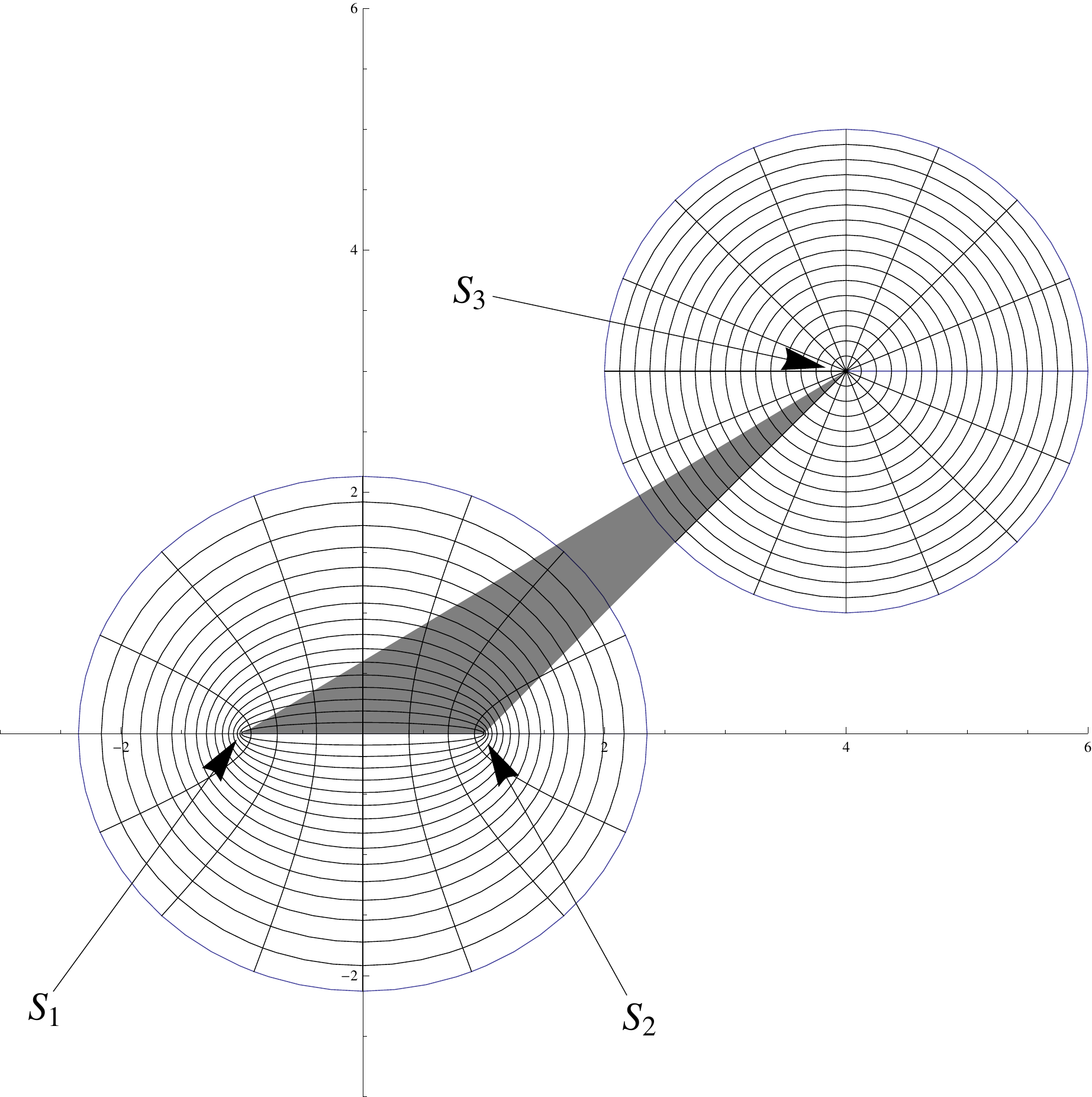}
	\caption{Triangle whose vertices are given by the singular points of two Killing tensors, one non-degenerate and the other degenerate,  $(\boldsymbol{K}_{EH},\boldsymbol{K}_{P}) \in \k^2(\mathbb{E}^2) \times  \k^2(\mathbb{E}^2)$.  In this case, $\boldsymbol{K}_{EH}$ is in the canonical form.} 
\label{triangle}
\end{figure}
Considering the fundamental joint invariants (\ref{eq:joint-invariants}), this arrangement induces the following invariant conditions for the pair $(\boldsymbol{K}_{P},\boldsymbol{K}_{EH}) \in \k^2(\mathbb{E}^2) \times \k^2(\mathbb{E}^2)$:
\begin{align}
\Delta'_1 \neq 0, \Delta'_3 &= 0, \Delta'_4 \neq 0, \nonumber\\
\Delta'_6 \neq 0, \Delta'_7 &= \Delta'_9 \neq \Delta'_8.  
\end{align}
We then have the following components for the respective Killing tensors as follows \cite{mclenaghan2002group}
\begin{equation}
K^{ij}_{P} = \left(
\begin{array}{cc}
 (y-b)^2 & (x-a) (y-b) \\
 (x-a) (y-b) & (x-a)^2
\end{array}
\right), \quad
a, b \in \R,
\label{polarKT}
\end{equation}
and
\begin{equation}
K^{ij}_{EH} = \left(
\begin{array}{cc}
 y^2 + k^2 & -x y \\
 - x y & x^2
\end{array}
\right),
k^2 \in \R,
\label{ehKT}
\end{equation}
where $k^2$ is the invariant given by (half) the distance between the foci of the elliptic-hyperbolic web (see Section \ref{section:joints}).  We then ask the question of how this setup affects the arbitrary parameters in the Smorodinsky-Winternitz potential (\ref{eq:swpotty}).  

We must first determine the freedom we have in the arbitrary parameters afforded by the vector space of Killing tensors.  Taking $F_1$ (\ref{eq:f1}) and $F_2$ (\ref{eq:f2}) as the first integrals of the Hamiltonian system defined by (\ref{eq:sw-ham}), the Killing two-tensor obtained from taking a linear combination as
\[
\boldsymbol{K}_g = c_1 \boldsymbol{K}_{P} + c_2 \boldsymbol{K}_{EH} + c_3 \g, \qquad c_1, c_2, c_3 \in \R,
\] 
\emph{must} also satisfy the compatibility condition with the Smorodinsky-Winternitz potential $V$ admitted by (\ref{eq:sw-ham}).  However, this raises a concern since $\boldsymbol{K}_g$ depends on six arbitrary parameters given by $k^2, a, b, c_1, c_2, c_3$, and six also happens to be the dimension of the vector space of Killing tensors $\k^2(\mathbb{E}^2).$  This implies that $\boldsymbol{K}_g$ must in fact be given by the general formula 
\begin{equation}
\begin{array}{rcl}
{\bf K} & = & \displaystyle (\beta_1 + 2\beta_4 y + \beta_6 y^2)\frac{\partial}{\partial x}
\odot  \frac{\partial}{\partial x} \\ 
& & + \displaystyle (\beta_3 - \beta_4x -
 \beta_5y - \beta_6x y)\frac{\partial}{\partial x} \odot \frac{\partial}{\partial y} \\ 
& & + \displaystyle (\beta_2 + 2\beta_5x+\beta_6x^2) \frac{\partial}{\partial y}\odot
\frac{\partial}{\partial y},
\end{array}
\label{gKt}
\end{equation}
from which it follows by the compatibility condition (\ref{eq:supercompat}) that the potential $V$ is trivial: $V =$ constant.  Therefore, a non-trivial potential will be obtained if at least one of the six parameters vanishes.  Considering the geometric implications afforded if we retain both Killing tensors $\boldsymbol{K}_{P}$ and $\boldsymbol{K}_{EH}$, the metric $\g$, and $k^2$, the only real freedom that remains is to switch off the parameters $a$ and $b$.  This presents us with the following cases:
\begin{center}
\begin{enumerate}[label=(\arabic*)]
	\item $a \neq 0, \quad b\neq 0$,\\
	\item $a = 0,\quad b\neq 0$,\\
	\item $a \neq 0, \quad b=0$, \\
	\item $a = b = 0$.
\end{enumerate}
\end{center}

The first case imposed that the SW potential would assume a trivial form.  We can introduce a new fundamental invariant (replacing $\Delta'_9$ in (\ref{eq:joint-invariants}), since with respect to the two orthogonal coordinate webs $\Delta'_9 = \Delta'_7$) given by the area of $\triangle S_1S_2S_3$.  In terms of the fundamental joint invariants (\ref{eq:joint-invariants}) and the area of $\triangle S_1S_2S_3$, we then have the following invariant characterisation.
\begin{center}
\begin{enumerate}[label=(\arabic*)]
	\item $\Delta'_1 \neq 0, \Delta'_3 = 0, \Delta'_4 \neq 0, \Delta'_6 \neq 0, \Delta'_7 \neq 0, \Delta'_8 \neq 0$, the area of $\triangle S_1S_2S_3 \neq 0$.  
	\item $\Delta'_1 \neq 0, \Delta'_3 = 0, \Delta'_4 \neq 0, \Delta'_6 \neq 0, \Delta'_7 - \Delta'_8 = 0$, the area of $\triangle S_1S_2S_3 \neq 0$.  In this case the $\triangle S_1S_2S_3$ takes the form of an isosceles triangle. 
	\item $\Delta'_1 \neq 0, \Delta'_3 = 0, \Delta'_4 \neq 0, \Delta'_6 \neq 0,\Delta'_7 \neq \Delta'_8,$ the area of $\triangle S_1S_2S_3 = 0$.  We remark that this is the case when the singular points all lie on the $x-$axis.  
	\item $\Delta'_1 \neq 0, \Delta'_3 = 0, \Delta'_4 \neq 0, \Delta'_6 \neq 0, \Delta'_7 = \Delta'_8$, area of $\triangle S_1S_2S_3 =0$
\end{enumerate}
\end{center}

We can next characterise the arbitrary parameters $a$ and $b$ as joint invariants themselves by writing them in terms of the fundamental joint invariants (\ref{eq:joint-invariants}).  To this end, we consider the general Killing tensor whose components with respect to Cartesian coordinates $(q^1,q^2)$ are given by (\ref{gKt}) and suppose that $\beta_6 \neq 0, \Delta =0,$ where 
\[
\Delta = (\beta_4^2 - \beta_5^2 + \beta_6(\beta_2-\beta_1))^2 + 4(\beta_6 \beta_3 + \beta_4 \beta_5)^2.
\] 
If we consider the Killing tensor $\boldsymbol{K}_{P}$, we can compare it with  (\ref{gKt}) via
\begin{equation}
K_{ij} = \ell g_{ij} + m \tilde{K}_{ij}^{P},
\label{eq:blah}
\end{equation}
where $\tilde{K}_{ij}^{P} = g_{i k}g_{\ell j}K^{k \ell}_{P}$, with $K^{k \ell}_{P}$ given by (\ref{polarKT}).  Comparing components on both sides of (\ref{eq:blah}) then amounts to identifying
\begin{align*}
\alpha_1 + 2 \alpha_4 y + \alpha_6 y^2 &= \ell + m (y-b)^2, \\
\alpha_3 - \alpha_4 x - \alpha_5 y - \alpha_6 x y &= -m(x-a)(y-b),\\
\alpha_2 + 2\alpha_5 x + \alpha_6 x^2 &= \ell + m(x-a)^2,
\end{align*}
which admits the following identification 
\[
(a, b) = \left(\frac{-\alpha_5}{\alpha_6}, \frac{-\alpha_4}{\alpha_6}\right).
\]
Through a bit of algebraic manipulation, we can then obtain an expression for each parameter in terms of the joint invariants (\ref{eq:joint-invariants}):
\[
a = \frac{(\Delta'_4)^2(\Delta'_8 - \Delta'_7) - \tfrac{1}{2}\Delta'_6}{2\Delta'_4\sqrt{\Delta'_6}}
\]
and
\[
b = \sqrt{\Delta'_7 - \left(a - \frac{\sqrt{\Delta'_6}}{2\Delta'_4}\right)^2}.
\]
Solving the PDEs that come from imposing the compatibility condition (\ref{eq:supercompat}) for the Killing two-tensor defined by (\ref{polarKT}) and the SW potential (\ref{eq:swpotty}), we arrive at the following equation
\begin{equation}
\omega^2(b x - ay) + \frac{a(b-y)\alpha}{x^4} + \frac{b(-a+x)\beta}{y^4} = 0,
\label{eq:compatsw}
\end{equation} 
which puts into evidence the following relationship
\begin{center}
\begin{enumerate}[label=(\arabic*)]
	\item $a \neq 0, b \neq 0$ yields the trivial potential $V =0$,
	\item $a =0, b\neq 0$ yields the potential $V = \dfrac{\beta}{y^2}$,
	\item $a \neq 0, b = 0$ yields the potential $V = \dfrac{\alpha}{x^2}$,
	\item $a = 0, b = 0$ yields the general Smorodinsky-Winternitz potential (\ref{eq:swpotty}). 
\end{enumerate}
\end{center}
Henceforth, we arrive at the conclusion that the more geometric structure the pair $(\boldsymbol{K}_{P},\boldsymbol{K}_{EH}) \in \k^2(\mathbb{E}^2) \times \k^2(\mathbb{E}^2)$ has (in the sense of placing the polar coordinate web in restricted positions), the more general the Smorodinsky-Winternitz potential (\ref{eq:swpotty}) becomes.  
In each of the cases described above the corresponding pair $(\boldsymbol{K}_{P},\boldsymbol{K}_{EH}) \in \k^2(\mathbb{E}^2) \times \k^2(\mathbb{E}^2)$ belongs to a different orbit of the group action $SE(2) \circlearrowright \k^2(\mathbb{E}^2) \times \k^2(\mathbb{E}^2)$ (see Section \ref{section:group-action}, eq. \ref{eq:transformed-KT-parameters}).  In fact, our above analysis shows that these corresponding orbits have been explicitly distinguished by means of joint invariants (\ref{eq:joint-invariants}) of the action (\ref{eq:transformed-KT-parameters}).  Considering our work as an extension of the results obtained in \cite{adlam2008joint}, we see that we have demonstrated that the use of joint invariants to classify the orbits of $(\k^2(\mathbb{E}^2) \times \k^2(\mathbb{E}^2))/SE(2)$ leads to a classification of the corresponding superintegrable potentials defined by quadratic first integrals in $\mathbb{E}^2$.  We remark on the further development of this perspective in Chapter \ref{conc}.
\section{The Tremblay-Turbiner-Winternitz System}
\label{section:ttw}
We now consider the Tremblay-Turbiner-Winternitz (TTW) potential introduced in 2009 \cite{TTW2009infinite}.  It is defined in polar coordinates by the natural Hamiltonian
\begin{equation}
H = p_r^2 + \frac{1}{r} p_r + \frac{1}{r^2}p_{\theta}^2 - \omega^2 r^2 + \frac{1}{r^2}\left(\frac{\lambda_1}{\cos^2{k\theta}} + \frac{\lambda_2}{\sin^2{k\theta}}\right),
\label{ttweq}
\end{equation}
which demonstrates that its potential,
\begin{equation}
V(r, \theta) = - \omega^2 r^2 + \frac{1}{r^2}\left(\frac{\lambda_1}{\cos^2{k\theta}} + \frac{\lambda_2}{\sin^2{k\theta}}\right),
\label{ttwpotty}
\end{equation}
is of the from $V(r, \theta) = \frac{f(\theta)}{r^2} + g(r)$ and so is separable in (canonical) polar coordinates.  By the results of Section \ref{section:separability}, we have that the system admits a first integral of motion taking the form
\[
F(q^i,p_i) = \tfrac{1}{2}K^{i j}(q^i)p_i p_j + U(q^i), \quad i,j = 1,2,
\]
where $K^{ij}$ is the Killing two-tensor which generates (canonical) polar coordinates, i.e.~(\ref{eq:polar-canon-KT}) and $U$ is determined via $\textnormal{d}U = \boldsymbol{\hat{K}}\textnormal{d}V$.  We then seek to answer the following question posed in the introduction of this thesis: 
\begin{center}
\emph{For which values of $k$ is the TTW system multi-separable?}  
\end{center}
In other words, we seek to determine for which values of $k$ the TTW potential (\ref{ttwpotty}) admits an additional (non-trivial) functionally independent quadratic first integral of motion.

To allow for arbitrary $k \in \R$, we cannot transform the TTW system back into Cartesian coordinates.  Indeed, such a transformation would require the use of trigonometric identities in $\cos k\theta$ and $\sin k\theta$, which require $k$ to in general be an integer.  Thus we must carry out all our calculations in polar coordinates.  To this end we shall first rewrite the formula defining a general Killing tensor $\boldsymbol{K}_g \in \k^2(\mathbb{E}^2)$ in polar coordinates.  

On this note, we make use of the result that on a manifold $\m$ of constant curvature, generalised Killing tensors are formed as sums of symmetrised tensor products of generalised Killing vectors on $\m$ (see \cite{delongthesis}).  This means that on $\mathbb{E}^2$, we have  
\begin{equation}
\boldsymbol{K}_g = A^{i j} \boldsymbol{X}_i \odot \boldsymbol{X}_j + B^i \boldsymbol{X}_i \odot \boldsymbol{R} + C \boldsymbol{R} \odot \boldsymbol{R}, \quad i,j = 1,2,
\label{eq:killing-as-sum}
\end{equation}
for some tensor $A^{ij}$, vector $B^i$ and constant $C.$  If we take as our basis of Killing vectors on $\mathbb{E}^2$
\begin{align}
\boldsymbol{X}_1 &= \cos\theta \frac{\partial}{\partial r} - \frac{\sin\theta}{r}\frac{\partial}{\partial \theta}, \nonumber \\
\boldsymbol{X}_2 &= \sin\theta\frac{\partial}{\partial r} + \frac{\cos\theta}{r}\frac{\partial}{\partial \theta},\label{eq:suppy} \\
\boldsymbol{R}   &=\frac{\partial}{\partial \theta} \nonumber,
\end{align}
and let
\begin{equation}
\left(A^{i j}\right) = \left( \begin{array}{cc}
\beta_1 & \beta_3\\
\beta_3 & \beta_2
\end{array} \right),\quad
B^i = (\beta_4,\beta_5),\quad C = \beta_6,
\label{eq:parameters}
\end{equation}
then we find upon substituting (\ref{eq:suppy}) and (\ref{eq:parameters}) into (\ref{eq:killing-as-sum}) and taking the symmetric tensor products then
\[
\boldsymbol{K}_{\tilde{p}} = A^{ij} \boldsymbol{X}_i \odot \boldsymbol{X}_j + B^i \boldsymbol{X}_i \odot \boldsymbol{R} + C \boldsymbol{R} \odot \boldsymbol{R}
\]
and we derive that the components of the general Killing tensor  $\boldsymbol{K}_p \in \k^2(\mathbb{E}^2)$ defined with respect to polar coordinates are given by
\begin{align}
K^{11}_{\tilde{p}} &= 
 2 \beta_3 \cos{\theta } \sin{\theta}+\beta_1\cos^2{\theta} +\beta_2\sin^2{\theta}, \nonumber \\
K^{1 2}_{\tilde{p}} &=  2 \beta_4\cos{\theta } +2 \beta_5\sin{\theta } + \frac{2 \beta_3  \cos{2 \theta}}{r}-\frac{2 \beta_1\cos{\theta } \sin{\theta } }{r}+\frac{2\beta_2 \cos{\theta } \sin{\theta } }{r}, \label{eq:polarKT} \\
K^{2 2}_{\tilde{p}} &=  \beta_6-\frac{2 \beta_3  \cos{\theta } \sin{\theta }}{r^2}+\frac{\beta_1\sin^2{\theta }}{r^2}+\frac{\beta_2\cos^2{\theta }}{r^2}  -\frac{2 \beta_4\sin{\theta }}{r}+\frac{2 \beta_5\cos{\theta }}{r}.  \nonumber
\end{align} 
Substituting $\boldsymbol{K}_{\tilde{p}}$ and the potential (\ref{ttwpotty}) into the compatibility condition,
\begin{equation}
\textnormal{d}\left(\boldsymbol{\hat{K}}_{\tilde{p}}\textnormal{d}V\right) = 0,
\label{compatttw}
\end{equation}
we arrive at a complicated trigonometric equation in terms of the arbitrary parameters $k, \lambda_1, \lambda_2, \beta_i, i = 1, \dots 6$ (see eq.~(\ref{bad}) in Appendix \ref{bad-equation}).  To solve this equation, we expand it over the following set $M$ of trigonometric functions
\begin{align}
M =& \{1, \cos(\theta), \sin(\theta), \cos(2\theta), \sin(2\theta), \nonumber \\
  & \quad \cos((1 + k)\theta), \sin((1 + k)\theta), \cos((1 - k)\theta), \sin((1 - k)\theta), \nonumber \\
	& \quad \cos((1 + 2k)\theta), \sin((1 + 2k)\theta), \cos((1 - 2k)\theta), \sin((1 - 2k)\theta), \nonumber \\
	& \quad \cos((1 + 3k)\theta), \sin((1 + 3k)\theta), \cos((1 + 3k)\theta), \sin((1 + 3k)\theta), \nonumber \\
  & \quad \cos(2(1 + k)\theta), \sin(2(1 + k)\theta), \cos(2(1 - k)\theta), \sin(2(1 - k)\theta), \nonumber \\
	& \quad \cos(2(1 + 2k)\theta), \sin(2(1 + 2k)\theta), \cos(2(1 - 2k)\theta), \sin(2(1 - 2k)\theta), \nonumber \\
	& \quad \cos(2(1 + 3k)\theta), \sin(2(1 + 3k)\theta), \cos(2(1 + 3k)\theta), \sin(2(1 + 3k)\theta)\},
\label{eq:trig-basis}
\end{align}
which is the set obtained by expanding the trigonometric functions in (\ref{bad}), ignoring any trig identities.
The set $M$ is linearly dependent when the arguments of the trigonometric functions coincide (see Appendix \ref{section:trig-proof}), a property which depends on the values assumed by the parameter $k$.  By plotting the arguments as functions over $k$ and determining the intersection points of these functions, we thus determine that the set (\ref{eq:trig-basis}) is linearly \emph{dependent} for the following values of $k$: 
\begin{equation}
\begin{split}
k = &\pm 2,\pm \frac{3}{2},\pm 1,\pm \frac{1}{2},\pm \frac{1}{4},\pm \frac{1}{6},\pm \frac{1}{8},\pm \frac{1}{10},\pm \frac{1}{12},\pm \frac{1}{14},\pm \frac{1}{16},\pm \frac{3}{4},\pm \frac{2}{3},\\
&\pm\frac{3}{8},\pm \frac{1}{3},\pm \frac{3}{10},\pm \frac{2}{7},\pm \frac{3}{14},\pm \frac{1}{5},\pm \frac{3}{16},\pm \frac{2}{5},\pm \frac{1}{7}.
\end{split}
\label{eq:specialk-values}
\end{equation}
In the case when $k$ does not equal one of the above values, the set $M$ is linearly \emph{independent} and so we can set each of the coefficients to zero in the trigonometric equations that we obtain upon the substituting  (\ref{eq:killing-as-sum}) and (\ref{ttwpotty}) into (\ref{compatttw}) and expanding over the set $M$.  This argument leads to an overdetermined system of polynomial equations which we can then solve for the variables $k,\beta_i,$ and $\lambda_j, i=1,\dots,6, j=1,2$.  Before completing this computation we shall reduce the number of polynomial equations with the following argument. 

Since the TTW potential (\ref{ttwpotty}) is separable in polar coordinates then if it were to admit another quadratic first integral of motion it would have been defined by a Killing two-tensor which is an element of the 5-dimensional vector subspace
\[
\boldsymbol{K}_s = A^{i j} \boldsymbol{X}_i \odot \boldsymbol{X}_j + B^i \boldsymbol{X}_i \odot \boldsymbol{R}, \quad i,j=1,2,
\]
which is the vector space (\ref{eq:killing-as-sum}) with the ``polar" term factored out.  Substituting this Killing tensor $\boldsymbol{K}_s$ into the compatibility condition (\ref{compatttw}) for the TTW potential (\ref{ttwpotty}), we arrive a system of equations for which it is possible to expand over the basis $L = \{1,\cos\theta,\sin\theta,\cos 2\theta,\sin 2\theta\}$.  This results in the following necessary condition on the parameters of the Killing tensor $\boldsymbol{K}_s$:
\[
\beta_4^2 + \beta_5^2 = 0.
\]
Since all parameters are assumed to be real, the above holds provided $\beta_4 = \beta_5 = 0$.  Therefore we conclude that if the potential (\ref{ttwpotty}) is separable in another orthogonal coordinate system, other than polar, it can only be the Cartesian system of coordinates.  We can now repeat the same argument outlined above with a simpler Killing tensor.  Namely, the Killing two-tensor that defines Cartesian coordinates in a general position \cite{mclenaghan2002group}, 
\begin{equation}
K_{C}^{i j}= 
\left(
\begin{array}{cc}
 \cos^2(\theta -\phi ) & -\frac{1}{2} r \sin(2 (\theta -\phi )) \\
 -\frac{1}{2} r \sin(2 (\theta -\phi )) & r^2 \sin^2(\theta -\phi)
\end{array}
\right),
\label{eq:killing-cartesian}
\end{equation}
where $\phi$ denotes the angle by which the Cartesian coordinate system is rotated (recall the Cartesian web in Figure \ref{non-canon-webs}).  Substituting (\ref{eq:killing-cartesian}) and (\ref{ttwpotty}) into the compatibility condition (\ref{compatttw}), we arrive at the following equation
\begin{dmath}
\csc^2k\theta \left(6 k \cos2 (\theta -\phi ) \cot k \theta +\left(4+2 k^2-3 k^2 \csc^2 k \theta\right) \sin 2(\theta -\phi )\right) \lambda _2 +\sec^2 k \theta \left(-6 k \cos 2 (\theta -\phi ) \tan k \theta +\left(4+2 k^2-3 k^2 \sec^2 k \theta \right) \sin 2 (\theta -\phi )\right)\lambda _1  = 0.
\label{eq:compat-cond-cartKT}
\end{dmath}
We can then expand this over a subset of $N \subset M$ given by 
\begin{align}
N =& \{1, \cos(2\theta), \sin(2\theta), \nonumber \\
  &\quad \cos(2(1 + k)\theta), \sin(2(1 + k)\theta), \cos(2(1 - k)\theta), \sin(2(1 - k)\theta), \nonumber \\
	&\quad \cos(2(1 + 2k)\theta), \sin(2(1 + 2k)\theta), \cos(2(1 - 2k)\theta), \sin(2(1 - 2k)\theta), \nonumber \\
	&\quad \cos(2(1 + 3k)\theta), \sin(2(1 + 3k)\theta), \cos(2(1 + 3k)\theta), \sin(2(1 + 3k)\theta)\}
\label{eq:trig-basis-subset}
\end{align}  
which we determine to be linearly dependent for the following values of $k$
\begin{equation}
k = \pm 1,\pm 2,\pm 2/3,\pm 1/2,\pm 2/5,
\label{eq:reduced-k}
\end{equation}
which are indeed a subset of the values obtained in (\ref{eq:specialk-values}).  Expanding the equation (\ref{eq:compat-cond-cartKT}) over the set $N$, and assuming it to be linearly independent (i.e. excluding the values of $k$ given by (\ref{eq:reduced-k})), we arrive at the following (reduced) system of polynomial equations in the variables $\phi, k, \lambda_1,\lambda_2$:
\begin{align*}
  (-2+k) (-1+k) \left(\lambda _1-\lambda _2\right)\sin2\phi  &= 0,\\
 (-2+k (-15+23 k)) \left(\lambda _1-\lambda _2\right)\sin2\phi  &= 0,\\
 (-2+k (15+23 k)) \left(\lambda _1-\lambda _2\right) \sin2\phi &= 0,\\
 (1+k) (2+k) \left(\lambda _1-\lambda _2\right)\sin2\phi  &= 0,\\
 (-1+k) (-1+2 k)  \left(\lambda _1+\lambda _2\right) \sin2\phi &= 0,\\
 (1+k) (1+2 k) \left(\lambda _1+\lambda _2\right) \sin2\phi &= 0,\\
 (-1+2k)(1+2k)  \left(\lambda _1+\lambda _2\right)\sin2\phi  &= 0,
\end{align*}
\begin{align*}
 (-2+k) (-1+k) \left(\lambda _1-\lambda _2\right) \cos2\phi &= 0,\\
 (-2+k (-15+23 k))\left(\lambda _1-\lambda _2\right)\cos2\phi &= 0,\\
 (-2+k (15+23 k)) \left(\lambda _1-\lambda _2\right)\cos2\phi &= 0,\\
 (1+k) (2+k) \left(\lambda _1-\lambda _2\right)\cos2\phi &= 0,\\
 (-1+k) (-1+2 k)  \left(\lambda _1+\lambda _2\right)\cos2\phi &= 0,\\
 (1+k) (1+2 k)  \left(\lambda _1+\lambda _2\right)\cos2\phi &= 0,\\
 (-1+2k)(1+2k)  \left(\lambda _1+\lambda _2\right)\cos2\phi &= 0.
 \end{align*}
There is no non-trivial solution for which the above system is satisfied.  Considering the set of linearly \emph{dependent} functions, we substitute into equation (\ref{eq:compat-cond-cartKT}) with the special values of $k$ given in (\ref{eq:reduced-k}) to yield the following possible cases
\begin{enumerate}[label=(\arabic*)]
	\item $k = \pm 1$ \\
	Imposes $\phi = 0$. 
	\item $k = \pm 2$ \\
	Imposes $\{\lambda_1 = 0,\phi= n\pi/2\}$ or, $\{\lambda_2 =0, \phi = n\pi/4\}, \quad n \in \N$.
	\item $k = \pm 2/3, \pm 1/2, \pm 2/5$,\\
	 Imposes $\lambda_1 = \lambda_2 = 0$.
\end{enumerate}
Thus, we conclude that $k = \pm 1$ are the only values of $k$ for which the full TTW potential (\ref{ttwpotty}) admits two quadratic first integrals of motion.  For these values of $k$, the TTW system reduces to the Smorodinsky-Winternitz system discussed in Section {\ref{section:sw}}.  Hence, we have proven the following
\begin{prop}
The general TTW potential given by (\ref{ttwpotty}) is a multi-separable superintegrable potential only when $k = \pm 1$.
\end{prop}
The case $k = \pm 2$ yields the well-known Calogero system on the line  (see e.g.~\cite{calogero1969,chanu2008super}).  
The conclusion of our above argument is that only when either $\lambda_1$ or $\lambda_2$ vanishes does this potential admit separation in general Cartesian coordinates according to a non-canonical Killing tensor $\boldsymbol{K}_C$ whose components are given by (\ref{eq:killing-cartesian}).
In particular, the compatibility condition $\textnormal{d}(\boldsymbol{\hat{K}}_C\textnormal{d}V) = 0$ will be satisfied for the potential 
\[
V(r, \theta) = -\omega^2 r^2 + \frac{\lambda_1}{\cos 2\theta} + \frac{\lambda_2}{\sin 2\theta}
\]
 when either $\lambda_1 = 0$ and $\phi = n \pi/2$, or $\lambda_2 = 0$, and $\phi = n \pi/4$. 

Our results demonstrate that this approach to defining an infinite family of superintegrable systems appears to destroy the geometry of the associated orbit space $({\k}^2(\mathbb{E}^2) \times {\k}^2(\mathbb{E}^2))/SE(2)$.  In other words, it seems unlikely for systems defined in this manner to admit two quadratic first integrals of motion for multiple values of the introduced parameter.  In 2010 \cite{post2010}, this method was utilised by Post and Winternitz to construct another infinite family of superintegrable systems dependent on a parameter $k$, namely a potential that appeared as a second-order perturbation of the \emph{Kepler-Coulomb potential}.  They then established an equivalence between this new system and the TTW system through a St\"{a}ckel transform, which maps Hamiltonians to Hamiltonians, and hence developed a direct relation between the integrals of motion and trajectories of each system.  Indeed, it would be of interest to pursue this equivalence through the invariant theory of Killing tensors which could possibly facilitate the search for these \emph{St\"{a}ckel equivalent} systems on a manifold.  



\chapter{Conclusions and Outlook}
\label{conc}

\ifpdf
\graphicspath{{Chapter4/Chapter4Figs/PNG/}{Chapter4/Chapter4Figs/PDF/}{Chapter4/Chapter4Figs/}}
\else
\graphicspath{{Chapter4/Chapter4Figs/EPS/}{Chapter4/Chapter4Figs/}}
\fi

We have invariantly characterised the Smorodinsky-Winternitz superintegrable potential modulo the action of the isometry group $SE(2)$ in terms of joint invariants of the associated Killing tensors by taking them as points in the product space $\k^2(\mathbb{E}^2) \times \k^2(\mathbb{E}^2)$.  We can extend this result to the problem of classification of arbitrary superintegrable potentials defined in spaces of constant curvature by classifying first any associated Killing tensors that define first integrals of motion for the class of superintegrable systems in question (e.g., superintegrable systems that admit only quadratic first integrals of motion).  These Killing tensors are identified as elements of the corresponding product spaces, whose factors are the corresponding vector spaces of Killing tensors representing the associated orbits in the induced orbit space under the action of the isometry group in such product spaces.

Additionally, we have proven that the general TTW potential is a superintegrable system with two  functionally independent quadratic first integrals of motion only for $k = \pm 1$, that is when it reduces to the Smorodinsky-Winternitz potential.  If we allow one of the arbitrary parameters in the TTW potential to vanish, then a second multi-separable system is obtained when $k=\pm 2$.  Therefore we have established that the general superintegrability of the TTW system for other $k$ values must come from a (non-trivial) higher order (i.e.~not quadratic) first integral.    

The results contained within this thesis motivate the complete classification of Hamiltonian systems with two degrees of freedom, which are superintegrable in the sense of being multi-separable.  The introduction of an additional arbitrary parameter into a superintegrable potential appears to be very destructive, in the sense that it destroys the property of multi-separability for all but one choice of the arbitrary parameter.  This hints towards a characterisation on the Euclidean plane for all superintegrable systems obtained as perturbations of known superintegrable systems, and further, the classification of all multi-separable superintegrable systems defined on the Euclidean plane.  

Furthermore, we can extend our results to studying superintegrable systems defined on different two-dimensional spaces, such as on the Minkowski plane or the $2D$ sphere.  In the case of the former, the study becomes complicated by the fact that there are nine orthogonal coordinate webs and instead of a discrete set of singular points, there are singular \emph{regions} admitted by complex-valued eigenvalues of the associated characteristic Killing tensor.  In the case of the $2D$ sphere, the study is simplified as a result of both the additional symmetry of the manifold and its compactness.  Indeed, there are instead only two orthogonal coordinate webs defined on the sphere, and the set of singular points remains discrete.  It is of immediate interest to the author to carry out the classification of all multi-separable systems defined on spaces of constant curvature with two degrees of freedom. 




\bibliographystyle{amsplain}
\renewcommand{\bibname}{Bibliography} 
\bibliography{References/biblio} 

\appendix
\chapter{Compact Formulae}
\section{The Action of the Isometry Group}
\label{group-action-stuff}
We provide compact formulas for the equations appearing in Section \ref{section:ITKT}.  Of course, these equations have been produced in the literature for higher $(n > 3)$ dimensional spaces (see, e.g~\cite{horwood2008thesis}), though we rewrite them specifically for $\mathbb{E}^2$ at the indulgence of the author. 

This begins with the compact formulae for the group action $SE(2) \circlearrowright \mathcal{K}^2(\mathbb{E}^2)$, for the general Killing tensor given by 
\[
\boldsymbol{K} = A^{i j} \boldsymbol{X}_i \odot \boldsymbol{X}_j + B^i \boldsymbol{X}_i \odot \boldsymbol{R} + C \boldsymbol{R} \odot \boldsymbol{R},
\]
 where we can define the parameters by
 \[
\boldsymbol{A} = \left(
\begin{array}{cc}
	\beta_1 & \beta_3\\
	\beta_3 & \beta_2  
\end{array}
		  \right), 
\quad 
\boldsymbol{B} = \left(
\begin{array}{c}
	\beta_4 \\
	-\beta_5
\end{array}
	  \right),
\quad C = \beta_6.
\]
The transformation laws for the basis vectors are given by
\begin{equation}
\boldsymbol{X}_i = \Lambda^j{}_i \boldsymbol{\tilde{X}}_j, \quad \boldsymbol{R} = \mu^j \boldsymbol{\tilde{X}}_j + \boldsymbol{\tilde{R}},
\label{eq:basis-tran}
\end{equation}
where $\mu^j = \epsilon^k{}_i \delta^i \Lambda^j{}_k$, 
where $\epsilon^k{}_{i}$ is the two-dimensional Levi-Civita symbol with the usual orientation, $\epsilon^1{}_2 = +1, \epsilon^2{}_1 = -1$, $\delta^i = (p_1, p_2)$ and the $\Lambda^j{}_k$ are elements of the special orthogonal group, $SO(2)$, i.e.
\[
\boldsymbol{\Lambda} = \left(
\begin{array}{cc}
	\cos\theta & -\sin\theta\\
	\sin\theta & \cos\theta  
\end{array}
		  \right),
\]
we find that the action of $SE(2) \circlearrowright \k^2(\mathbb{E}^2)$ induced by the action of $SE(2) \circlearrowright \mathbb{E}^2$, i.e.~$q^i \to \Lambda^i{}_j \tilde{q}^j + \delta^i$, can be written in a nice, compact way:
\begin{align}
\tilde{A}^{i j} &= \Lambda_k{}^i\Lambda^j{}_p A^{k p} + 2 B^{k}\Lambda_k{}^{(i}\mu^{j)} + \beta_6 \mu^i\mu^j, \nonumber\\
\tilde{B}^j &= \Lambda^j{}_\ell B^{\ell} + \beta_6 \mu^{j}, \label{A2}\\
\tilde{C} & = \beta_6 \nonumber,
\end{align}
where the parentheses around the indices denote symmetrisation as defined in Section \ref{conventions}.  This is equivalent to equation  (\ref{eq:transformed-KT-parameters}).
\section{Kogan's Recursive Approach}
\label{noway}
We can also provide compact formulae for those appearing in the derivation of group invariants using Kogan's recursive approach (see Section \ref{section:group-action}).  Recall that we first compute the action of the subgroup of translations on the vector space of Killing tensors $\k^2(\mathbb{E}^2).$ This is induced by the action of the translational group on $\mathbb{E}^2$, namely $q^i \to \tilde{q}^i + \delta^i$, which yields the compact form of eq.~(\ref{eq:trans-group-action}) as
\begin{align}
\tilde{A}^{i j} &= A^{ij} + 2B^{(i}\epsilon^{j)}{}_m\delta^m + C\epsilon^i{}_m\epsilon^j{}_\ell\delta^m\delta^\ell, \nonumber\\
\tilde{B}^i &= B^j + C \epsilon^i{}_{\ell} \delta^{\ell}, \label{A3}\\
\tilde{C} &= \beta_6. \nonumber
\end{align}
The equivalent compact form of the normalisation equations (\ref{eq:normal-equations}) are then given by
\[
\delta^i = \epsilon^i{}_j B^j/C.
\]
Finally, we needed to compute the action of the subgroup of rotations $SO(2) \circlearrowright \k^2(\mathbb{E}^2)$.  This action is obtained in a straightforward manner as the induced action of $SO(2) \circlearrowright \mathbb{E}^2$, namely $q^i \to \Lambda^i{}_j \tilde{q}^j$ to yield the following compact form of eq.~\ref{eq:rot-group-action} as
\begin{align}
\tilde{A}^{ij} &= \Lambda_k{}^i\Lambda^j{}_p A^{kp}, \nonumber\\
\tilde{B}^i &=  \Lambda_{\ell}{}^i B^{\ell}, \label{A4}\\
\tilde{C} & = \beta_6 \nonumber.
\end{align}
This concludes the compact representation of the Section \ref{section:group-action} equations. Though, we remark that the development can indeed be continued in order to derive the invariants of the vector space of valence two Killing tensors under the action of the isometry group, $SE(2)$.  Specifically, by substituting for $\delta^i$ into (\ref{A3}) and using the transformation law given by (\ref{A4}).  Provided that $C \neq 0$, then a basis for the invariants are admitted by the trace and determinant of $\tilde{A}^{ij}$.  The case when $C=0$ must then be treated separately (see \cite{horwood2008thesis} for explicit details).



\chapter{The Killing Tensor Equation in Local Coordinates}
\label{kt-solved}
 With respect to local (Cartesian) coordinates, i.e.~$g_{ij} = \delta_{ij}$, the Killing tensor equation in Section \ref{conventions}, namely 
 \[
 [\g,\boldsymbol{K}] = 0,
 \]
 is equivalent to\footnote{This equivalence holds true \emph{only} in $\mathbb{E}^2$, since contravariant and covariant components with respect to Cartesian coordinates are equivalent.} 
\[
K_{\left(i j; k\right)} = 0, \quad i,j,k = 1,2,
\]
where the covariant derivative  (see Section \ref{conventions}) is in fact a partial derivative since all the connection coefficients vanish in a space of constant curvature.  Using the fact that $K_{i j}$ is symmetric, i.e.~$K_{i j} = K_{j i}$, and applying the symmetrization rule, this equation becomes
\[
K_{i j, k} + K_{j k, i} + K_{k i, j} = 0, 
\]
which yields a system of PDE's in $K_{i j}$, for  $i, j = 1, 2$.  Namely, we have
\begin{center}
\begin{enumerate}[label=(\arabic*)]
	\item $K_{11,1} = 0$,
	\item $K_{22,2} = 0$,
	\item $K_{12,1} = -\tfrac{1}{2}K_{11,2}$,
	\item $K_{12,2} = -\tfrac{1}{2}K_{22,1}$ 
\end{enumerate}
\end{center}
We will take the Killing tensor components as functions in the Cartesian coordinates $(q^1,q^2) = (x,y)$, so that, for instance $K_{11,1} = \frac{\partial K_{11}(x,y)}{\partial x}$.  Integrating $(1)$ with respect to $x$ and $(2)$ with respect to $y$ then yields that $K_{11}(x,y) = F(y)$ and $K_{22}(x,y) = G(x)$, respectively, where $F$ and $G$ are arbitrary functions.  Taking a partial derivative of $(3)$ with respect to $x$, we find that
\begin{align*}
\frac{\partial^2 K_{12}}{\partial x^2} &= 0, \\
	&\Rightarrow K_{12,1}(x,y) = H(y) = -\frac{1}{2} \frac{\partial F}{\partial y}.
\end{align*}
Integrating $ K_{12,1}$ with respect to $y$ then yields that
\begin{equation}
F(y) = -2 \int{H(y) dy} + \beta_1.
\label{eq:fy}
\end{equation}
Analogously, if we take a partial derivative of $(4)$ with respect to $y$ , we find that
\begin{align*}
\frac{\partial^2 K_{12}}{\partial y^2} &= 0 \\
	&\Rightarrow K_{12,2} = M(x) = -\frac{1}{2}\frac{\partial G}{\partial x},
\end{align*}
so 
\begin{equation}
G(x) = -2\int{M(x) dx} + \beta_2.
\label{eq:gx}
\end{equation}
If we now use the fact that the partial derivatives commute, so $\frac{\partial K_{12}}{\partial x\partial y} = \frac{\partial K_{12}}{\partial y\partial x},$ then since $K_{12,1}(x,y) = H(y)$ and $K_{12,2} = M(x)$, we find that
\[
\frac{\partial H}{\partial y} = \frac{\partial M}{\partial x} = -\beta_6,
\] 
where $\beta_6$ is a separation constant.  Thus, we can separate these two PDEs to obtain that 
\[
H(y) = -\beta_6 y -\beta_4, \quad M(x) = -\beta_6 - \beta_5.
\]
Substituting these into (\ref{eq:fy}) and (\ref{eq:gx}), respectively, we find after integrating that
\[
F(y) = \beta_6 y^2 + 2\beta_4 y + \beta_1, \quad G(x) = \beta_6 x^2 + 2\beta_5 x+ \beta_2,
\]
so that 
\[
K_{11} = \beta_1 +  2\beta_4 y +\beta_6 y^2  , \quad K_{22} = \beta_2 + 2\beta_5 x + \beta_6 x^2.
\]
Lastly, integrating $K_{12,1}(x,y) = H(y)$ with respect to $x$, and using $K_{12,2} = M(x)$, we will find that
\[
K_{12} = \beta_3 - \beta_4 x -\beta_5 y + \beta_6 x y.
\]
This gives the general solution to the Killing tensor equation defined on the Euclidean plane.  Indeed, we see that it depends on 6 arbitrary parameters, $\beta_i, i = 1, \dots 6$.
\chapter{A Complicated Trigonometric Equation}
\label{bad-equation}
Imposing the compatibility condition $\textnormal{d}\left(\boldsymbol{\hat{K}}\textnormal{d}V\right) = 0$ in the case of a general Killing tensor $\boldsymbol{\hat{K}}$ in six arbitrary parameters and the potential of the TTW system given in Section \ref{section:ttw}
as eq. (\ref{ttwpotty}), we arrive at the following trigonometric equation
\begin{dmath}
r^5 \omega ^2 \left(r \sin 2 \theta  \left(\beta _1-\beta _2\right)-2 r \cos 2 \theta  \beta _3+2 \cos \theta  \beta _4+2 \sin \theta  \beta _5\right)
+3 k^2 r^4 \sec ^4 k \theta \left(\sin 2 \theta  \left(-\beta _1+\beta _2\right)+2 \cos 2 \theta  \beta _3+2 r \cos \theta  \beta _4+2 r \sin \theta  \beta _5\right) \lambda _1
+\sec ^2 k \theta \left(\left(4+r^2+2 k^2 r^4\right) \sin 2 \theta  \left(\beta _1-\beta _2\right)-2 \left(4+r^2+2 k^2 r^4\right) \cos 2 \theta  \beta _3
-2 r \left(3+2 k^2 r^4\right) \cos \theta  \beta _4-2 r \left(3+2 k^2 r^4\right) \sin \theta  \beta _5\right) \lambda _1
+2 k r^2 \left(1+r^2\right) \cot k \theta  \csc ^2 k \theta \left(\cos 2 \theta  \left(\beta _1-\beta _2\right)+2 \sin 2 \theta  \beta _3+r \sin \theta  \beta _4-r \cos \theta  \beta _5\right) \lambda _2
+3 k^2 r^4 \csc ^4 k \theta \left(\sin 2 \theta  \left(-\beta _1+\beta _2\right)+2 \cos 2 \theta  \beta _3+2 r \cos \theta  \beta _4+2 r \sin \theta  \beta _5\right) \lambda _2
+\frac{2}{-1+\cos 2 k \theta } \left(-\left(4+r^2+2 k^2 r^4\right) \sin 2 \theta  \left(\beta _1-\beta _2\right)+2 \left(4+r^2+2 k^2 r^4\right) \cos 2 \theta  \beta _3
+2 r \left(3+2 k^2 r^4\right) \cos \theta  \beta _4+2 r \left(3+2 k^2 r^4\right) \sin \theta  \beta _5\right) \lambda _2
-2 k r^2 \left(1+r^2\right) \sec ^2 k \theta \left(\cos 2 \theta  \left(\beta _1-\beta _2\right)+2 \sin 2 \theta  \beta _3+r \sin \theta  \beta _4-r \cos \theta  \beta _5\right) \lambda _1 \tan k \theta
 = 0,
 \label{bad}
\end{dmath}
which must be solved for $\omega, \lambda_1,\lambda_2,k,\beta_i, i=1,\dots, 5$.
\chapter{Linear Independence of a Set of Functions}
\label{section:trig-proof}
We provide a proof to the claim in Section \ref{section:ttw} that the set $M$ (\ref{eq:trig-basis}) is linearly independent when the arguments of the trigonometric functions do not coincide.  Indeed, it suffices to prove the following.
\begin{prop}[cf. Yuan \cite{yuan}]

The set of trigonometric functions defined by 
\[
\{1, \sin(r\theta), \cos(r\theta) \},
\]
for positive $r \in \R$ is linearly independent over $\R$. 
\end{prop}
Though it is straightforward to compute the Wronskian of the set we provide an alternative proof which allows us to extend the result.  
\begin{proof}
Suppose that we have a non-trivial dependency relation with a minimal number of terms in $\sin(r\theta)$ and $\cos(r\theta)$, 
\[
\sum{ c_r \sin(r\theta) + d_r \cos(r\theta)} = 0.
\]
Consider the largest possible real $r_0$ for which $d_{r_0}$ is non-vanishing.  If we take a large, even number of derivatives eventually the $d_{r_0}$ coefficient will dominate all other coefficients in the other cosine terms.  Suppose we take a considerable amount of an even number of derivatives so that this is the case.  If we then substitute with $\theta = 0$, we eliminate all sine terms and are left with a large number which cannot be equal to zero.  Therefore, it must be the case that  no cosine terms appear in the dependency relation. 

Similarly, we can consider the largest possible real $r_1$ for which $c_{r_1}$ is non-vanishing.  We will then take a substantial amount of an odd number of derivatives until the coefficient of the $\cos (r_1\theta)$ term dominates all other coefficients in  the other cosine terms (where these cosines come from differentiating sines).  Substituting with $\theta =0 $, we again find ourselves with a large number which cannot be equal to zero.  Thus, it follows that there can be no sine terms in the dependency relation.

What remains is that 1 is the only function which can appear in the non-trivial dependency relation, which is clearly impossible.  Thus, we conclude that no such dependency relation can exist and so the set must be linearly independent.
\end{proof}
Prompted by the discussion in \cite{yuan}, we arrive at a generalisation of the above proposition.  Considering that 
\[
\cos(r \theta) = \frac{e^{i r\theta} + e^{- i r \theta}}{2}, \qquad \sin(r\theta) = \frac{e^{i r \theta} - e^{-i r\theta}}{2 i },
\]
it suffices to prove that the set of functions $\{ e^{i r\theta} \}$  for $r \in \R$ are linearly independent over $\R$.  Further, we can show that the set of functions $\{e^{z \theta} \}$ for a $z \in \C$ are linearly independent over $\C$.
\begin{prop}
The set $\{e^{ z \theta}\}$ for $z \in \C$ is linearly independent over $\C$.
\end{prop}
\begin{proof}
Suppose we have a non-trivial dependency relation with a minimal number of terms given by
\[
\sum{a_z e^{z \theta}} = 0.
\]
Differentiating this equation with respect to $\theta$ then yields that 
\[
\sum{z a_z e^{z \theta}} = 0,
\]
as well.  By our assumption, there must exist a non-zero $z_0 \in \C$ with $a_{z_0} \neq 0$, and so subtracting this from the previous equation we find that
\[
\sum{(z - z_0) a_z e^{z t}} = 0
\]
holds as well.  But this is a new non-trivial dependency relation \emph{in fewer terms}, which contradicts our original assumption.  Therefore, the set of functions $\{e^{z \theta}\}$ is indeed linearly independent over $\C$.  
\end{proof}
We comment that the linear independence of the above functions follows from a more general property regarding eigenvectors of a linear operator, where those with distinct eigenvalues are linearly independent.  Considering that the exponential, sine, and cosine functions are all essentially eigenfunctions of differentiation\footnote{A non-zero function $f$ is called an eigenfunction of a linear operator $A$ iff it satisfies $A f = \lambda f,$ for some scalar $\lambda$ called the eigenvalue.}, one can prove the result holds in this more general context.

\end{document}